\documentclass[a4paper]{book} 
\usepackage[utf8]{inputenc}
\usepackage{lmodern}
\usepackage[a4paper, top=3cm,bottom=3cm,left=3.5cm,right=3.5cm]{geometry}
\usepackage[greek,british]{babel}
\usepackage{csquotes}

\title{Indeterministic finite-precision physics and intuitionistic mathematics}
\author{Tein van der Lugt}
\date{31st July 2020}

\usepackage[style=numeric-comp,
            sortcites=true,
            giveninits=true, 
            useprefix=true,
            sorting=nyt]{biblatex}
\DeclareFieldFormat{titlecase}{#1} 
\DeclareFieldFormat{postnote}{#1} 
\DeclareSortingNamekeyTemplate{
  \keypart{\namepart{family}}
  \keypart{\namepart{prefix}} 
  \keypart{\namepart{given}}
  \keypart{\namepart{suffix}}
}
\DeclareNameAlias{author}{family-given} 
\renewbibmacro{begentry}{\midsentence}  
\usepackage{tikz, xcolor, wrapfig}

\usepackage[pdfusetitle]{hyperref}
\usepackage{bookmark} 
\bookmarksetup{numbered} 
\usepackage[shortlabels]{enumitem}
\setlist{noitemsep}

\usepackage{emptypage}
\usepackage[Sonny]{fncychap} 
\usepackage{lettrine}

\counterwithout{footnote}{chapter}

\usepackage{marginnote}
\reversemarginpar
\newif\iftodo
\newif\iftodow
\newif\ifnote
\newif\ifcaution
\newif\ifnieuw

\newcommand{\nieuw}[1]{\ifnieuw{\marginnote{\normalsize\normalfont\color{green}\textsf{NEW}}{\color{green}{#1}}}\else#1\fi}
\todofalse  
\todowfalse
\nieuwfalse
\cautionfalse

\usepackage{amsmath, amsthm, amssymb, bm, commath}
\newtheorem{theorem}{Theorem}[chapter]
\newtheorem{corollary}[theorem]{Corollary}
\newtheorem{proposition}[theorem]{Proposition}

\theoremstyle{definition}
\newtheorem{definition}[theorem]{Definition}
\newtheorem{notation}[theorem]{Notation}
\newtheorem{convention}[theorem]{Convention}
\newtheorem{example}[theorem]{Example}
\newtheorem{assumption}[theorem]{Assumption}
\newcommand{\defn}[1]{\emph{\textbf{#1}}}

\renewcommand{\a}{\alpha}
\renewcommand{\b}{\beta}
\renewcommand{\d}{\delta}
\renewcommand{\phi}{\varphi}
\newcommand{\e}{\varepsilon}
\newcommand{\g}{\gamma}
\newcommand{\s}{\sigma}

\newcommand{\apart}{\mathrel{\#}}

\def\N{\mathbb{N}}
\def\Z{\mathbb{Z}}
\def\Q{\mathbb{Q}}
\def\R{\mathbb{R}}
\def\S{\mathbb{S}}
\def\NN{\mathcal{N}}
\def\D{\mathcal{D}}
\def\M{\mathcal{M}}
\def\E{\mathcal{E}}
\def\Sing{\mathcal{S}}
\def\Som{\Sigma^\omega}
\def\Sst{\Sigma^*}
\DeclareMathOperator{\dom}{dom}
\DeclareMathOperator{\ran}{ran}
\DeclareMathOperator{\diam}{diam}
\DeclareMathOperator{\vol}{vol}
\newcommand{\inv}[1]{\frac{1}{#1}}
\newcommand{\ddt}[1]{\frac{\dif #1}{\dif t}}
\let\o\overline

\begin{document}
    \pagenumbering{gobble}

    \begin{titlepage}
    \newgeometry{top=4cm,bottom=4cm,left=2cm,right=2cm,marginparwidth=1.75cm}
    
    \center 
    \large

    \includegraphics[width=.4\textwidth]{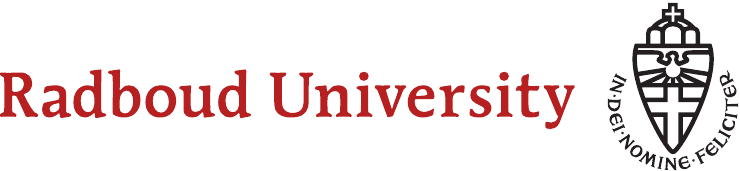}\\[2cm]

    \makeatletter
    {\sffamily\huge Indeterministic finite-precision physics\\[1mm]and intuitionistic mathematics}\\[3cm] 
    \makeatother
    
    {\Large Tein \textsc{van der Lugt}}\vfill\vfill
    
    \textsc{Bachelor's thesis}\\[0.2cm]
    Radboud Honours Academy\\[2cm]

    \emph{Supervised by}\\[0.3cm]
    Prof.\ N.P.\ \textsc{Landsman}\\[2cm]
    
    Department of Mathematics\\
    Radboud University Nijmegen\\[1cm]
    
    \makeatletter
    {\large \@date}
    \makeatother


    \vfill
\end{titlepage}
    
    \hspace{0pt}\vfill 
    \noindent\textsc{Abstract.}\ \ \ In recent publications in physics and mathematics, concerns have been raised about the use of real numbers to describe quantities in physics, and in particular about the usual assumption that physical quantities are infinitely precise. In this thesis, we discuss some motivations for dropping this assumption, which we believe partly arises from the usual point-based approach to the mathematical continuum. We focus on the case of classical mechanics specifically, but the ideas could be extended to other theories as well. We analyse the alternative theory of classical mechanics presented by Gisin and Del Santo \cite{dSG19}, which suggests that physical quantities can equivalently be thought of as being only determined up to finite precision at each point in time, and that doing so naturally leads to indeterminism. Next, we investigate whether we can use intuitionistic mathematics to mathematically express the idea of finite precision of quantities, arriving at the cautious conclusion that, as far as we can see, such attempts are thwarted by conceptual contradictions. Finally, we outline another approach to formalising finite-precision quantities in classical mechanics, which is inspired by the intuitionistic approach to the continuum but uses classical mathematics.
    \hspace{0pt}\vfill
    
    \frontmatter
    
    \phantomsection
    \addcontentsline{toc}{chapter}{\contentsname}
    \tableofcontents
    
    \phantomsection
\addcontentsline{toc}{chapter}{Preface}
\chapter*{Preface}

\lettrine[lines=3]{O}{ver the past year,} the project that has led to this thesis has formed a great opportunity for me to get acquainted with a number of fields within mathematics, philosophy and physics, and to have many interesting discussions with people working in these fields. I could not have foreseen the multitude of directions that this project has ventured into and along the way, it has proven a significant challenge to contain its scope. As a result, I think there are many more discussions to be held and options to be considered on this topic, and this project is, at least for me, unfinished. Because I also realise that many of the reasonings in this thesis might be naive and that I am not remotely acquainted with all relevant existing literature, and simply because I would like to hear thoughts and opinions of others on this topic, I encourage the reader to share any useful comments. As a further disclaimer, this thesis addresses some fundamental and empirically undecidable philosophical questions, which I try to approach in an objective way and the answers to which I am agnostic to, even when arguments that I give might seem to indicate otherwise.

I am deeply grateful to all the people who have guided, helped and inspired me along the way. First of all, I would like to thank my supervisor Klaas Landsman, who has aided me with much devotion throughout the year, formed a valuable source of inspiration and ideas and introduced me to many other experts on relevant topics. Great thanks also go to Nicolas Gisin, whose work in recent years has formed the main inspiration for this thesis and who kindly invited me to Geneva, where I was able visit him just in time before the pandemic started to impact travel regulations across Europe. We had some very insightful brainstorm sessions and I hope we will keep in contact about this topic. The same goes for Flavio Del Santo, whom I unfortunately could not meet in person. I would like to thank Wim Veldman for introducing me to intuitionistic mathematics through his course at Radboud University, for his enthusiasm about this project and for helping me out with many intuitionistic questions that emerged along the way. My meeting with Rosalie Iemhoff was also of great help in this area. Moreover, I would like to thank Bryan Roberts for his enthusiasm and effort in making many useful suggestions from the philosophical side, as well as for inviting me to the online LSE/Cambridge seminars on philosophy of physics, which not only brought forward insights useful to this thesis but also boosted my personal interest in this field, which I was (and am still) new to. Finally, I would like to thank Freek Wiedijk for being the second reader of this thesis, and last but not least, the Honours programme of the Faculty of Science for offering the opportunity to work on this project for the entire academic year and to travel to Geneva.

    \mainmatter
    
    \chapter{Introduction}

\lettrine[lines=3]{F}{or centuries, the} intimate symbiosis between mathematics and physics has formed a great source of inspiration for both fields. Their mutual independence, however, has proven to be an important factor in this relationship: over the years, many abstract mathematical structures that were developed completely independently of physics have turned out to be surprisingly suitable for applications in physics, while conversely, many of modern mathematics’ most important research directions were directly or indirectly inspired by findings in physics. Mathematics has long proven its unreasonable effectiveness in the natural sciences, at least by its success in making empirical predictions. However, precisely because much of the mathematical formalism used in physics today has originated independently of physics, it might be questioned whether this formalism also provides the best candidate to describe physical reality from an ontological perspective.

One example of an originally purely mathematical structure used everywhere in physics today is the real number system. The formalisation of the continuum in terms of real numbers in the late nineteenth century was accompanied by the emergence of paradoxes about uncountable infinities, and this played an important role in the motivation for the development of more constructive and `intuitive’ approaches to the continuum. One of these approaches was developed in the early twentieth century by the Dutch mathematician Luitzen Egbertus Jan Brouwer (1881--1966), founder of intuitionistic mathematics. The differences between intuitionistic mathematics and classical mathematics, as Brouwer called the usual approach to mathematics that we still use today, eventually led to one of the greatest debates in the foundations of mathematics, of which classical mathematics was the clear winner \cite{vandalen1990war}.

Somewhat surprisingly, however, physical considerations played a very insignificant role in this foundational crisis of mathematics.\footnote{Einstein, who was a figure of great influence at this time, remained stubbornly neutral in the conflict between intuitionistic and classical mathematics and wrote to Born: ``I do not intend to plunge as a champion into this frog-mice battle (\emph{Frosch-Mäusekrieg}) with another paper lance’’ \cite{vandalen1990war}.} As a result, the mathematical real number system was not designed to accurately represent the `physical continuum’, i.e.\ the number line representing the possible values of physical quantities. Still, this mathematical formalism is widely applied in contemporary mathematical physics.

In recent years, however, multiple publications in physics and mathematics have expressed doubts as to whether the `real’ numbers do indeed deserve a place in physical reality. The authors of these publications are mainly concerned that, informally speaking, real numbers are infinitely \emph{precise} and contain an infinite amount of \emph{information}, and that this could imply that real numbers do not accurately represent physical quantities. Several alternative number systems have been proposed; in particular, Nicolas Gisin and Flavio Del Santo have proposed a different view on physics in which physical quantities are, at each point in time, only finitely precise and can be described with finitely much information \cite{gisin2017time, gisin2018, gisin2019real, dSG19, gisin2020comment, dS20}. We will refer to such theories under the broad term \emph{finite-precision theories}.

In addition, the question arises whether the continuum of intuitionistic mathematics might be more suitable to represent physical quantities; this continuum, after all, ought to be more `intuitive’ than the classical one. However, also the development of intuitionism was not carried out with physical applications in mind; as we will see in this thesis, its philosophy might even be so human-centred that application to physics would not make much sense.

The aim of this thesis is to give an account of some of the motivations for and consequences of finite-precision theories of physics, to discuss the appropriateness of intuitionistic mathematics to formulate such theories and to propose a new mathematical formalism for finite-precision classical physics. Note that many of the considerations in this thesis can be viewed as regarding either the epistemology or the ontology of physics; our focus, however, will be throughout on the ontology.

We start in Chapter~\ref{chap:whats-the-problem} by discussing the motivations for finite-precision physics in more detail. We argue that one cannot empirically decide whether physical quantities are finitely or infinitely precise. In Chapter~\ref{chap:gisin}, we summarise Gisin and Del Santo's approach to finite-precision classical mechanics and give some comments about it. In Chapter~\ref{chap:intuitionistic-physics}, we explore some of the motivations for using constructive or intuitionistic mathematics for physics in general; we will see that applying intuitionistic philosophy to physics is not straightforward and perhaps even impossible. Nevertheless, we will try to use the language of intuitionistic mathematics to formulate a  finite-precision physics, an attempt which suffers from a variety of problems. In Chapter~\ref{chap:further-development}, an alternative formulation of finite-precision classical mechanics is proposed, which uses classical mathematics. Chapter~\ref{chap:open-questions} presents some open questions and suggestions for future research on finite-precision theories of physics. Appendix~\ref{app:preliminaries-physics} contains very brief introductions to the mathematical formalism of classical mechanics (which can be skipped by those acquainted with classical mechanics, but also introduces notation used in Chapter~\ref{chap:further-development}) and to determinism. Finally, Appendices~\ref{app:intuitionism} and~\ref{app:computability-theory} contain preliminaries from intuitionistic mathematics and computability theory, respectively.
    \chapter{Are physical quantities finitely precise?}\label{chap:whats-the-problem}

\lettrine[lines=3]{I}{n recent years,} multiple authors have raised concerns about the use of the classical real number system in physics. More specifically, they question what we will call the \defn{orthodox} interpretation of physical theories,\footnote{We will sometimes just say `orthodox theories'. The use of the term `orthodox' is inspired from \textcite{dS20}; the term `classical' might also be possible, but we already use that to distinguish between classical and quantum theories (in either orthodox or non-orthodox interpretations).} namely the interpretations according to which physical quantities have infinitely precise (point-like) values and can be represented by arbitrary real numbers. See for example Gisin and Del Santo~\cite{gisin2017time, gisin2018, gisin2019real, dSG19, gisin2020comment, dS20}, \textcite{dowek2013real}, \textcite{visser2012}, \textcite{drossel2015}, \textcite{chaitin2004} and \textcite{lev2017}. Most concerns given in these articles have to do with either the fact that real numbers are infinitely precise, or the fact that almost all real numbers are uncomputable, which, in a certain sense, means that they contain infinitely much information---while actual infinities are in many cases considered non-physical (see also \textcite{ellis2018}). In section~\ref{sec:problem-randomness-and-indeterminism}, I outline (a version of) the main argument from \textcite{gisin2018} and an alternative interpretation of real numbers emerging from it. In section~\ref{sec:problem-infinite-information}, we investigate some other commonly used physical arguments against the real numbers, which I believe are not correct or at least incomplete.

Note that the question whether real numbers accurately describe ontological reality assumes a realistic attitude; although some of the arguments presented below are based on the in principle limitations of measurement, we always assume that a theory should be able to describe the \emph{perfect-information} state of the Universe, and not just that which is known to observers.

Although many of the arguments presented in this chapter are applicable to a large range of physical theories (indeed, all theories that use the real numbers), following \textcite{gisin2018}, we will largely focus on classical mechanics, which is a simple example of a physical theory using the real numbers. We will discuss the reason for this in section~\ref{sec:beyond-classical}. Appendix~\ref{sec:hamiltonian-formalism-orthodox} introduces the mathematical formalism of (the orthodox interpretation of) classical mechanics.

\section{Chaotic systems and finite-precision physics}\label{sec:problem-randomness-and-indeterminism}
In the orthodox interpretation, a classical physical system is thought of as being represented by a single real vector indicating a point in phase space, which causes the time evolution of the system to be completely determined by the initial conditions (see Appendix~\ref{sec:hamiltonian-formalism-orthodox}). This also applies to classical \emph{chaotic systems}, which are, loosely speaking, systems whose behaviour over time is very sensitive to the initial configuration. The ubiquity of these chaotic systems in physics was only recognised during the previous century, when it became clear that e.g.\ the long-term outcome of weather predictions is drastically influenced by only small changes to the initial condition \cite{lorenz1963,sep-chaos}. A simpler example of a chaotic system, however, is a double pendulum (see e.g.\ the figure).

\begin{wrapfigure}{o}[1cm]{0pt}
    \centering
    \fbox{\begin{tikzpicture}[scale=1,x=1cm, y=1cm]%
        \draw (-1,0) -- (1,0);
        \draw (-1.0,0) -- (-0.8,0.2);
        \draw (-0.8,0) -- (-0.6,0.2);
        \draw (-0.6,0) -- (-0.4,0.2);
        \draw (-0.4,0) -- (-0.2,0.2);
        \draw (-0.2,0) -- (-0.0,0.2);
        \draw (-0.0,0) -- (0.2,0.2);
        \draw (0.2,0) -- (0.4,0.2);
        \draw (0.4,0) -- (0.6,0.2);
        \draw (0.6,0) -- (0.8,0.2);
        \draw (0.8,0) -- (1.0,0.2);
        
        \draw[black!30] (0,0) -- (0.919239,-0.919239) -- (2.045072,-0.269239);
        \draw[fill,black!40] (0.919239,-0.919239) circle (0.035);
        \draw[black!30] (0,0) -- (0.791390,-1.031359) -- (2.080268,-0.861675);
        \draw[fill,black!40] (0.791390,-1.031359) circle (0.035);
        \draw[black!30] (0,0) -- (0.650000,-1.125833) -- (1.905704,-1.462298);
        \draw[fill,black!40] (0.650000,-1.125833) circle (0.035);
        \draw[black!30] (0,0) -- (0.497488,-1.201043) -- (1.528848,-1.992433);
        \draw[fill,black!40] (0.497488,-1.201043) circle (0.035);
        \draw[black!100] (0,0) -- (0.336465,-1.255704) -- (0.986465,-2.381537);
        \draw[fill,black!100] (0.336465,-1.255704) circle (0.035);
        \draw[fill] (0,0) circle (0.05);
        
        \node[anchor=center] at (.5,-2.9) {\small Double pendulum};
    \end{tikzpicture}}
\end{wrapfigure}

Let us consider a double pendulum whose configuration at time $t=0$ is specified by a set of real numbers, containing the real number $x_0$, say,
\[ x_0 = 0.36824317\ldots, \]
which represents e.g.\ the distance between the point of suspension and the tip of the pendulum, where the length of the upper arm of the pendulum is taken as the length unit. According to the orthodox interpretation, all digits in the infinite decimal expansion of $x_0$ are determined at $t=0$; that is, in order to fully describe the system at $t=0$, all digits in the decimal expansion should be given, even if the very far-away digits might not be physically relevant (i.e.\ have a physically insignificant influence on the configuration of the system). Furthermore, because the system is chaotic, every digit in the decimal expansion of $x_0$ is relevant to the behaviour of the double pendulum over time; that is, for every $n$, there exists a timepoint $t$ at which the value of the $n$-th digit of $x_0$ influences the configuration of the pendulum on the order of magnitude of the pendulum itself. The orthodox interpretation therefore seems right to the extent that all digits in a real number should exist (i.e.\ have a well-determined value) at least at some point in time, even if they have no physical relevance at $t=0$, because they can obtain physical relevance as time progresses (at least in a chaotic system).

However, the orthodox interpretation makes the additional assumption that all digits of $x_0$ (and all other relevant real numbers characterising the initial condition) are already determined \emph{at $t=0$}, even if they only obtain physical relevance in the far future. This assumption cannot be empirically verified, nor falsified, because the moment one performs a measurement, the measured (digits of) quantities become physically relevant. That is to say, we cannot empirically decide whether physically irrelevant digits have a well-determined value without making these values physically relevant and thereby forcing them to having a well-determined value.\footnote{In this section, we focus on the decimal expansion of real numbers because it is intuitive and intuitively brings across the current argument; however, the argument does not depend on this particular representation of the reals. As we will see later in this thesis, identifying physical quantities with their decimal expansion poses some problems.}

\

\Textcite{gisin2018} uses a similar argument and suggests that an alternative theory of classical mechanics, and of real numbers in general, is possible. In this theory, not all digits have a well-determined value at $t=0$; instead, digits only attain a determined value once they become physically relevant (and this happens by an indeterministic process, as will become clear below). He argues that this theory is empirically equivalent to the orthodox interpretation of real numbers, that is, the theories yield the same measurement results and can therefore not be distinguished empirically.\footnote{For this reason, we sometimes call these theories merely \emph{interpretations} of the same theory; however, because the two approaches give very different philosophical accounts of reality, a scientific realist would certainly regard them as two distinct theories.} Let us call such theories in which, as opposed to orthodox theories, quantities are only finitely precise \defn{finite-precision theories}. These finite-precision theories form the main subject of this thesis.\footnote{\label{fn:determined-precision-warning}Although the term ‘finite precision’ might be associated to the finite precision of measurements, this meaning is not intended here. Instead, finite precision of a physical quantity means that the quantity is \emph{inherently determined} up to only finite precision; Nature simply has not yet determined a more precise value. This can (but in many other ways cannot) be compared to the term \emph{uncertainty} used in quantum mechanics; however, also this term misleading, even in quantum mechanics, as it suggests that it depends on human knowledge about physical quantities, while it is actually a property of the physical system itself.}

What exactly `physically relevant’ means in this context, and where the border between relevant and irrelevant lies, is of course unclear. The reasoning above suggests that, in order for the finite-precision and orthodox theories to be empirically equivalent, a minimum requirement is that quantities should at least be physically relevant if they are or have been measured by intelligent beings. This suggests that the action of measurement might play an important role in finite-precision theories, and that a `classical measurement problem' exists, similar to the measurement problem of quantum mechanics. We will discuss the classical measurement problem in more detail in section~\ref{sec:gisin-classical-measurement-problem}. Note, however, that limitations on human measurement, even in principle ones, are \emph{not} used as an argument in favour of finite-precision theories by Gisin or by me; in our view, these limitations do not tell us anything about the inherent (ontological) determinateness or preciseness of values of physical quantities themselves.

\subsection{Randomness and indeterminism}\label{sec:randomness-indeterminism}
Because the orthodox and finite-precision theories are empirically equivalent, one can give arguments for either of the two only on the basis of naturalness or elegance. We will now discuss one such argument, which I think speaks in favour of the naturalness of finite-precision theories and is based on the fact that almost all real numbers (with respect to the Lebesgue measure) are uncomputable. In brief, this means that the decimal expansions of these numbers cannot be computed by an algorithm (as opposed to the decimal expansions of computable numbers like $1$, $\sqrt 2$ and $\pi$). (See Appendix~\ref{app:computability-theory} for more details on computable numbers.) In fact, almost all real numbers are \emph{1-random}, which is a much stronger mathematical definition of randomness which intuitively captures that the decimal expansions are incompressible, unpredictable and patternless.\footnote{For an introduction to the mathematical theory of 1-randomness (or \emph{algorithmic randomness}), see e.g.\ \textcite{terwijn2016randomness} or \textcite{dasgupta2011}.}

As a result, $x_0$ is 1-random with probability one,\footnote{Here, the meaning of \emph{probability one} is that among the set of all possible initial conditions of the double pendulum, the set of initial conditions where $x_0$ is 1-random has full measure; that is, when ‘drawing’ an initial condition ‘at random’ from this set, the value of $x_0$ is almost surely 1-random.} which means that the behaviour of the double pendulum over time is seemingly random and unpredictable (in the informal sense). Still, in the orthodox interpretations of chaotic systems like the classical double pendulum, it is assumed that the behaviour of the system over all time is completely encoded in the configuration of the system at $t=0$ only, which is why one speaks of \emph{deterministic chaos}. This apparent coexistence of chaos and determinism, which manifests itself in the emergence of randomness from simple physical laws, can be considered counterintuitive or unnatural.

In the finite-precision theory as introduced above, on the other hand, digits of physical quantities only attain a well-determined value once they become physically relevant. The existence of chaotic systems, whose behaviour at least appears random, suggests that it is reasonable to assume that the process by which these new values become determined is indeterministic.\footnote{One might also argue for this by noting that if the new values were completely fixed by the physical laws and the finitely many digits defining the initial state, this would probably mean that the resulting real numbers would not be 1-random. This is a heuristic argument, however, since we cannot simply identify physical laws with mathematical algorithms.} As a result, the time evolution of a chaotic system is not completely encoded in the initial condition. In this way, finite-precision theories bridge the gap between the mathematical and physical notions of randomness, namely by promoting 1-randomness of real numbers to physical indeterminism.

\

While I believe that these observations speak in favour of finite-precision theories, whether the orthodox theories or the finite-precision theories are more natural is partly, of course, a matter of personal taste. Let me stress once more, however, that the theories are empirically equivalent, so that the assumption that physical quantities are determined with infinite precision is empirically both unverifiable and unfalsifiable. This means that physical quantities could \emph{just as well} be determined up to only finite precision\footnote{See footnote~\ref{fn:determined-precision-warning}.} at each point in time, and consequently, that classical physics could \emph{just as well} be indeterministic as deterministic (which also holds for physical theories in general, as remarked in section~\ref{sec:determinism}). It is therefore surprising that the finite-precision interpretation of real numbers has barely been studied before and has only so recently been brought forward by Gisin.

\subsection{Intuitionistic mathematics}
The view that at each point in time, only a finite number of digits in the decimal expansion of a physical quantity have a well-determined value very closely resembles the philosophy of infinite sequences in intuitionistic mathematics, in which time plays a central role (see e.g.\ sections~\ref{sec:choice-sequences-continuity} and~\ref{sec:intuitionism-time}). This puts forward the idea that intuitionistic mathematics might be the right language to express finite-precision theories of physics, as has been suggested in \textcite{gisin2020comment} (see also the popular account by \textcite{quanta2020}). This is investigated further in Chapter~\ref{chap:intuitionistic-physics}.

\subsection{How have we come to the orthodox interpretation?}\label{sec:how-have-we-come-to-orthodox}
I think that the widespread acceptance of the orthodox interpretation can be attributed to at least two factors.
\paragraph{Precision of measurements}
The precision of human measurements has rapidly increased over the past few centuries. This might have led us to believe that we can in principle measure quantities with arbitrary but finite precision,\footnote{Note that we are talking about real physical quantities, and not about quantities such as expected position and momentum in quantum mechanics.} and that therefore, physical quantities have to be determined with infinite precision (because otherwise, measurements of quantities with a lower level of precision than the level of the measurement apparatus would be inconsistent and irreproducible). However, in this reasoning, ‘arbitrary but finite precision’ is too readily extrapolated to ‘infinite precision’, because it is not taken into account that when performing multiple measurements with increasing precision on the same physical quantity, the precision to which the quantity is determined can increase in the time between the measurements (but still be finite at each point in time). Moreover, it is not taken into account that the very act of measurement might cause the physical quantity to acquire a more (but still finitely) precise value. As we have argued above, this is not a far-fetched suggestion, since measurements naturally increase the physical relevance of the measured quantities; moreover, note that measurements also play a central role in the (widespread) Copenhagen interpretation of quantum mechanics. 

Note (again) that I do \emph{not} use the in principle limitations on human measurement as an argument for finite-precision theories as opposed to orthodox theories, but only that it follows from these limitations that the theories are empirically equivalent.\footnote{One could say that a finite-precision theory carries some of the arguments against  \emph{predictability} over to \emph{determinism}, by introducing finite precision, which is frequently associated with empiricism only, to the ontological level. See also section~\ref{sec:determinism}. Finite-precision theories are (in my view) not intended to identify predictability with determinism, however!}

In this answer to the question in the section title, we once again see a connection to intuitionistic mathematics, for Brouwer suggested that the law of the excluded middle, and the limited principle of omniscience (LPO) in particular (see Appendix~\ref{app:intuitionism}), is accepted by classical mathematicians because they tacitly extrapolate reasoning about finite sequences to reasoning about infinite sequences (see section~\ref{sec:NN}).

This observation suggests that infinitely precise real numbers are merely an idealisation, a limit case, of physical reality. Indeed, Gisin has suggested that the orthodox theory describes a ‘view from the end of time’ (i.e.\ a view of the system in the limit $t\to\infty$). This can also been seen as the reason that the orthodox theory is deterministic.

\

\paragraph{The mathematical continuum}\label{sec:continuum}
Another explanation for the widespread adoption of the view that physical quantities are infinitely precise is the mathematical formalisation of the continuum in the nineteenth century. While the notion of the continuum goes back to Ancient Greece, only in the nineteenth century was it formalised as being built up of an infinite set of points. Although viewing the continuum as a set of points is a logical consequence of the central place of set theory in classical mathematics, it does not correspond to the intuition behind the continuum. This inspired the development of continua in constructive mathematics,\footnote{Both Brouwer and Weyl spoke of the `intuitive continuum’ \textcite{weyl1918kontinuum,van2002brouwer}, while Borel spoke of the `geometric continuum’ \cite[§1.4.2]{TrvD}.} and in particular intuitionistic mathematics, where not points but ever-shrinking rational intervals are central to the continuum. Also in nonconstructive mathematics, approaches have been developed to formalise the continuum using \emph{regions} (e.g.\ open sets) instead of points; see e.g.\ \textcite{hellman2018varieties} or \textcite{johnstone1983point}. It might have been that the widespread acceptance of the usual classical definition of the real numbers has led to the view that also physical quantities are given by infinitely precise points. However, while the notion of the continuum is essential to physics, it is questionable whether physics needs it to consist of points. Accordingly, the finite-precision theories discussed in Chapters~\ref{chap:intuitionistic-physics} and~\ref{chap:further-development} make use of regions (which we shall later call \emph{domains of indeterminacy}) instead of points.\footnote{Ideally, we should try to make a clear distinction between the `mathematical continuum' and the `physical continuum'; for example, even if physical quantities should be expressed by regions instead of points, numbers like 1 and $\pi$ can still be said to be correspond to an infinitely precise point on the mathematical continuum. However, making this distinction is difficult because, for example, the relation between physical quantities can depend on point-like mathematical constants like $\pi$.}

\subsection{Indeterministic classical physics}\label{sec:indeterministic-classical-physics}
Historically, deterministic theories have generally been regarded as more natural or intuitive than indeterministic theories. This can among other things be attributed to the fact that many macroscale phenomena like the falling of an apple or the motion of the planets around the sun look deterministic. 
The preference for deterministic theories led to historic debates at the advent of quantum mechanics in the early twentieth century, and to the development of deterministic quantum theories, notably Bohmian mechanics (see also Appendix~\ref{sec:determinism}), which remain to be developed to this day \cite{hooft2016cellular}, and might be said to try to `bring quantum closer to classical'.

However, the possibility that even classical mechanics need not be deterministic (and that viewing it as indeterministic might even be more natural than viewing it as deterministic) shows us that quantum mechanics is not necessarily the only or the first theory to introduce indeterminism to physics. An indeterministic interpretation of classical mechanics as the one outlined in this section might therefore cause the community to be more at ease with the concept of indeterminism and, accordingly, with quantum mechanics, by `bringing classical closer to quantum' instead of the other way around. This was one of the main motivations for Gisin to develop his theory.\footnote{Personal communication.} He has suggested that the real numbers can be seen as the hidden variables of classical mechanics \cite{gisin2019real}, comparing them to the hidden variables which are supplemented to quantum theory by Bohmian mechanics in order to make quantum theory deterministic. See also \textcite{dS20} for more historical discussion on this issue.

The finite-precision interpretation, being time-irreversible, also influences the relationship between classical mechanics and thermodynamics. This is discussed in more detail in section~\ref{sec:past-thermo}.



\subsection{Parmenides and Heraclitus time}\label{sec:heraclitus-parmenides}
It is useful to distinguish two notions of time in finite-precision physics. The first, which \textcite{gisin2017time} calls \defn{Parmenides time}, corresponds to time evolution as given (in the case of classical mechanics) by the Hamiltonian differential equations of motion. Parmenides time could also be called `boring time' \cite{gisin2017time}, as its evolution is completely determined on the basis of the initial condition (i.e.\ it is deterministic); no new information is generated as Parmenides time passes. It has no preferred direction but is just another parameter of spacetime. (Parmenides was an Ancient Greek philosopher according to whom existence is timeless and change is deceptive). On the other hand, we have what Gisin calls \defn{Heraclitus time}, the evolution of which is indeterministic and which corresponds to generation of new information and to the change from \emph{potential} to \emph{actual} (and, perhaps, to free will \cite{gisin2017time}); it could also be called `creative time'. In the current setting, it refers to the process of determination or `actualisation' of new digits of physical quantities. (Heraclitus was an Ancient Greek philosopher who believed, on the other hand, that existence is constantly changing; \textgreek{πάντα ῥεῖ} (`everything flows')).

Parmenides time and Heraclitus time could be compared with propagation of the wave function via the Schrödinger equation and indeterministic wave function collapse in quantum mechanics, respectively.

In some sense, Parmenides time and Heraclitus time are perpendicular, since the Hamiltonian differential equations can be solved using finite-precision quantities as initial conditions (i.e.\ Parmenides time evolution can be calculated within one Heraclitus time slice). However, if the actualisation of digits is indeed triggered by measurements or by them becoming `physically relevant' over the course of Parmenides time, then it seems that Parmenides time and Heraclitus time must be inextricably linked. We will return to this issue later.

\subsection{What about theories besides classical mechanics?}\label{sec:beyond-classical}
A natural question to ask is why we focus on the example of classical mechanics, as we already know that precisely on the small scale, classical mechanics does not accurately represent reality. We do this first of all because classical mechanics, in its orthodox interpretation, is an archetype of a deterministic theory, and as explicated in section~\ref{sec:indeterministic-classical-physics}, we think it is useful to show that even this theory \emph{can} be interpreted indeterministically. But indeed, most motivations for questioning that physical quantities are infinitely precise can be generalised to any other physical theory that uses the real numbers. Moreover, similarly to classical physics, the idea that quantities get more precisely determined over time would introduce (another level of) indeterminism to these theories, by promoting mathematical randomness to physical indeterminism. Classical mechanics, by virtue of being a simple theory which is usually regarded as deterministic, allows us to explore this particular indeterministic process.

\section{Infinite information}\label{sec:problem-infinite-information}
Many of the publications mentioned in the beginning of this chapter are in particular concerned with the aforementioned fact that most real numbers are uncomputable, which (can be and) is frequently described as them `containing infinitely much information’, and that this in conflict with the alleged principle that the Universe should have a ‘finite information density’, i.e.\ that a finite volume of space ‘contains at most finite information’. I am not convinced by the validity of this \emph{finite-information principle} and believe that the question whether it holds or not is empirically underdetermined. In this section we will try to analyse some of the arguments used in favour of this principle.

Before doing that, however, we must have clarity on what exactly is meant by the information ‘contained’ in a physical system. It makes sense to define it as the minimal amount of information necessary to completely specify the configuration of that system. Since completely specifying the configuration of a system requires the system to be isolated, this definition is limited to isolated systems, and it is therefore questionable whether it makes sense to speak about the information contained in a particular \emph{finite} region of space. Indeed, we cannot say that the information in a single real number representing, say, the $x$-coordinate of a particle, is ‘stored’ at the location of that particle, for that information depends on our description of the system (e.g.\ we can always choose a coordinate system in which this physical quantity is a rational number).\footnote{In classical mechanics, in the presence of forces with infinite range like Newtonian gravity, a change in the location of the particle immediately affects the behaviour of the system at arbitrarily large distances; therefore, in this case we cannot even say that the physical quantity in question is ‘localised’ at the location of the particle, let alone the information in the real number representing its value.} The only well-posed question seems to be whether \emph{all} physical quantities relevant to describing the entire system can together be described using finitely much information (i.e.\ can be expressed by a finite algorithm). Therefore, let us from now on only consider isolated systems of finite spatial extent. The Bekenstein bound indeed only applies to such systems \cite{page2018}.

\subsection{The Bekenstein bound}
Some of the earlier stated papers \cite{gisin2018, dS20, dowek2013real, chaitin2004} in particular mention the \emph{Bekenstein bound} as an argument against the infinite information in real numbers. Derived first in 1981 by Jacob Bekenstein in the context of black hole physics \cite{bekenstein1981}, it provides an upper bound on the ratio between the entropy and energy in an isolated system enclosed in a sphere of finite radius $R$:
\begin{equation}\label{eq:bekenstein}
    \frac{S}{E} \leq \frac{2\pi k R}{\hbar c},
\end{equation}
where $k$ is Boltzmann’s constant, $\hbar$ is Planck’s constant and $c$ is the speed of light. 

Entropy is often used as a measure of information, so that the infinite information in real numbers would contradict this bound. However, I think this link between entropy and information is too readily made. First of all, there exist many different notions of entropy, which should not thoughtlessly be identified with each other. What complicates matters even more is that entropy is not a property inherent to a physical system; rather, it is a property of our \emph{description} of the system \cite{Cat08,FrW11}. Hence, there is no such thing as `the’ entropy of the system. For this and other reasons, the domain of applicability of the Bekenstein bound is unclear (it is often ambiguously stated as in Equation~\eqref{eq:bekenstein} without specifying the specific entropy notion that is meant) \cite{page2018}.

The entropy involved in black hole thermodynamics, for example, which is the context that the Bekenstein bound was first derived in, is very different from entropy in statistical mechanics. In order to pass from thermodynamic entropy to statistical mechanical entropy, one needs to, at least in continuous systems,\footnote{It is true that some (quantum) theories allow the existence of discrete systems; however, if we only consider discrete systems (i.e.\ systems with a countable number of possible microstates), I do not see why there is a problem of infinite information in the first place.} coarse-grain the system \cite{FrW11}, i.e.\ divide state space into countably many compartments; this means that microstates are already assumed from the start to correspond to a region of state space, instead of a point. In doing so, all information relevant to the current discussion, namely the information contained in uncomputable, infinitely precise real numbers, is lost.

Furthermore, while there has been much work on the connection between (statistical mechanical) entropy and the algorithmic information contained in \emph{finite} binary strings (i.e.\ the \emph{Kolmogorov complexity}; see e.g.\ \textcite{GrW08}, \textcite{tadaki2019}), whether and how this generalises to a link with uncomputability of \emph{infinite} strings corresponding to real numbers is yet unclear (but would be interesting to investigate in more detail).

\subsection{What information?}
\textcite[section~IV]{gisin2018} outlines another argument for the principle of finite information density, based on the observation that although the information storage capacity has dramatically increased over the past century, there will always remain a certain minimal amount of energy, mass or space that is needed to encode one bit of information. This refers, however, to information that is stored by human beings in a digital format, which, almost by definition, does indeed have a finite density; in my view, the argument does not apply to the information necessary to completely describe the state of the system itself.

In addition, arguments based on thought experiments known as Landauer’s principle or Szilard’s engine \cite{sep-information-entropy} are sometimes used to make the connection between thermodynamic entropy and information (not requiring coarse-graining by passing through statistical mechanics) \cite{dS20}. However, Landauer’s principle deals with information processing carried out within the Universe, so it can again be questioned whether it applies to all information necessary to \emph{describe} the system; similarly, Szilard’s engine involves an intelligent being (``Maxwell’s demon’’) knowing the state of the system and actively interacting with the system from within.

In \cite{lloyd2002}, Lloyd calculates an upper bound on the number of logical operations performed and numbers of bits registered within the Universe in its lifetime, using the Bekenstein bound and the Margolus-Levitin theorem. He argues that these numbers also provide a lower bound on the number of logical operations and bits required to simulate the entire Universe on a (quantum) computer. He notes correctly, however, that whether it is also equal to the \emph{minimum} amount necessary to run such a simulation is a controversial question which cannot be decided on the basis of only these physical principles.

\

I agree with Lloyd and stick with the conclusion of section~\ref{sec:problem-randomness-and-indeterminism} that it cannot be decided, in particular not by the arguments discussed in this section, whether the real numbers can be used to accurately represent the values of physical quantities, but that it is interesting to explore the possibility of finite-precision and finite-information physics.

As a final note, it is of course possible to argue against the real numbers from an epistemic or operationalist\footnote{\label{fn:operationalism}Operationalism is the view that a concept is only meaningful when we have a method of measurement for it; more abstractly, it views any concept as nothing more than a `set of operations' \cite{sep-operationalism}.} perspective, since the real numbers do not represent our knowledge of a physical system and there is no method to measure a quantity with infinite precision, nor a method to store a measurement result that contains infinite information. Most discussions in this thesis, however, assume an ontological perspective.
    \chapter{Gisin's alternative classical mechanics}\label{chap:gisin}

\lettrine[lines=3]{I}{n a series of} publications from 2017 to 2020, Nicolas Gisin and Flavio Del Santo, motivated by the arguments given in Chapter~\ref{chap:whats-the-problem}, propose a candidate alternative interpretation of classical mechanics which is indeterministic and uses only finitely much information for each physical quantity \cite{gisin2017time, gisin2018, gisin2019real, gisin2020comment, dSG19, dS20}. In section~\ref{sec:gisin-fiqs}, we summarise this theory and discuss two more aspects of their publications; in section~\ref{sec:gisin-reaction}, we discuss some apparent problems arising from their approach.


\

\section{Finite-information quantities}\label{sec:gisin-fiqs}
The idea of indeterministic classical mechanics as first outlined in \textcite{gisin2018} was already discussed in section~\ref{sec:problem-randomness-and-indeterminism}. The theory is worked out in greater detail in \textcite{dSG19}. Instead of the decimal expansion as in section~\ref{sec:problem-randomness-and-indeterminism}, the authors consider the binary expansion of a real number $\g$, without loss of generality situated in the interval $[0,1]$, which represents some (dimensionless) physical quantity (e.g.\ a ratio of distances between particles):
\[ \g = 0.\g_1\g_2\g_3\dots, \]
where $\g_j\in\{0,1\}$ for each $j > 0$. Instead of assuming, as is done in classical mathematics and orthodox physics, that all binary digits of $\g$ are given at once, i.e.~at each point in time all digits are either 0 or 1, they propose that at each point in time only finitely many digits are 0 or 1. The other digits take a value between 0 and 1, defined as the \defn{propensity} of that digit at time $t$. This propensity can be seen as the tendency of the digit to take the value 1 at a later stage. More specifically, the authors define:
\begin{definition}
	A \defn{finite-information quantity} is an infinite sequence of propensities $(q_1,q_2,\dots)$ such that:
	\begin{enumerate}[\quad (i)]
		\item $q_j \in \Q\cap [0,1]$ for all $j>0$;
		\item (\emph{necessary condition}) $\sum_{j=0}^\infty (1-H(q_j)) < \infty$, where \[ H(q_j) = -q_j\log_2 q_j - (1-q_j)\log_2(1-q_j) \] is the base-2 entropy of the probability distribution corresponding to $q_j$.
	\end{enumerate}
\end{definition}

Both conditions (i) and (ii) are imposed to make sure that FIQs contain only finitely much information. Gisin and Del Santo also give a sufficient condition for (ii):
\begin{quote}
	(\emph{sufficient condition}) For each time $t$, there exists $M(t) \in \N$ such that $q_j = \inv2$ for all $j>M(t)$.
\end{quote}

Let us for the moment restrict our attention to finite-information quantities $\g$ satisfying this additional constraint, as Del Santo and Gisin also mostly do. For each $t$, let $N(t)$ be the largest $n$ such that at time $t$, $q_j \in \{0,1\}$ for all $j$ with $0 < j \leq N(t)$. Then $N(t) \leq M(t)$, and the sequence of propensities associated with $\g$ can be divided into three sections.

In the first section, $0 < j \leq N(t)$, all propensities $q_j$ are either 0 or 1. This means that the corresponding digits $\g_j$ have a well-determined value, equal to the propensity.

In the second section, $N(t) < j \leq M(t)$, the propensities $q_j$ take a rational value between 0 and 1.\footnote{In \textcite{dSG19} and \textcite{dS20} the additional assumption seems to be made that $q_j$ cannot be equal to 0 or 1 for $j > N(t)$.} These propensities are taken to be objective, ontological properties of the physical quantity. Over time, they undergo a dynamical evolution which moves them closer to either 0 or 1. When one of these numbers is reached, the bit $\g_j$ changes from \emph{potential} to \emph{actual}.

The third group of propensities, $j > M(t)$, satisfy $q_j  = \inv2$. According to the authors, this means that the outcome of the bit $\g_j$ is totally random.

\


\subsection{The classical measurement problem}\label{sec:gisin-classical-measurement-problem}
In section \textsc{v}, \textcite{dSG19} discuss the question of under what circumstances a bit value is changed from potential to actual, i.e.~a propensity becomes either 0 or 1. The authors present two possible answers: (i) The actualisation happens spontaneously as time passes, i.e.\ the process of actualisation (but not the outcome) only depends on the theory itself, and is not influenced by e.g.\ strong emergence.\footnote{A high-level phenomenon is said to be \defn{strongly emergent} from a lower-level domain if the phenomenon arises from the lower level, but not all truths about the phenomenon can be deduced, even in principle, from the lower level. This opposes \defn{reductionism}, which is roughly the view that a higher-level object is nothing more than its constituent parts.} (ii) The actualisation occurs when a higher level requires it. This means that a (strongly) emergent process influences a lower-level one: this is referred to by \emph{top-down causation}. The higher level process can be a macroscopic measurement apparatus, for example, which might require a finite-information quantity to take on a more definite value. 

The situation can be compared to that of the \emph{measurement problem} in quantum mechanics, which has been described as the problem of ``explaining why a certain outcome, as opposed to its alternatives, occurs in a particular run of an experiment'' \cite{brukner2017quantum} and, in particular, why, when, how and \emph{whether} collapse of the wave function to an eigenstate occurs \cite{sep-qt-issues}.
The question of why, when, and how actualisation of digits occurs in finite-precision classical mechanics and how this process relates to measurements could accordingly be called the `classical measurement problem' \cite{dSG19,dS20}.
\textcite{dSG19} suggest that of the approaches to the classical measurement problem discussed above, option (i) is reminiscent of the \emph{objective} or \emph{spontaneous collapse models} of quantum mechanics such as the continuous spontaneous localisation (CSL) model \cite{sep-qm-collapse}, which describe a process of wave function collapse which is integrated into quantum theory itself; while option (ii) can be compared to the Copenhagen interpretation, according to which it is the act of measurement itself that induces wave function collapse \cite{sep-qt-issues,dSG19}. The classical measurement problem will remain unresolved, however, just as its quantum counterpart.

The classical measurement problem is also closely related to the nature of the relation between Parmenides time and Heraclitus time (cf.\ Schrödinger propagation and wave function collapse, respectively). We will return to this in section~\ref{sec:gisin-reaction-hamiltonian-evolution} and later in this thesis.

\section{Discussion}\label{sec:gisin-reaction}
Gisin and Del Santo’s proposal sets the stage for an interesting discussion on the relation between real numbers and indeterminism and their role in classical physics. While the formalism presented in \textcite{dSG19} succeeds in describing intuitively what it would mean for physical quantities to be inherently uncertain, the theory is still in its infancy and there seem to be some problems that limit its potential to become a more complete mathematical theory which represents objective reality.

\subsection{Base-2 dependence and interdependence of propensities}
First of all, if the indeterminacies of physical quantities are indeed ontological, objective properties of the system, it is unnatural to describe them in terms of the base-2 expansion of the reals, and doing so would lead to an incomplete theory, as we will see now.

Gisin and Del Santo do not go into the question whether the propensities associated with a FIQ are dependent or independent random variables.\footnote{The usual theory of random variables can perhaps not be applied to propensities, as Gisin notes that propensities are not probabilities as in the usual sense of the word; in particular, they do not satisfy Kolmogorov’s axioms of probability theory \cite{gisin1991}. However, it seems reasonable that also for propensities there must be some notion of dependence or independence. Two propensities can be said to be dependent, for example, if the sole process of the transition of the value of $\g_j$ from potential to actual causes a change in the value of the propensity $q_i$, with $i \neq j$.} However, it follows from a reasonable argument that they must in general be dependent.
Namely, note that if the propensities are assumed independent, then every FIQ induces a probability density function on the continuum via the joint probability of the propensities, as exemplified in Figure~\ref{fig:bar-charts}(a).
However, not all probability distributions can be reconstructed by taking the joint probability distribution in this way, as shown in Figure~\ref{fig:bar-charts}(b).\footnote{In particular, if $f(x)$ is a probability density function on the continuum induced by independent propensities $(q_1,q_2,\dots)$, then for all $x\in (0,\inv2)$, we must have $f(x+\inv2) = \frac{q_1}{1-q_1} f(x)$.}
This would mean that the set of possible probability distributions that appear in FIQ-theory depends on properties inherent to the description of the system, such as the chosen unit and coordinate system; as a result, the theory cannot describe objective ontological reality.

\begin{figure}
    \centering
    \begin{tikzpicture}[x=0.65cm, y=0.65cm]
        \begin{scope}
            \draw (-.5,0) -- (4.5,0);
            \draw (0,0) -- (0,1) -- (1,1) -- (1,2) -- (2,2) -- (2,1) -- (3,1) -- (3,2) -- (4,2) -- (4,0); 
            \draw (1,0) -- (1,1); \draw (2,0) -- (2,1); \draw (3,0) -- (3,1); 
            \draw (0,0) -- (0,-0.07) node[anchor=north] {0};
            \draw (4,0) -- (4,-0.07) node[anchor=north] {1};
            
            \node[anchor=center] at (-0.8,2) {(a)};
            \node[anchor=center] at (2,-1cm) {$(q_i)_{i=1}^\infty = (\inv2,\inv3,\inv2,\inv2,\ldots)$};
        \end{scope}
        \begin{scope}[shift={(7,0)}]
            \draw (-.5,0) -- (4.5,0);
            \draw (0,0) -- (0,1) -- (1,1) -- (1,2) -- (2,2) -- (3,2) -- (3,1) -- (4,1) -- (4,0); 
            \draw (1,0) -- (1,1); \draw (2,0) -- (2,2); \draw (3,0) -- (3,1); 
            \draw (0,0) -- (0,-0.07) node[anchor=north] {0};
            \draw (4,0) -- (4,-0.07) node[anchor=north] {1};
            
            \node[anchor=center] at (-0.8,2) {(b)};
            \node[anchor=center, align=left] at (2,-1cm) {$P(\g_i = 1) = \inv2$ \ for all $i$};
        \end{scope}
        \begin{scope}[shift={(14,0)}]
            \draw (-.5,0) -- (4.5,0);
            \draw (0,0) -- (0,1.5) -- (4,1.5) -- (4,0); 
            \draw (1,0) -- (1,1.5); \draw (2,0) -- (2,1.5); \draw (3,0) -- (3,1.5); 
            \draw (0,0) -- (0,-0.07) node[anchor=north] {0};
            \draw (4,0) -- (4,-0.07) node[anchor=north] {1};
            
            \node[anchor=center] at (-0.8,2) {(c)};
            \node[anchor=center, align=left] at (2,-1cm) {$(q_i)_{i=1}^\infty = (\inv2,\inv2,\inv2,\ldots)$};
        \end{scope}
    \end{tikzpicture}
    \caption{(a) Example of a probability distribution on $[0,1]$ arising from a sequence of independent propensities $(q_i)_{i=1}^\infty$.
    (b) Example of a probability distribution on $[0,1]$ which does not arise from a sequence of independent propensities. A random variable $\g\in [0,1]$ distributed according to this distribution has the property that for all bits in its binary expansion $0.\g_1\g_2\ldots$, $P(\g_i = 1) = \inv2$; but these propensities are dependent.
    (c) The distribution of (b) cannot be reconstructed from a sequence of independent propensities, since this would yield another probability distribution. We see that propensities do not completely describe the probability distribution and hence the ontology associated with finite-information quantities.}
    \label{fig:bar-charts}
\end{figure}
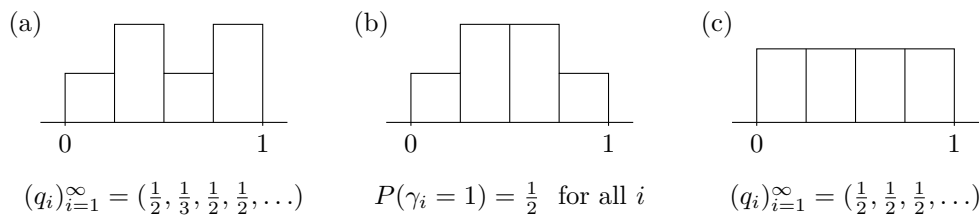

Therefore, the propensities must in general be dependent. This means, however, that not all information present in the physical system is encoded in the values of the propensities, so that FIQs do not provide a complete description of reality. Figures~\ref{fig:bar-charts}(b) and (c), for example, show examples of differing probability distributions associated with the same sequence of (dependent) propensities.

In fact, the conclusion that the binary expansion of reals cannot be used to formulate a \emph{complete} theory of finite-precision quantities is exactly analogous to the problem that prohibits defining the real numbers in intuitionism on the basis of their binary expansion (or expansion in any other base), as discussed in section~\ref{sec:int-reals-other-constructions} on page~\pageref{sec:int-reals-other-constructions}. A similar result also follows from considerations in computable analysis (see the discussion below Definition~\ref{def:some-representations} on page~\pageref{computability-decimal-expansion-problem}).

Indeed, it is arguably not the indeterminacy of the binary digits of real numbers that matters physically, but the indeterminacy of the location of the physical quantity on the continuum line itself.
This suggests that the theory would be improved if not the binary representations of real numbers were taken as a starting point in representing the indeterminacy in quantities, but instead more geometrical properties of the continuum were used.\footnote{Perhaps, the mathematical idea that the continuum is built up of individual points (which can then be expressed by their binary expansion) has led Gisin and Del Santo to their current formulation. As I remarked in section~\ref{sec:continuum}, however, there might be no physical motivation for viewing the continuum as consisting of points.}
In section~\ref{sec:intuitionistic-reals-physics}, we will take a first look at an approach to finite-precision physics that uses shrinking rational intervals instead of the binary expansion of the reals; in section~\ref{sec:mathematics-development}, we will consider a more consistent formulation of finite-precision classical mechanics.

\subsection{Measuring information}\label{sec:measuring-information}
To ensure that the information contained in (or `encoded by') the physical system per unit volume is finite, Gisin and Del Santo require that the propensities $q_j$ of a FIQ be rational numbers, and that the sum of their negentropies is finite: $\sum_j (1-H(q_j)) < \infty$ (the \emph{necessary condition}). These are two different notions of information: the former is concerned with the algorithmic information in a description of the perfect-information state of the system, while the latter is associated with the amount of information required to communicate the outcome of an actualisation of a digit of the FIQ. It is not clear whether the sum of negentropies provides the correct measure of information in a FIQ. We can see, for instance, that this choice does not guarantee a bound on \emph{algorithmic} information: it is not difficult to construct a sequence $(q_1,q_2,\dots)$ of rational propensities for which $\sum_j (1-H(q_j))$ converges, but which is uncomputable.\footnote{Namely, construct such a sequence $(q_j)_j$ that converges to $\inv2$ fast enough such that $\sum_j(1-H(q_j))$ converges, while making sure that the sequence is not computable; the latter is possible because the set of uncomputable rational sequences is dense in the set of all rational sequences.}

A first attempt to a solution could be to replace the \emph{necessary condition} by the \emph{sufficient condition}, as all propensity sequences that satisfy the \emph{sufficient condition} are computable. However, the probability distributions on the continuum induced by such sequences are necessarily discrete and show discontinuities only at dyadic numbers (as in Figure~\ref{fig:bar-charts}), which makes the theory even more dependent on the choice of unit and coordinate system. Another potential solution could be to replace the \emph{necessary condition} by the requirement that the sequence of propensities associated with a FIQ is computable.

In addition, the measure of information proposed in the formulation of the necessary condition, namely a sum of negentropies of individual propensities, seems to require that the propensities are independent, which they, as discussed in the previous section, are most likely not.

\

Finally, also the requirement that propensities are rational is not free of problems. \textcite[section~III-B]{dSG19} remark that replacing the reals by rationals in physics leads to what they call `Pythagorean no-go theorems’: for example, three particles cannot be placed on the vertices of a right-angled triangle, since the distance between the two particles on the hypotenuse would be irrational. However, this also holds for propensities: if the distance between two particles A and B is represented by a FIQ with rational propensities, and the same holds for the distance between A and another particle C, then the distance between A and C is in general not expressible in terms of rational propensities.

A potential solution could be to let propensities take values in the set of computable numbers. Together with our previous suggestion, this would mean that FIQs are defined as computable sequences of computable numbers. However, as noted earlier, the approaches developed in the next sections will be based on the geometry of the continuum as a whole, rather than on the base-2 expansion of individual points.

\subsection{Connection to Hamiltonian time evolution}\label{sec:gisin-reaction-hamiltonian-evolution}
While Gisin and Del Santo do discuss the possibilities for the mechanism behind the evolution of propensities $q_j$ which is involved in the transition of bit values $\g_j$ from potential to actual, they do not discuss how this evolution is incorporated in the dynamical evolution of FIQs through the Hamiltonian equations of classical mechanics. In the terms of section~\ref{sec:heraclitus-parmenides}, they do not discuss how Heraclitus time and Parmenides time are linked. Here, the rationals again seem problematic: when FIQs undergo Hamiltonian evolution, rational propensities do in general not stay rational.
Understanding the connection between Parmenides time and Heraclitus time turns out to be a difficult problem, which we will revisit later in this thesis.

    \chapter{Intuitionistic physics?}\label{chap:intuitionistic-physics}

\lettrine[lines=3]{T}{he main aim} of this chapter is to explore the motivations for and problems of using intuitionistic or constructive mathematics for physics. The discussion in section~\ref{sec:constructivising-physics} applies to physics in general, while in sections~\ref{sec:intuitionistic-reals-physics} and~\ref{sec:doubts-again}, an intuitionistic interpretation of finite-precision physics in particular is attempted and debated. Our conclusion is that the usefulness of constructivism and intuitionism in describing an ontological physical theory is questionable. Finally, section~\ref{sec:lawless-indeterminism} attempts to investigate the relation between physical indeterminism and Kreisel and Troelstra's definition of lawless sequences (introduced in section~\ref{sec:lawless-sequences-mathematics}). This section has a similar negative conclusion.

\section{Constructivising physics}\label{sec:constructivising-physics}
While constructive mathematics was originally developed for pure mathematics, the past century has seen multiple debates on the justification of using nonconstructive or constructive methods in applied mathematics and physics in particular. One important debate in the 1990s was between philosopher Geoffrey Hellman and constructivist Douglas Bridges \cite{hellman1993, richman1999, hellman1998, bridges1999}. More recently, alternative quantum logics have been proposed that are intuitionistic \cite{landsman2017}. Also Gisin has suggested that his alternative classical mechanics might best be expressed in the language of intuitionism \cite{gisin2020comment}. In this section, we review and analyse a number of arguments or motivations to use either constructive or nonconstructive mathematics in physics that can be found in the literature. These motivations can be roughly divided into three categories: purely mathematical motivation, technical considerations, and physical motivation, the latter of which can be subdivided into epistemic and ontological motivations. Of course, there is some overlap between these categories.\footnote{Another topic which touches on the applicability of constructive mathematics to physical sciences, but which I do not discuss here, is on the physical meaning of classical undecidability and incompleteness results. See e.g.~\textcite{svozil1995}.} The discussion in this section is not specific to finite-precision theories but applies to physics in general. Also note that this section focuses on (mainly Bishop's) constructive mathematics, and less so on intuitionistic mathematics, which are not the same thing (see section~\ref{sec:constructive-mathematics}).

\subsection{Purely mathematical motivation}
Ask any ‘radical constructivist’ whether to use classical or constructive mathematics in physics, and they will most likely answer that constructive mathematics is best to use in all cases. They might give the same arguments as they would for defending pure constructive mathematics: ‘What purpose does it serve to say that something exists, when it cannot be constructed? How do you know that the time axis is totally ordered, when the relation $<$ on $\R$ is not decidable?’ To me, these arguments are not convincing, for the simple reason that mathematics is not physics. The largest part of the debate between constructive and nonconstructive mathematics takes place entirely within mathematics. The BHK (Brouwer-Heyting-Kolmogorov) interpretation, for example, is an interpretation of what it means to have a constructive proof of a mathematical statement (see section~\ref{sec:constructive-mathematics}), and has little to do with physics. While pure mathematics is \emph{practised}, the goal of theoretical physics is to \emph{describe} physical reality (be it empirical or ontological reality);\footnote{According to Brouwer: ``Het gebouw der intuitieve wiskunde [is] zonder meer een \emph{daad}, en geen \emph{wetenschap}'' \cite[p98]{brouwer1907} (``The construct of intuitive mathematics is simply a \emph{deed}, and not a \emph{science}'').} hence, motivations to use constructive mathematics for physics should take into account which mathematics has the best representational capacity. Hellman draws the same conclusion:
\begin{quote}
	Why should there be any restrictions a priori on the character of the mathematics that may be used to describe real or idealized physical systems?
	\textelp{}
	In general, in scientific applications of mathematics, the goal of explaining and understanding natural phenomena is paramount, not achieving a constructive interpretation of results. \cite{hellman1998}
\end{quote}

\subsection{Technical considerations}
When it comes to the power of proving results, constructive mathematics is usually thought of as lagging behind classical mathematics. But is that a bad thing? It might be, if this means that certain mathematical results that seem essential to the development of physics cannot be proven constructively. An example of such a result is Gleason’s theorem, which lies at the foundations of quantum mechanics. It also has deep physical significance, as it rules out a certain class of hidden variable theories. As was shown by Hellman in 1993 \cite{hellman1993}, Gleason’s theorem is not constructively provable; in fact, it implies LLPO \cite{bridges1999} (a principle which, similarly to LPO, does not hold constructively, and is even false in intuitionistic mathematics; see section~\ref{sec:NN}).\footnote{Gleason concluded that this was a profound shortcoming of constructive mathematics: ``The work of Bishop and others \textelp{} can be said to have breathed new life into constructivist mathematics: it shows that a great deal of applicable mathematics can indeed be constructivized. A great deal, however, is not all, and, if our assessment is sound, it is in any case not enough.’’} However, as with many classical analytic theorems, alternative formulations which are classically equivalent to Gleason's theorem \emph{can} be proven constructively. Helen Billinge \cite{billinge1997} proved a number of such alternatives to Gleason’s theorem. In addition, Fred Richman and Douglas Bridges noted that Gleason’s theorem as formulated in Hellman’s 1993 paper was classically but not constructively equivalent to Gleason’s original formulation, upon which Richman proved that Gleason’s original formulation is in fact constructively provable \cite{richman1999}.\footnote{There were similar discussions about unbounded and uncomputable operators (see \textcite{bridges1999} for an overview) and the singularity theorems of Hawking and Penrose, which also have profound physical significance because they prove the big bang (from assumptions) \cite{hellman1998}. See also \cite[section I]{cattaneo1995}.}

More generally, the discussion of whether constructive mathematics has the required technical capacity to formulate modern mathematical theories of physics seems inconclusive: many classical theorems remain constructively unproven,\footnote{\label{fn:weyl}Hermann Weyl, in whose opinion ``it is the function of mathematics to be at the service of the natural sciences’’ \cite[p61]{weyl1949}, was initially fascinated by Brouwer’s intuitionism but later realised its impractical nature: ``Mathematics with Brouwer gains its highest intuitive clarity. \textelp{} It cannot be denied, however, that in advancing to higher and more general theories the inapplicability of the simple laws of classical logic eventually results in an almost unbearable awkwardness. And the mathematician watches with pain the greater part of his towering edifice which he believed to be built of concrete blocks dissolve into mist before his eyes.’’ \cite[p54]{weyl1949} 
This was before Bishop published his \emph{Foundations of Constructive Analysis}.} but history (and, in particular, Bishop’s \emph{Foundations of Constructive Analysis} \cite{bishop1967}) have shown that many useful and classically equivalent alternatives to these theorems can be proven constructively.\footnote{Constructive mathematics is indeed more versatile than classical mathematics, in the sense that it distinguishes between statements that are classically equivalent (such as classical and approximate variants of theorems like the intermediate value theorem discussed in section~\ref{sec:int-real-functions}).}\textsuperscript{,}\footnote{Note that the discussion here is on constructive mathematics, not intuitionistic mathematics; the case for intuitionistic mathematics is more sophisticated as it also has theorems that do not hold classically.}

An interesting note in the discussion on constructivism is that while classical mathematicians usually see the restrictiveness of intuitionistic logic as a weakness, constructive mathematicians tend to stress that their theories admit more models. That is, every theorem proved in Bishop’s constructive mathematics (i.e.~proved with intuitionistic logic) is also true within a plethora of other theories, such as recursive function theory, Weihrauch’s type-two effectivity theory (Appendix~\ref{sec:computable-analysis-preliminaries}), intuitionistic mathematics and classical mathematics \cite{bridges1999}. Hence, in this sense constructive mathematics is not more restrictive at all; instead, it leaves open more possibilities than classical mathematics. However, it is questionable whether this property is also to the benefit of physical theories, rather than purely mathematical and recursive theories such as the ones listed above.

\

Finally, an observation sometimes made by constructivists is that in a way, physicists already use constructive mathematics in everyday practice without knowing it \cite{bauer2013intuitionistic}. They refer to the theory of Smooth Infinitesimal Analysis (SIA), which formalises the use of infinitesimals in analysis (replacing the $\e,\d$-formalism introduced by Cauchy and Weierstraß). Here a quantity $dx$ is called an infinitesimal (of second degree) if its square $(dx)^2$ is zero. While in classical mathematics all such infinitesimals would be zero, using intuitionistic logic we can only prove that $\neg\neg (dx = 0)$ for any infinitesimal $dx$, that is, they are \emph{potentially} zero.
An axiom of SIA which physicists tend to use on a daily basis is the \emph{principle of micro-affinity}, stating that an infinitesimal change in the variable of a function causes a linear change in the value of the function.

However, the simple observation that physicists sometimes use infinitesimals in their derivations cannot be used as an argument in favour of using constructive mathematics for physics, for these infinitesimals could easily be replaced by $\e,\d$-arguments in all cases, as Hilbert \cite{hilbert1926} has already argued. (Alternatively, classical theories of infinitesimal analysis can be constructed by using nonstandard models of the real numbers \cite{goldblatt2012lectures}.) Furthermore, to decide which mathematics to use for physics, we should not focus on the everyday work of the typical physicist, but instead on the more formal formulations of physical theories.

\subsection{Physical motivation}
While the arguments discussed above mostly focus on the mathematical practices of the physicist, the philosophical aspect of physical theories should also be taken into consideration. When doing this, the question arises which type of mathematics has the largest representational capacity, i.e.\ the largest capacity to represent physical reality. Attempts to answer this question can be divided into roughly two camps: those that are concerned with \emph{epistemic} physical reality and with \emph{ontological} reality, respectively.

\subsubsection{Epistemic perspective}
Some convincing arguments can be given that seem to show that constructivism reflects our knowledge about the physical world better than nonconstructivism. The general invalidity of the principle of the excluded middle (PEM) in constructive mathematics, for instance, could reflect our uncertainty about the existence of physical objects or values of physical quantities.

Consider the statement $\phi$ given by\footnote{Example taken from \textcite{bauer2008blog}.}
\[ \phi = \text{`There exists a particle which does not interact with anything in the Universe.’} \]

The statement cannot be falsified, viz.\ negative evidence would lead to a contradiction, which means that $\neg\neg\phi$ holds. On the other hand, no evidence of the particle exists, nor can it ever be given, which means we will never be able to conclude $\phi$. The distinction between $\neg\neg\phi$ and $\phi$, which seems important to an epistemic approach to physics, can only be made by using constructive mathematics.

The argument for constructivism is particularly strong from the stance of operationalism.\footnote{See footnote~\ref{fn:operationalism} on page~\pageref{fn:operationalism}.} 
See \textcite{cattaneo1995} for an example. The same can be said about intuitionism: \textcite{bauer2008blog}, for instance, gives an operationalist motivation for Brouwer’s continuity principle in physics, based on the fact that we can only use a finite amount of resources to arrive at a conclusion about the physical world (we will see in section~\ref{sec:intuitionistic-reals-physics} that there is also an ontological motivation for the continuity principle, when accepting a finite-precision interpretation).

In these epistemic approaches, the existential quantifier is probably interpreted as saying that ‘we have physical (positive) evidence’, or at least that `we must be able to obtain physical evidence using a finite amount of resources'. This sounds much like the constructive (BHK) interpretation of $\exists$ in mathematics (section~\ref{sec:constructive-mathematics}). However, as a modest scientific realist myself, I think we should at least be able to consider the possibility that a theory represents physical reality, and using constructive mathematics for physics on purely epistemic grounds either deprives us of that possibility, or, in the best case, it leaves open the question whether ontological theories should follow the same approach and adopt constructive mathematics as their foundation.

\subsubsection{Ontological perspective}\label{sec:ontological-perspective}
PEM can also be rejected from an ontological perspective of physics, in this case not because of our lack of knowledge about a physical system, but because of either physical indeterminism (for statements about the future) or indeterminacy (for statements about the present). Indeterminacies can arise from, for example, superpositions in quantum mechanics (‘Is Schrödinger’s cat alive?’) \cite{landsman2017} or from the (ontological) property of finite precision of physical quantities in Gisin’s alternative classical mechanics---so that statements about, for example, the relation between two physical quantities need not have a well-determined truth value. Gisin himself has proposed that intuitionism might indeed better reflect physical ontological reality \cite{gisin2020comment}. Indeterminism of the physical world is also closely related to PEM, because the truth of statements about the future need not be decided in the present.\footnote{One might argue that PEM should be associated with \emph{fatalism} rather than determinism. We will not make this subtle distinction here.}\textsuperscript{,}\footnote{The relation between PEM and statements about the future was already recognised by Aristotle in \emph{De Interpretatione}.} Intuitionistic mathematics in particular could be seen as a promising candidate to formulate indeterministic theories of physics, because of the central importance given to the notion of time (see section~\ref{sec:intuitionism-time}).\footnote{See \textcite{trzesicki1994} for some connections between intuitionism and indeterminism in terms of tense logic.} It is tempting to somehow ‘couple’ the physical and intuitionistic notions of time. We will investigate an approach which is based on this idea in the following sections. However, the fact that intuitionistic time is an `intuitive' or `subjective' notion while time in an ontological theory of physics should ideally be objective indicates that these concepts of time are fundamentally different and that any attempt to equate them might be fruitless. This is discussed in more detail in sections~\ref{sec:intuitionism-time} and~\ref{sec:equating-time}.

A similar problem emerges when using \emph{constructive} mathematics in general (not intuitionism specifically) to describe an ontological theory of physics. Namely, instead of a clash between `intuitive' and `objective' time, we see a clash between constructivism and scientific realism. From a realistic perspective, for example, it should (sometimes) be possible to say that something (physical) either exists or does not exist, independent of our knowledge---the statement $\phi$ from the previous section, for example, should satisfy $\phi\lor\neg\phi$. This is most likely the reason that the ontological motivation for using \emph{constructive} mathematics in physics is heavily underrepresented in the literature compared to the technical and epistemic considerations.
Perhaps one would ideally want to have a theory in which PEM does not hold for certain statements about the future (in order to express indeterminism), but does hold for those statements for which PEM is physically justified. Because the latter class contains statements, such as $\phi$, which are not constructively decidable, it seems unlikely that constructive mathematics can be consistently used to describe such a theory.

Indeed, it might be that in using constructive mathematics one is always forced to have an empirical attitude to physics. And trying, alternatively, to formulate yet another mathematics which respects the different physical motivations for decidability would be cumbersome, if not impossible.\footnote{Such an approach would be considerably complicated by the fact that it is not always clear whether a mathematical statement represents a statement about physical reality or not.} A similar conclusion is drawn by \textcite{hellman1998}:
\begin{quote}
	I believe \textelp{} that radical constructivism which grants the coherence of classical reasoning about the physical is inherently unstable.
\end{quote}
However, Hellman does not discuss the possibility that also an ontological interpretation can have aspects, like indeterminism or indeterminacy, for which classical principles like PEM are not desirable. In section~\ref{sec:intuitionistic-reals-physics}, we will see that the intuitionistic philosophy of the continuum and time can at least \emph{inspire} us to change our view of physical quantities. In section~\ref{sec:doubts-again} we revisit the incompatibility of intuitionistic and physical time, as well as the status of PEM in physics.

\subsection{What does physics say about constructivism in pure mathematics?}
Until now, this chapter has been concerned with reasons to use constructive mathematics to formulate physical theories. Can we, however, also reason the other way around? Can physical arguments be used to promote constructivism in pure mathematics?

One point of view, which seems to be taken by Hellman \cite[p202]{hellman1993} and perhaps also Weyl (see footnote~\ref{fn:weyl}), is that the main purpose of mathematics lies in its application to the natural sciences. This would imply that motivations for using constructive mathematics in physics should always take into account the representational capacity of mathematical theories. While this is true in part for fields such as mathematical physics, it does not hold for mathematics in its purest form. Although important research directions in mathematics like analysis have been inspired by the physical sciences, they have come to lead a life of their own and can now be considered separately from the physical sciences. Therefore, the motivation for using constructivism in pure mathematics should come from a different angle than in the case of mathematical physics. This has historically indeed been the case, as the development of constructive and intuitionistic mathematics was mostly carried out without physical applications in mind.\footnote{This holds to a lesser extent for Brouwer, who used some principles of physics in his thesis on intuitionistic mathematics from 1907 \cite{brouwer1907}, with which he obtained a PhD title in both mathematics and physics.}\textsuperscript{,}\footnote{However, the fact that the line between pure mathematics and physics is often blurred and that there is a constant exchange of ideas between the two subjects definitely complicates matters and, taken together with the fact that the overwhelming majority of mathematics developed during the previous century is classical, is one of the reasons that attempts to constructivise physics have not been taken very seriously by the wider scientific community.}

That is not to say that this different angle has nothing to do with physics at all: indeed, a significant factor in the motivation for constructivism in pure mathematics, and especially intuitionism, is the human condition, e.g.~the limitations on our ability to mentally construct mathematical objects,\footnote{Oscar Becker even notes that ``the very possibility of intuiting mathematical objects rests on our ability to realise them concretely in the world'' \cite{roubach2005being}.} and the fact that we can only perform finitely many operations per unit time. These arguments find their origin in the physical realm---but they are nevertheless manifestly different from the arguments based on representational capacity discussed above.

So what is the human condition, and how does the BHK interpretation follow from it? One way to talk about the BHK interpretation in a more formal way is in the framework of \emph{realisability}: In realisability, it is made explicit what it means that a statement can be proven or a mathematical object can be constructed. The notion of ‘proof’ of a statement $\phi$ is replaced by the notion of a \emph{realiser} $r$ for $\phi$, written as $r \Vdash \phi$. $r$ can be thought of as a computer program or other method witnessing the truth of $\phi$. In this way, constructive mathematics can be brought into relation with various models of computation. Usually these are mathematical models; \textcite{bauer2013intuitionistic}, however, explores the possibility of using physical processes or objects as realisers.

His approach is closely connected to the so-called \emph{physical Church-Turing thesis}. This is the statement that Turing’s mathematical notion of computation captures precisely those computations that can be performed using physical devices.\footnote{This thesis differs from the more widely-known \emph{Church-Turing thesis}, which states that Turing’s computability notion (as well as most other computability notions such as recursion theory, which are equivalent to Turing's) captures precisely the computations that can be performed by an (idealised) mathematician using only a pen and (an arbitrary amount of) paper. See also section~\ref{sec:church-turing-computability}.} Evidence against this thesis has been given from different fields of physics: some suggest that quantum computers may be able to venture ‘beyond the Turing barrier’ \cite{arrighi2012physical}, and a general relativistic model of ‘hypercomputation’ has been proposed that makes use of black holes to perform supertasks \cite{andreka2009hypercomputing}, which would mean it would be possible to decide the halting problem.

Bauer uses an informal notion of ‘realisability in the physical world’ to arrive at conclusions such as ‘Decidability of reals is real-world realized iff we can solve the halting problem’. More generally, this would mean that it might be possible to ‘constructively’ decide LPO, if a physical notion of constructivism is adopted and e.g.\ hypercomputation is indeed physically possible. Bauer also gives conditions on when Brouwer’s continuity principle is real-world realised \cite{bauer2013intuitionistic}.\footnote{This is closely related to the debate on the meaning of ‘lawlike’: while some constructivists accept Church’s thesis, which states that lawlike is synonymous to computable in Turing’s sense, others suggest that a broader notion should be adopted, which captures all processes that can be physically computed by humans, or even laws which depend on, but are fully determined by, \emph{indeterministic} processes (like BKS sequences, discussed in section~\ref{sec:lawlike-sequences}).}

Although this approach is not formal and is closely related to unsolved problems in physics, it is interesting to speculate how physics can in fact influence the constructivism debate in pure mathematics. It is important to note, however, that these results only impact decidability of equality on \emph{mathematically} defined reals; they cannot be used for the discussion of the previous subsections, as they have no influence on, for example, the determinacy of real numbers \emph{representing physical quantities}.

Finally, whatever the reasons for adopting constructivism, the fact remains that constructive, intuitionistic and classical mathematics are all (as far as we know) consistent mathematical frameworks. Hence, the foundational `debate' is about not much more than a difference in interpretation or personal preference---which is perfectly acceptable in pure mathematics, but when applied to physics could bring along a difference in representational capacity for the objective world. As we have seen in the preceding sections, whether constructive mathematics can be successfully used to this end remains questionable. In section~\ref{sec:doubts-again}, we will revisit this question for intuitionistic mathematics specifically.

\section{Intuitionistic reals in classical physics}\label{sec:intuitionistic-reals-physics}

We will now investigate what intuitively happens when we apply some of the \emph{ideas of} the intuitionistic philosophy of the continuum and of lawless sequences to physics. \Textcite{gisin2020comment} has suggested that this could introduce a notion of ‘creative time’ to physics and that it could provide a solution to the problems introduced in Chapter~\ref{chap:whats-the-problem}, in a way similar to his own alternative formulation of classical mechanics, discussed in Chapter~\ref{chap:gisin}.

Note that Del Santo and Gisin's theory is not intuitionistic: for example, it assumes the existence of fully determined, `finished' infinite sequences of propensities, contrary to the intuitionistic view of infinite sequences. Moreover, it makes use of probabilities, which are not involved in the intuitionistic definition of real numbers. Finally, we have seen that using the binary expansion to define reals is not possible intuitionistically (section~\ref{sec:int-reals-other-constructions}, page~\pageref{sec:int-reals-other-constructions}). Motivated by these considerations, we will instead use the intuitionistic definition of real numbers based on infinite sequences of shrinking and dwindling rational intervals (see section~\ref{sec:int-reals}), of which at each point in time only a finite initial segment is known.

However, it will turn out that in many respects, also this approach does not conform to the intuitionistic spirit, in part because of the challenges discussed in section~\ref{sec:ontological-perspective}, and that it is also very naive from a physical perspective (see sections~\ref{sec:equating-time} and~\ref{sec:doubts-again} respectively). However, intuitionistic reals provide an intuitive way to think of finite-precision quantities, based on geometrical aspects instead of binary expansions, and they set the stage for a more consistent formulation, which we will present in Chapter~\ref{chap:further-development}. In this section we focus on some conceptual consequences only.

We focus on the Hamiltonian formulation of classical point mechanics. The idea is to concretise the intuitionistic notion of a course of time using the time parameter of physics. Using intuitionistic reals to describe physical quantities in this way would mean that at each moment in time, each (one-dimensional) physical quantity is only determined up to a rational interval of nonzero length.\footnote{\Textcite[p435]{hellman1998} finds that a constructive ‘rule for generation’ (a sequence of rationals approaching a real) should not be identified with a physical object (the real). Here, however, we do exactly that.} Once again, this is meant in an ontological way: Nature has simply not yet revealed a more specific value for the physical quantity. As time passes, the rational interval can at some points suddenly shrink, which we sometimes refer to by \defn{collapse} of the interval, meaning that the indeterminacy of the value of the physical quantity decreases. In the terminology introduced in section~\ref{sec:heraclitus-parmenides}, this manifests a step in Heraclitus time. In contrast to the orthodox interpretation of classical mechanics, where all information about the physical quantities is already present in the initial conditions, in this interpretation, new information about the real numbers is generated as time passes.

The rational intervals associated with physical quantities take on a quite important role in the ontology of our new physics,\footnote{This is a problematic fact, just as the use of rational numbers in \textcite{dSG19} (see section~\ref{sec:gisin-reaction}); we will ignore this problem for now and focus on the conceptual consequences only, but we return to the matter in section~\ref{sec:technical-problems}.} but their role can also be generalised to other formulations of finite-precision theories (such as the one given in Chapter~\ref{chap:further-development}). We use the general term \defn{domain of indeterminacy (DOI)} to refer to the set of real numbers (or rather: the region of the continuum) that the physical quantity is determined to lie in at some time $n$ (in Heraclitus time). In this section, we denote the DOI of a physical quantity $x$ at time $n$ by $x_n$, and $x_n$ is a rational interval. In the simple case that $x$ is not subject to Hamiltonian evolution (i.e.\ does not evolve over Parmenides time), DOIs are subject to the requirement that $x_{n+1} \subseteq x_n$ (a more rigorous definition will have to wait until Chapter~\ref{chap:further-development}).

\

Let us investigate how some fundamental results from intuitionistic mathematics translate to physics. As an example, consider a quantity $x$ that represents the position of a stationary particle in a one-dimensional system, such that by the above requirements, $(x_0,x_1,\ldots)$ is a sequence of shrinking rational intervals. If $n$ is such that $x_n$ contains zero, then the statement
\[ x > 0 \]
is `vermetel' or `reckless' (in Brouwer's terminology; see Appendix~\ref{app:intuitionism}) at time $n$, which is to say: the statement ``the particle is to the right of zero’’ does not have a well-defined truth value at time $n$.\footnote{In Kreisel’s formalism of the creating subject, this reads $\neg\,\square_n (x > 0)$.} The statement can, however, become either true or false in the future.\footnote{\label{fn:triangle}This example is perhaps too simple, because it is not invariant under shifts and reflections of the coordinate system, so it may be difficult to say that it can have an inherent truth value at all. As an alternative, consider three particles in two- or three-dimensional space, and the statement ‘the particles form the vertices of an obtuse triangle’.}

The physical time parameter itself also becomes an inherently imprecise notion. Consider a particle that moves freely and bounces off a wall at time $t_0$. Its position is given by
\[ x(t) = \begin{cases}
	y(t) & \text{ if } t \leq t_0, \\
	z(t) & \text{ if } t \geq t_0,
\end{cases} \]
for some functions $y,z$ with $y(t_0) = z(t_0)$ (example taken from \textcite{bauer2013intuitionistic}).\footnote{This example is again naive, because the momentum of the particle is discontinuous at $t_0$, which is only possible intuitionistically if we accept that the particle's momentum is a partial function of time (i.e.\ not defined for all times $t$). Indeed, it is questionable whether elastic collisions are possible in a finite-precision physics; this agrees corresponds to the observation that in regular physics, elastic collisions are often seen as only \emph{limiting cases} of real physical phenomena. As an alternative example, avoiding walls and elastic collisions, consider the situation in footnote~\ref{fn:triangle} with the particles now moving relative to each other, and denote by $t_0$ the time at which the triangle becomes obtuse.} If the initial position of the particle is not infinitely precisely determined, then neither is the time $t_0$. We see that the domain of definition of $x$ is $(-\infty,t_0] \cup [t_0, \infty)$. Interestingly, in intuitionistic mathematics this set is not equal to the whole continuum $\R$, because that would mean that $\forall t\in\R\,(t \leq t_0 \lor t \geq t_0)$, which is not the case in intuitionism (see Equation~\eqref{eq:llpo-reals}). Hence also the time continuum is `viscous’ and cannot be sharply cut into two: this is what Gisin calls `thick time’ \cite{gisin2020comment}.

\

Moreover, the real number version of Brouwer’s continuity principle BCP$_\R$ \eqref{eq:bcp-r} makes sense in finite-precision theories, since any relation between physical quantities and natural numbers that is well-determined at some point in time should depend only on finitely much information about the physical quantity in question, analogously to BCP$_\R$. It also means that any total function that takes only physical quantities as arguments should be continuous, which is analogous to the statement that all real functions in intuitionism are continuous (which is equivalent to BCP$_\R$). Indeed, continuity (and even smoothness) of functions is often implicitly assumed in physical practice.\footnote{The momentum of the particle in the previous example is a function of time that is not continuous but is partial (defined only on $(-\infty,t_0]\cup [t_0,\infty)$).}

As another example, consider the mean value theorem and its constructive approximate variant. Suppose $f$ is a function of space, for which $f(x) = \a$ and $f(y) = \b$, where $x\apart y$ and $\a,\b$ are physical quantities with $\a \apart \b$ (that is $\neg(\a_n \approx_\S \b_n)$, where $n$ is the current time, in the notation of Appendix~\ref{app:intuitionism}). Because $\a$ and $\b$ are inherently determined only up to a finite precision, we see that the classical mean value theorem does not hold: there is no $z$ such that $x < z < y$ and $f(z) = \frac{\a+\b}{2}$ with certainty. The constructive approximate variant is valid, however, in the following way: for any $m\in\N$, if we wait sufficiently long, the indeterminacy in $\a,\b$ and $f$ will become sufficiently small to pinpoint a $z$ such that $|f(z) - \frac{\a+\b}{2}| < \frac{1}{m}$ with certainty (again, the present approach limits us to speaking in informal terms; see Chapter~\ref{chap:further-development} for a more formal approach).

Finally, using the intuitionistic philosophy of the reals sheds light on the nature of collisions in a finite-precision point-particle\footnote{Point-particle seems to suggest that the location of each particle is known to infinite precision, but all that I mean here is that the particles have zero diameter, so that the positions can be determined to \emph{arbitrary} accuracy.} physics. Consider a system of two particles moving in three-dimensional space approaching each other along the $z$-axis,\footnote{That is, \emph{more or less} along the $z$-axis, as their direction of movement is not defined with infinite precision.} each of the particles' position being associated a certain domain of determinacy. When the distance between the particles decreases and becomes on the order of magnitude of the indeterminacy in the particles' positions, then the positions become more physically relevant (as is certainly the case if, for instance, the particles exchange an inverse-square force). This causes (one of the) DOIs to collapse to a smaller interval, upon which the particles approach each other closer. This process continues until the particles pass each other; however, if collapse of the DOIs is indeed indeterministic, then the particles avoid a collision and thus pass by each other \emph{with probability 1}, as the probability is zero that the DOIs pertaining to their $x$- and $y$-positions converge to the same point. In orthodox point mechanics, on the other hand, there are initial conditions which are guaranteed to lead to point-particle collisions.

\

We see that using the intuitionistic philosophy of real numbers to describe physical quantities intuitively provides an alternative approach to a finite-precision theory of classical mechanics, and that doing so reveals many interesting features of this finite-precision theory. Moreover, in contrast to Gisin and Del Santo's theory, the approach given above does not depend on the base 2 representation of real numbers, which solves one of the problems discussed in section~\ref{sec:gisin-reaction}. As we will see in section~\ref{sec:technical-problems}, however, the approach sketched above still suffers from some major technical problems. First, however, let us make some more remarks on the interpretative consequences of using intuitionism for physics.

    

	
	
	
	
	
	

\section{Problems of the intuitionistic approach}\label{sec:doubts-again}
In section~\ref{sec:constructivising-physics} we already concluded that a constructive approach to physics might not be compatible with an ontological (realist) worldview. In section~ \ref{sec:intuitionistic-reals-physics} and in \textcite{gisin2020comment}, on the other hand, we have seen that the way time is treated in intuitionism might better represent the flow of time in physics than classical mathematics does. This suggests to ‘couple’ the intuitionistic and physical notions of time. As we will see below, however, such an attempt cannot be taken too seriously.

\subsection{Equating intuitionistic and physical time}\label{sec:equating-time}
The main reason for this is that time forms an important aspect of intuitionistic mathematics precisely because intuitionism is formulated from the perspective of the (idealised) practising mathematician, or, in Brouwer’s terminology, the \emph{creating subject} (see also section~\ref{sec:intuitionism-time}). Intuitionistic time refers to the time experienced by the creating subject, and can therefore be seen as subjective, whereas time in an (ontological) theory of physics should, ideally, be an objective\footnote{Though not necessarily global or absolute (considering relativity theory).} phenomenon.

As discussed in section~\ref{sec:intuitionism-time}, time in intuitionism relates, in the first place, to (i) the time necessary (for the idealised mathematician) to prove a theorem or to perform computable operations. In particular, the creating subject can only perform finitely many computations per unit of intuitionistic time. Apart from that, time appears (ii) in the concept of choice sequences, which are defined in a `step-by-step manner’ (such that the amount of new information considered or `revealed' per unit of intuitionistic time is finite).

Therefore, if intuitionistic time were literally `equated’ to physical time, then not only would it (ii) take time for physical quantities to become more precisely determined, but also (i) trivial calculations like Parmenides time evolution (which involves computing the solution to Hamilton’s equations) could only be performed with finitely many computations at a time. This would mean that even the finite-precision state of the Universe would only be determined (`exist’) at a discrete set of timepoints,\footnote{Namely, we cannot let the creating subject calculate the state of the Universe at every time $t\in\R$, nor at every time $t\in\Q$, because both options would require infinitely much computation per unit time, even if the state at every time is only determined with finite precision.} which is a profound and undesirable consequence.\footnote{For example, if (Parmenides) time is discrete, then the system cannot always `sense’ when a quantity becomes physically relevant and thus when its domain of indeterminacy should collapse. While there are indeed theories that assume that time is discrete, in the current setting, discrete time is not desirable.}

Indeed, for our ontological finite-precision theory of physics, \emph{only the lawless aspects} of intuitionism are relevant, which correspond to the generation of truly new information and hence to indeterminism, rather than the time needed to perform trivial computations of which the outcome is fixed given the present configuration of the Universe.
Because it is difficult (or as some would argue, impossible) to separate lawlikeness and lawlessness in intuitionism (see section~\ref{sec:separating-lawlike-lawless}), and because both notions are subject to the same `intuitive time' experienced by the creating subject, it is questionable whether the philosophy behind finite-precision physics can truly be called intuitionistic. This is in line with the conclusion of section~\ref{sec:ontological-perspective}; the above arguments also apply to the debate of using intuitionism or constructivism for indeterministic theories in general.

\

To support 
these arguments, let me make some more remarks relating to the principle of the excluded middle. First of all, the motivation to use intuitionism for indeterministic theories because the invalidity of the principle of the excluded middle expresses indeterminism well pales by the observation that lawless aspects are not necessary to see that PEM is intuitionistically invalid. Take for example the statement $G$ defined in section~\ref{sec:NN}; it does not satisfy $G\lor\neg G$, because the Goldbach conjecture has not been proven nor disproven; but the truth value of $G$ is, in a certain sense, still fixed.\footnote{This example might actually not say very much about physics, since its truth value is not directly related to a statement about physical objects; however, it is in my opinion self-evident that also truth values of statements about physical objects which are fixed by the current state of the Universe should not depend on some `creating subject’ having proven or disproven the statement.} So in general, the mathematical future (of the creating subject) is open for a different reason than why the future of the Universe is open in an indeterministic theory of physics.\footnote{Indeed, although the admission of lawless sequences distinguishes intuitionism from other forms of constructive mathematics, lawless sequences are not the most important part of the intuitionistic philosophy; Brouwer himself, for instance, did not accept lawless sequences at first, and even when he did, he never used an explicit example of a lawless sequence in his proofs. Lawless sequences were only explicitly defined by Kreisel (see section~\ref{sec:lawless-kreisel-troelstra}).}

Secondly, the present-day formulations of indeterministic theories (in classical mathematics, that is) show that the invalidity of PEM is not \emph{necessary} to express indeterminism. Instead of saying that certain statements about physical objects do not have a well-determined truth-value, we can also say that these statements are ill-defined themselves. For example, when two domains of indeterminacy (DOIs, recall the terminology of section~\ref{sec:intuitionistic-reals-physics}), pertaining to the $x$-coordinates of two particles, overlap, then the statement ``particle 1 is to the left of particle 2’’ can be seen as ill-defined. In this respect, therefore, theories formulated in classical logic might be equally able to represent indeterministic physics as theories formulated in intuitionistic logic.

A constructivist could still argue that it is \emph{easier} to use intuitionistic logic, because then one need not worry about what statements are well-defined and which are not. While this reasoning is sometimes used in the context of computability theory (see \textcite[p441]{bridges1999}), it is questionable whether it can be applied to physics, considering the cases where PEM holds physically but not constructively, as discussed previously.

\subsection{Technical problems}\label{sec:technical-problems}
Apart from the conceptual issues of the previous section, using shrinking sequences of rational intervals as proposed in section~\ref{sec:intuitionistic-reals-physics} brings along severe technical problems from the physical side. These are similar to the problems of Del Santo and Gisin’s approach discussed in section~\ref{sec:gisin-reaction}.

First of all, rationality of a number is conserved neither under many reasonable coordinate transformations, nor under Hamiltonian time evolution (even for rational lapses of time). Hence, rational intervals do in general not transform to rational intervals, and moreover, hyperrectangles representing the DOIs of multidimensional quantities do not transform into hyperrectangles. This is a problem if we want to assign ontological meaning to these rational intervals and rectangles.

Moreover, the problem of section~\ref{sec:gisin-reaction-hamiltonian-evolution} remains, i.e.\ it is unclear how Heraclitus time and Parmenides time are linked. One of the reasons for this is precisely that the set of rational intervals is not invariant under Hamiltonian time evolution.

Finally, it is still unclear in what way the processes of collapse of DOIs of different physical quantities depend on each other. Two quantities that are inextricably linked, like the velocity and position of a particle, can in general not both collapse randomly and independently, because this could lead to physically impossible world lines.

In the next chapter, we will solve these issues by using recursively enumerable open sets instead of rational intervals and by considering regions of phase space instead of DOIs of individual quantities.

\section{Lawless sequences and indeterminism}\label{sec:lawless-indeterminism}
As indicated above, the most important aspect of the intuitionistic philosophy that might be of use to finite-precision physics, or perhaps to any indeterministic theory of physics, is formed by choice sequences, and in particular their \emph{lawless} aspects. In this section, we focus on the formalisation and philosophy of lawless sequences proposed by Kreisel and Troelstra, summarised in section~\ref{sec:lawless-kreisel-troelstra}, and relate it to indeterminism in physics. In particular, we will try to investigate whether or not lawlessness follows automatically from the fact that the laws of an indeterministic physical theory admit multiple possible futures.\footnote{Unfortunately, to assume the well-posedness of this question, we have to ignore the strong arguments against the meaningfulness of Kreisel and Troelstra's definition and against the well-definedness of `lawless sequence' in general (see section~\ref{sec:separating-lawlike-lawless}). 
A promising idea for further research is to focus on extensional notions of lawlessness in intuitionism instead (section~\ref{sec:lawless-extensional}).}

One thing to get out of the way is that physical indeterminism of course does not imply \emph{absolute} lawlessness: indeterministic theories like quantum mechanics and finite-precision classical physics have laws of motion, but these do not \emph{completely} fix the future given the present.\footnote{As Popper indicates, “Indeterminism merely asserts that there exists at least one event (or perhaps, one kind of events \textelp{}) which is not predetermined” \cite{popper1950indeterminism}.} Moreover, even wave function collapse in quantum mechanics is not completely lawless, as it subject to \emph{probabilistic laws} (which lead to e.g.\ the law of large numbers). Unfortunately (from this point of view), however, probabilities play no role in the intuitionistic concepts of choice and lawless sequences.

In this section, therefore, we will only consider the simple case of an indeterministic theory which does not have probabilistic laws. We will further assume, for simplicity, that we can `extract’ one \emph{completely} indeterministic process from the evolution of a physical system, parameterisable by an infinite sequence of natural numbers which is not described by any (probabilistic) law (except perhaps a spread law, see section~\ref{sec:spreads}) and is independent of other physical processes; and which together with the physical laws completely determines the evolution of the system. Let us call such a sequence an \defn{evolution sequence} of a world. (We might associate this sequence with the sequence of throws of ‘God’s die’, but only if we divest this die of all probabilistic properties that dice are normally associated with.)\footnote{Asserting the existence of such a sequence does not refute indeterminism; remember that we view infinite sequences as projects that develop in time.} An example of such a sequence in finite-precision classical physics could be a sequence of natural numbers, each element describing the DOIs of all physical quantities.

\

If we assume the existence of such an ‘evolution sequence’, then does it make sense to impose the condition that it must be a \emph{lawless} sequence? In particular, let us consider the following two types of laws of physics:
\begin{itemize}
    \item[(a)] “Evolution of nature is given by (differential equations and) an (arbitrary) sequence”;
    \item[(b)] “Evolution of nature is given by (differential equations and) a \emph{lawless} sequence”.
\end{itemize}
The question arises: \textit{Are these formulations of the physical laws equivalent, and if not, how do they relate to each other}? Let us explore two possible views on this question, one of an `orthodox physicist’, who is used to using classical (nonconstructive) mathematical reasoning, and one of a ‘time extremist’, who thinks, analogously to the intuitionistic mathematician, that we should see physical time as a forever unfinished process.\footnote{In literature on philosophy of mathematics, one often finds discussions between ‘the classical mathematician’ and ‘the constructive mathematician’, but the present discussion is on philosophy of physics, and it would be inaccurate to call the time extremist a ‘constructivist’.}

Being asked the italicised question, the orthodox physicist might immediately respond that ``these two formulations are clearly not equivalent: because I can describe a world which has a lawlike evolution sequence. This world is possible according to the laws specified in (a), but not according to the laws specified in (b). On the other hand, every world that satisfies (b) also satisfies (a); hence (b) is strictly stronger than (a).''

This view follows from a common philosophical way to talk about physical laws, where laws are related to the set of \defn{possible worlds} at which they obtain---see section~\ref{sec:determinism}. In our simplified case, we can identify the set of possible worlds $\mathcal W$ with the set of all sequences of natural numbers (when (a) is adopted) or the set of all \emph{lawless} sequences of natural numbers (when (b) is adopted). For a world $W\in\mathcal W$, we denote by $W_t$ the $t$-th element of the evolution sequence of $W$, where $t\in\N$. Using this notation, we employ Earman's definition of determinism as given in Equation~\eqref{eq:earman-determinism}.

Note that the world $W$ described by the orthodox physicist, namely one given by a \emph{lawlike} evolution sequence, is not deterministic, as for determinism to hold at $W$, the future should be fixed given only a time slice $W_t$, which does not include information about the specific lawlike sequence that represents $W$.

\

Reasoning in terms of a (possibly infinite) set of \emph{finished} possible worlds, such as the orthodox physicist does, requires a quite platonistic view. The time extremist would argue that we should not, or only under very limited circumstances, talk about finished worlds. They would therefore give a completely different answer to the italicised question, drawing inspiration from Kreisel and Troelstra’s philosophy of lawless sequences: “Suppose that the physical laws are given by (a), and consider the sequence $\gamma$ defined as the evolution sequence of the actual world, i.e.~\emph{the sequence that my world experiences}. $\gamma$ is defined intensionally and is forever unfinished, and because formulation (a) imposes \emph{no} law on the continuation of $\gamma$ \emph{at any stage} of its development, the sequence $\gamma$ is a lawless sequence, in Kreisel and Troelstra’s sense. In other words, the absence of a specified law on the sequence in (a) makes the definition equivalent to definition (b). This means that (a) being obtained at the actual world is equivalent to (b) being obtained at the actual world.”

Confused, the orthodox physicist replies: “You are forgetting that Troelstra and Kreisel’s axioms are actually about the \emph{set of} lawless sequences and the relations between them; indeed to make sense of indeterminism, we have to consider not only the actual world, but also other possible futures. Even if we are constrained to being in a time slice of the actual world, we must at least be able to \emph{talk} about those other worlds. How would that match with your view?”

The time extremist might compromise: “I can still define other sequences $\gamma'$ intensionally by letting them denote the time evolution of given possible worlds, which are of course forever unfinished. In this way we are able to compare different worlds, and if the laws of nature are as formulated in (a), then the same argument as I used earlier can be used to justify that those worlds satisfy (b) as well.”

Here, however, the time extremist’s view deviates much from the traditional possible worlds view, because possible worlds are usually seen as finished objects. Indeterminism in physics is usually interpreted in the sense that one cannot infer which possible world is the actual world on the basis of only one timeslice of the actual world. This is a purely \emph{extensional} consideration (i.e.~based solely on the present \emph{configuration} of the actual world and the future \emph{configurations} of possible worlds, analogous to the use of the term \emph{extensional} in mathematics; see section~\ref{sec:lawless-kreisel-troelstra}). If we do want to speak about indeterminism using unfinished worlds in the intuitionistic sense, therefore, we have to relinquish this traditional possible worlds view.

\

Nevertheless, when we define the physical laws as in (b), Kreisel and Troelstra’s lawless sequences \emph{prima facie} seem to provide a good alternative approach to possible worlds, because when the set $\mathcal W$ of possible worlds contains only \emph{unfinished}, \emph{lawless} evolution sequences, then one can constructively prove from LS1~\eqref{eq:ls1} a positive version of Earman’s indeterminism~\eqref{eq:earman-determinism}. Namely, let us define \defn{positive indeterminism at} $W\in\mathcal W$ by
\[ \exists W'\in \mathcal W\,\exists t\,(W_t = W'_t \land W \apart W') \]
(where $W\apart W'$ means $\exists s\,(W_s \neq W'_s)$), i.e.\ we can find a possible world $W'$ and a timepoint $s$ such that $W'$ agrees with $W$ at time $t$ but is distinct from $W$ at $s$.
Assuming, now, that the timepoint $t$ lies in the actual world $W$’s past, viz.~$W_t$ is determined, then LS1 lets us find worlds $U, V\in\mathcal W$ such that $W_t = U_t = V_t$, but such that $U_{t+1} \neq V_{t+1}$. Then $U \apart V$ (i.e.~they are \emph{positively} unequal), and by cotransitivity of $\apart$, we have ${W\apart U \lor W\apart V}$. Hence, 
\[ (W_t = U_t \land W \apart U) \lor (W_t = V_t \land W \apart V) ,\]
which constructively proves positive indeterminism as defined above.

\

On a closer look, however, it is questionable whether Kreisel and Troelstra’s formalism provides an appropriate framework to talk about indeterminism. The axiom LS2 (see Equation~\eqref{eq:ls2}) is warranted mathematically by heavily relying on the intensional properties of lawless sequences: the mathematician knows the `identity’ of the sequences $\a$ and $\b$ and can therefore decide with certainty whether they are equal or not. Translated to possible worlds in physics, however, this axiom leaves us with the unusual fact that we can distinguish worlds purely on the basis of some intensional property, not pertaining to their spatio-temporal configuration. The essential point of indeterminism (in the current setting of a set $\mathcal W$ of unfinished possible worlds) is, on the other hand, that we \emph{cannot} distinguish two worlds which have not (yet) diverged.

Also LS3 poses problems to the time extremist’s view, because its consequence, Equation~\eqref{eq:lawless-neq-lawlike}, implies that (a) is not at all equivalent to (b).

Finally, if we adopt the view suggested by (b) that the set of possible worlds is given by Kreisel and Troelstra’s set of lawless sequences, then we have to somehow deal with the fact that this set seems to be language-dependent, which we discussed in section~\ref{sec:lawless-kreisel-troelstra}. In particular, if three evolution sequences $\a,\b,\g$ are related by some lawlike relation, for instance $\forall n(\g(n) = \a(n)+\b(n))$, then according to Troelstra, they cannot all be lawless, but their lawlessness ‘depends on the context’ \cite{troelstra77}. Hence, $\a$, $\b$ and $\g$ cannot all represent possible worlds, but which ones do and which one does not is a question that cannot be objectively answered.

\

Trying to make sense of Kreisel and Troelstra's philosophy of lawless sequences in the context of indeterminism in physics has turned out very difficult, if not impossible.\footnote{Having now fully understood Kreisel and Troelstra's concept and the controversy surrounding it, I think it might not even have been worth the attempt.}
By viewing infinite sequences and time as processes that are never finished, the time extremist’s first instinct was that definitions (a) and (b) of the physical laws are equivalent. Whether this is true is questionable, since we have seen that defining the physical laws as in (b) using Kreisel and Troelstra’s sequences brings along many problems.
Most of these follow from the intensional nature of Kreisel and Troelstra's definition, which already gives rise to much debate within intuitionistic mathematics (see section~\ref{sec:separating-lawlike-lawless}). However, while intensional definitions can be considered acceptable in (intuitionistic) mathematics, using them in physics leads to even more extreme situations. What matters in physics are arguably only the extensional properties of objects, for two worlds which have the same spatio-temporal content should probably be considered identical.\footnote{And finally, of course, the evolution of the Universe is in our view not determined by free choice of an idealised mathematician. One could involve deities or human free choice in discussions about indeterminism, but that is a story for another time.}

If one really prefers using definition (b) over (a), then it would be wise to use an \emph{extensional} definition of lawlessness such as the ones discussed in section~\ref{sec:lawless-extensional}. The relation between these lawlessness definitions and indeterminism in physics, and whether they are satisfied automatically when the physical laws are formulated as in (a), should be investigated in more detail. However, these extensional notions are classical and do there not reflect the way that we would ideally like to view time in physics.

It seems best to formulate the physical laws as in (a) (as we will do in Definition~\ref{def:blurred-world}), and to live with the fact that the actual world might be ‘extensionally lawlike’ when viewed from the end of time. As noted earlier, this does not contradict indeterminism because the law specifying a sequence cannot be inferred from an initial segment only.

    \chapter{Formalism for finite-precision classical mechanics}\label{chap:further-development}

\lettrine[lines=3]{U}{ntil now, we} have studied two approaches to formulating finite-precision theories. In Chapter~\ref{chap:gisin}, we discussed the first attempt at such a formulation, proposed by \textcite{dSG19} in 2019 and based on the binary expansion of real numbers. We concluded, amongst other things, that the central role of the binary expansion precludes the theory from describing complete objective reality. Inspired by the idea that intuitionistic mathematics might be a natural language for finite-precision physics, we attempted in section~\ref{sec:intuitionistic-reals-physics} to use the intuitionistic real number system and notion of time to describe finite-precision quantities, but failed, as we had anticipated in section~\ref{sec:ontological-perspective} and as became clear in section~\ref{sec:doubts-again}.

In this chapter, I will outline yet another approach which might provide a suitable formalism for finite-precision theories of classical physics. It is inspired by the intuitionistic philosophy of the continuum, but, building on previous conclusions, is itself entirely classical (as regards the mathematics) and, in particular, does not try to equate physical and intuitionistic time as in section~\ref{sec:intuitionistic-reals-physics}. In section~\ref{sec:ontology-of-indeterminacy} I discuss some final necessary physical considerations, and in section~\ref{sec:mathematics-development} I present the mathematical formalism.

Before doing this, however, let us briefly recap the motivation for developing a finite-precision theory in the first place, which we discussed in section~\ref{sec:problem-randomness-and-indeterminism}. The main aim of developing this kind of theory is to show that there is no physical evidence for the usual view that physical quantities are given by infinitely precise real numbers. I do not intend to question the empirical predictions of existing theories, but only suggest that different but empirically equivalent theories are possible, which correspond to an ontology in which physical quantities do not have point-like values. In particular, we can conclude from these alternative theories that not \emph{points} but \emph{regions} are central to the physical continuum, and perhaps even that points play no role in physical ontology (see also section~\ref{sec:continuum}).

This purpose justifies limiting ourselves, for now, to classical mechanics. This has several advantages: apart from it being a simple (but still quite general) theory, the absence of indeterministic aspects in the usual interpretation of classical mechanics allows us to investigate the nature of the indeterminism which arises naturally in finite-precision theories. By giving a precise and consistent mathematical formalism for finite-precision classical mechanics, I hope to show that finite-precision interpretations of theories in general have the potential to be consistent. This might inspire others (and myself) to develop finite-precision formalisms for other theories such as quantum mechanics and general relativity theory.

Although I do not buy the arguments discussed in section~\ref{sec:problem-infinite-information} for the supposed principle that a physical system should be describable by only finitely much information, I do think that such theories are elegant, for the reasons already outlined in section~\ref{sec:randomness-indeterminism}. Namely, mathematical uncomputability, which is the mathematical interpretation of ‘infinite information’, often comes along with some degree of randomness or unpredictability, as is e.g.\ the case for most real numbers; and finite-information physical theories, which assume that this information is generated over time instead of being given in the initial condition only, elevate this mathematical notion of randomness to the physical notion of indeterminism. It indeed makes sense to associate the mathematically random objects in physical theories with physical indeterminism (recall that mathematical 1-randomness is a rigorous formulation of e.g.\ patternlessness and unpredictability).

If one is only interested in finite precision and not in finite information per se, then the results of the following sections are still useful (just skip all results about r.e.\ openness).

\

Finally, \emph{if} the finite-precision interpretation were true and classical chaotic systems were indeed indeterministic, then this would (surprisingly) have practical implications, as we would not have to rely on quantum processes to generate truly random numbers. Quantum random number generators today have practical use cases in cryptography, statistics and lotteries; in addition, the existence of true (algorithmic) randomness in classical or quantum physics impacts the validity of probabilistic algorithms \cite{chaitin1978note}.

\section{The ontology of indeterminacy}\label{sec:ontology-of-indeterminacy}
There are multiple possibilities to describe `indeterminacy' or `finite precision' of physical quantities, some of which are based on probability distributions and others of which are based on converging sequences. As these different descriptions result in different ontologies, before presenting our new theory, let us study some of these possibilities in more detail.

It is natural and tempting to associate indeterminacy with probability distributions, especially when comparing indeterminacy to `uncertainty' in quantum mechanics. We have seen that e.g.\ Del Santo and Gisin’s approach to finite-precision physics uses probabilities.\footnote{As do some other suggested reformulations of physics, such as those presented by \textcite{caticha2019information} (who proposes a `blurred spacetime’ based on the information metric) or \textcite{friedman1999fuzzy} (who apply t-norm fuzzy logic to physics) or \textcite{ydri2001fuzzy} (who proposes quantisation or `fuzzification’ of the manifold underlying QFT). We do not discuss these approaches in more detail because they are not motivated by the considerations in Chapter~\ref{chap:whats-the-problem} and differ from our conception of finite-precision theories in many ways.} However, for a number of reasons, I feel that probabilities do not have a place in finite-precision theories (in my way of thinking about those theories, indicated in the introduction to this chapter). The main reason is that I am not sure what should be the interpretation of these probability distributions, as a frequentist interpretation is not possible. In quantum mechanics, one can prepare an ensemble of particles in the same state and perform repeated measurements, in which case the frequency diagrams approximate the probability distribution given by the wave function via the Born rule. In finite-precision physics, however, one cannot prepare an ensemble of physical quantities to have the same the state of indeterminacy (i.e.\ to be associated with the same probability distribution). This is necessarily true if we require that the finite-precision theory is empirically equivalent to the corresponding orthodox theory in which the physical quantities under consideration are infinitely precise. Suppose, namely, that there is a process to induce a physical quantity, like the position of a particle, to be determined up to some probability distribution, and that this process can be applied to multiple quantities in order to arrange them into the same probability distribution. Then measurement of these quantities in general yields different results, while according to the orthodox interpretation, all measurements should yield the same results, as the quantities were all prepared equally.

Moreover, if one values the arguments against infinite information in real numbers, it should be noted that while describing physical quantities by probability distributions removes the point-like nature of the quantity, probability distributions still require infinitely much information to describe in general.\footnote{One possibility to describe probability distributions with finitely much information is to assume that they are given by the maximum entropy principle \cite{caticha2008} and to only specify the constraints, making sure to use finitely much information (see e.g. the approach in \textcite{caticha2019information}, where the constraints are the expected values and the variance, so that the resulting distribution is Gaussian).} See also the comments given in section~\ref{sec:measuring-information}. We also see this in the fact that probabilities do not play a role in intuitionism.

\

A second class of approaches to indeterminacies which we have already encountered uses sequences of regions of the continuum (which can be seen as subsets of the real number line), which we called `domains of indeterminacy’ in Chapter~\ref{chap:intuitionistic-physics}. The collapse of a physical quantity to a smaller DOI (i.e.\ one that is contained in the previous DOI) represents an increase in the determinacy of the quantity. We will use this approach in the theory proposed in section~\ref{sec:mathematics-development}.

It might be tempting to view physical quantities as being uniformly distributed over their domain of indeterminacy, but 
the fact remains that probability distributions have no place in finite-precision theories that are empirically equivalent to their orthodox counterpart (at least not in the frequentist interpretation of probabilities). Instead, every collapse to a smaller DOI is possible, and there is no rule stating what new DOI is more likely than others.

\

Finally, one might represent finite-precision or finite-information physical quantities by Cauchy sequences of e.g.\ rationals or computable numbers. One such approach is outlined in \textcite[Box~1]{gisin2020comment}. The definition of a Cauchy sequence $\a\in\Q^\N$ in (constructive) mathematics involves a \emph{Cauchy modulus}, which is, for example, a sequence $\mu\in\N^\N$ such that
\[ \forall m,n > \mu(k) \,( |\a(n) - \a(m)| < 1/k )\]
(see section~\ref{sec:int-reals-other-constructions}). If the Cauchy modulus belongs to the physical ontology of the described quantity, then this approach is a special case of the one discussed previously, as $\a$ and $\mu$ together define a shrinking sequence of real intervals (or balls).

If the Cauchy modulus is not part of the ontology, on the other hand, then this can lead to classically impossible worldlines and, again, empirical inequivalence, as a physical quantity is not bounded to a well-defined area and can therefore change its value after being measured.

Gisin’s approach \cite{gisin2020comment} describes quantities as rational sequences given by a computable function in conjunction with a single random big generator. Depending on the computable function, this indirectly defines a Cauchy modulus and therefore a domain of indeterminacy. Although Gisin’s approach has the advantage that it explicitly expresses the idea that information is generated bit by bit (through the random big generator) as time passes, I prefer to explicitly use domains of indeterminacy because it is simpler and more general (viz.\ DOIs can take more general shapes).

\subsection{Domains of indeterminacy}\label{sec:DOIs}
If the ontology of indeterminacies is indeed best described by shrinking regions of the continuum, then the question remains which types of region are allowed, and how DOIs of different physical quantities relate to each other.

As to the second question, to solve the problem of interdependency of physical quantities discussed in section~\ref{sec:technical-problems}, I consider only one `fundamental’ physical quantity, which is the phase space vector. That is, the (perfect-information) state of the system is given by the domain of indeterminacy of this vector, which is a subset of phase space.

This also somewhat clarifies the relationship between Parmenides time and Heraclitus time (which I described as still being unclear in sections~\ref{sec:gisin-reaction-hamiltonian-evolution} and~\ref{sec:technical-problems}): the state of the system, given by a subset of phase space (the DOI), evolves according to the Hamiltonian differential equations (that is, the Hamiltonian flow is applied to the whole DOI; this corresponds to the flow of Parmenides time), and at certain points of time the DOI (suddenly) collapses, i.e.\ is replaced by a smaller DOI contained in the previous one. (This idea, here sketched only conceptually, will probably be clearer when I express it in mathematical form in the next section.) The question of what exactly causes DOI collapse, however, is related to the classical measurement problem discussed in section~\ref{sec:gisin-classical-measurement-problem} and remains unanswered.

In particular, we see that under Hamiltonian evolution, a hyperrectangle in phase space does in general not transform into a hyperrectangle, which means that the indeterminate state of the system can indeed not be expressed by giving the DOIs of individual (one-dimensional) quantities such as the coordinates of phase space.\footnote{More precisely, the DOI of the phase space coordinates defines a DOI for each individual position and momentum coordinate (by projecting the phase space DOI onto the relevant axis), but the DOIs of all these quantities together do not completely describe the DOI of the phase space coordinates.}

The set of possible DOIs should be invariant under Hamiltonian evolution, but also under a large set of coordinate transformations, so that the ontology of the theory does not depend on the chosen coordinate system. In addition, by the motivation in the introduction to this chapter, the DOIs should satisfy some notion of computability (i.e.\ they should be describable by a finite algorithm).\footnote{A looser condition which also fulfills the desire that DOIs should `contain finitely much information’, is perhaps that the set of possible DOIs is countable (since fixing an enumeration enables one to uniquely identify each DOI by a natural number).} It should be clear by now (from e.g.\ section~\ref{sec:technical-problems}) that using rational intervals or hyperrectangles does not work.
One possible requirement that does satisfy these restrictions is that a DOI should be a recursively enumerable open (r.e.\ open) subset of phase space. This is a notion from computable analysis; the necessary background is given in Appendix~\ref{sec:computable-analysis-preliminaries}. In brief, a subset $E \subseteq\R^n$ is r.e.\ open if and only if it is the union of a computable sequence of rational open balls.

R.e.\ open sets are in general not effectively \emph{decidable}; that is, for a r.e.\ open set $E$, in general there does not exist an algorithm which is guaranteed to halt on input $x\in\R^k$ and outputs a 0 if $x\notin E$ and a 1 if $x\in E$. However, the arguments in section~\ref{sec:equating-time} suggest that it should not `take the Universe time’ to perform trivial computations; what I mean by that is that the statement $x\in E$ should have a physically well-determined truth value, given the algorithm that describes $E$.\footnote{Note that requiring the DOIs to be effectively decidable does not work since the only effectively decidable subsets of $\R^n$ are $\emptyset$ and $\R^k$.}

\nieuw{In the formalism introduced below, I define phase space to be a subset of the (classical) real number field $\R^k$. One might object that a theory that uses the real numbers cannot solve the problems introduced in Chapter~\ref{chap:whats-the-problem} (and section~\ref{sec:problem-infinite-information} specifically). This is not true, however. The complete set $\R^k$ is in my view not problematic when it is simply used to formalise the continuum as a whole; the problems of the orthodox interpretation only arise because it assigns a physical interpretation to \emph{individual} real numbers. I do not attribute physical significance to the mathematical idea that $\R^k$ is a set of points; instead, the state of the system is defined as a r.e.\ open subset of phase space, which can be described using finite information, even though it is, in the purely \emph{mathematical} sense, a set of points.}

Finally, although r.e.\ openness of a set is preserved under various transformations (e.g.\ translations along computable vectors, rotations along computable angles, etc.), it is in general not preserved under many others, like translations along uncomputable vectors, and under Hamiltonian evolution by an uncomputable time interval. Hence, in assuming that all DOIs that can represent the state of a physical system are r.e.\ open, we still seem to require that there is some restricted set of coordinate systems and set of timepoints that is `preferred’ by nature. However, it might be argued that once a coordinate system is specified in which the DOI is r.e.\ open, then any `physically useful' coordinate transformation (e.g.\ taking the distance between the DOIs of the positions of two particles as the new unit) is of an effective (computable) nature and hence preserves r.e.\ openness. Anyway, r.e.\ openness being the best option I have thought of so far,\footnote{The \emph{recursive closed sets} as defined in \textcite[§2.4]{weihrauch2000} deserve closer inspection! Yet another possibility is to take \emph{infinite-time decidable} sets, defined through Turing machines (with a finite set of instructions and program states) which can run for $\omega$ steps (see e.g.\ \textcite{hamkins2000infinite}). According to our previous arguments about computation time, this should not be a problem. However, the notion of \emph{decidable} sets is based on a continuum that consists of points, and not regions, in constrast to our preferred view.} let us now try to work out a possible mathematical formalism of finite-precision theories based on r.e.\ openness.

\section{A mathematical formalism for finite-precision classical mechanics}\label{sec:mathematics-development}
By the motivation in the preceding section, we will generalise the Hamiltonian formalism presented in section~\ref{sec:hamiltonian-formalism-orthodox} from points in phase space $\M$ to r.e.\ open subsets of $\M$, and in doing so we will use a \emph{classical} (i.e.\ nonconstructive, non-intuitionistic) approach. We will make use of notions from computable analysis and results from ODE theory. The necessary background on computable analysis is given in Appendix~\ref{sec:computable-analysis-preliminaries}; the necessary results from ODE theory will be recalled when used.

\begin{assumption}
    We assume that the \defn{phase space} $\M$ is r.e.\ open submanifold of $\R^k$ (where e.g.\ $k = 6N$ when the system is three dimensional and contains exactly $N$ particles). Perhaps the results of this section can be generalised to arbitrary Poisson manifolds; however, computability structures on manifolds have only recently been introduced \cite{aguilar2017computable}.
    
    In many cases, we expect the restrictions on the physical system defining $\M$ to be of an effective nature, so that $\M$ is most likely indeed r.e.\ open. Often, for example, the Hamiltonian is smooth everywhere except on the collision set $\mathcal C$, as discussed in section~\ref{sec:hamiltonian-formalism-orthodox}. In that case we take $\M = \R^{6N}\setminus\mathcal C$ as the phase space, which is r.e.\ open as subset of $\R^{6N}$.
    
    Alternatively, if $\M$ is not r.e.\ open, we can replace all mentions of `r.e.\ open' below by `r.e.\ open in $\M\subseteq\R^k$'.
\end{assumption}
\begin{assumption}
    The Hamiltonian $h:\subseteq \R^k\to\R$ is a $(\rho^k,\rho)$-computable function with domain $\dom h = \M$ (notations as in Appendix~\ref{app:computability-theory}). 
    The Poisson bracket is also assumed to be computable, as is, consequently, the Hamiltonian vector field $X_h:\subseteq\R^k\to\R^k$, also having the domain $\dom X_h = \M$.
    
    This can be assumed because all algebraic operations like addition, multiplication, division, exponentiation etc.~are computable.
\end{assumption}

Contrary to the orthodox interpretation, we do not view the state of the physical system as being given by a single point $x$ in phase space $\M$, but rather by a subset of $\M$.
\begin{definition}\label{def:blurred-state}
    A \defn{(blurred) state}, or \defn{domain of indeterminacy}, is a r.e.\ open subset $E\subseteq\M$ with compact closure. We denote the set of all blurred states by $\mathcal E$.
\end{definition}
In our interpretation, the \emph{perfect-information} state of the system at any computable timepoint is completely specified by such a r.e.\ open set with compact closure. `Blurred state' and `domain of indeterminacy' are synonymous; however, I prefer to move away from the terms `indeterminacy' and `blurred' whenever possible, because they suggest that the actual state of the system is more precise than described by this theory, and that points still play a more important role than regions. Therefore, I will sometimes refer to domains of indeterminacy simply as `states'.

Compact closure is \emph{necessary} because of Proposition~\ref{prop:beta-E}. It is \emph{justified} physically because we can require that the domain of indeterminacy is always bounded in phase space, and that it has a positive distance from the \defn{singularity set} $\Sing = \R^k\setminus\M$ (so that the closure of $E$ in $\R^k$ is contained in $M$). The latter can be made slightly more plausible if we assume that the singularity set is precisely the set to which a smooth extension of the Hamiltonian $h$ is not possible (e.g.\ if $h$ contains the gravitational potential energy~\eqref{eq:newton-gravity} and $\Sing = \mathcal C$, the collision set); in that case, points close to the singularity set correspond to `extreme' scenarios, and we can assume that a domain of indeterminacy collapses before it reaches such a singularity.\footnote{If we even assume that $h$ is unbounded when it approaches the singularity set (which is the case for the gravity potential), then we can say that, because the energy of a physical system is a well-behaved physical quantity and therefore should have a bounded domain of indeterminacy, the interval $h(E)$ induced by the domain of indeterminacy $E\subseteq\M$ of the phase space coordinates should be bounded, which can only be the case if the distance $d(E, \Sing)$ is indeed positive.}

\

We use the following theorem adapted from \textcite{graca2018} about computability of ODE solutions.
\begin{theorem}\label{thm:graca-ivp}
	Assume $W$ is a r.e.\ open subset of $\R^k$ and $x\in W$. Consider the initial-value problem
    \begin{equation}\label{eq:graca-ivp}
        \begin{split}
            y'(t) &= f(y) \\
		    y(0) &= x
        \end{split}
    \end{equation}
    where $f: W\to \R^k$ is $C^1$ and $(\rho^k,\rho^k)$-computable on $W$. There exists a unique \emph{maximal interval of existence and uniqueness} $(\a(x),\b(x))\subseteq\R$ of the solution $y(\cdot)$ of \eqref{eq:graca-ivp} on $W$. The functions $\a$ and $\b$ are continuous and $\a(x)<0<\b(x)$. Furthermore, we have:
	\begin{enumerate}
		\item The operator $(x,t) \mapsto y(t)$ is computable; 
		\item The operator $x \mapsto (\alpha(x),\beta(x))$ is semi-computable.
	\end{enumerate}
\end{theorem}

The latter means that the mapping $\beta:\subseteq \R^k\to \R, x \mapsto \beta(x)$ is lower-semicompu\-table in the sense of Definition~\ref{def:some-representations}, and $x \mapsto \alpha(x)$ is upper-semicomputable.

\

We need the property of compact closure in Definition~\ref{def:blurred-state} because this ensures states $E$ can be subjected to Hamiltonian evolution by at least a small time interval:
\begin{proposition}\label{prop:beta-E}
    Let $E\subseteq\M$ have compact closure. Then there exists a (unique) maximal number $\b(E) \in (0,+\infty]$ such that $E\subseteq\D_h^{\b(E)}$, with $\D_h^{\b(E)}$ as defined in~\eqref{eq:dht}. Likewise, there exists a minimal $\a(E) \in [-\infty,0)$ such that $E\subseteq\D_h^{\a(E)}$.
\end{proposition}
\begin{proof}
    The function $x\mapsto \b(x)$ from the previous Theorem is defined on $\M$ and maps to strictly positive values.
    This means that $\M$ is covered by the sets $\{ x\in\M \mid \b(x) > 1/n \}$ for $n\in\N_{>0}$; since $E$ has compact closure, it admits a finite subcover, so there is an $n\in\N$ with $E\subseteq\D_h^{1/n}$. We define $\b(E)$ to be the maximal number such that $E\subseteq\D_h^{\b(E)}$, or $\infty$ if such a maximal number does not exist. The proof for $\a(E)$ is similar.
\end{proof}


The following results show that blurred states transform into blurred states under computable time evolution, so that Definition~\ref{def:blurred-state} makes sense physically.

\begin{proposition}\label{prop:dht-r-e-open}
	For any computable $t\in\R$, the set $\D_h^t$ is r.e.\ open.
\end{proposition}
\begin{proof}
	For $t=0$ we have $\D_h^t = \M$, which is r.e.\ open by assumption.

	Now assume $t>0$. Then $\D_h^t = \{ x\in \M : \beta(x) > t \}$.
	Because $t$ is computable, we can approximate it from above by a rapidly converging computable sequence of rationals $(t_i)_{i=1}^\infty\subseteq\Q$ such that $0 < t_i - t < 2^{-i}$ for all $i$. Moreover, because $x\mapsto \beta(x)$ is lower-semicomputable (that is, $(\rho^k,\rho_<)$-computable) by Theorem~\ref{thm:graca-ivp} and is defined on $\M$, there exists a machine $M$ which computes a function $F:\subseteq\Som\to\Som$ which, on input $p\in\Som$ such that 
	$\rho^k(p) \in \M$, enumerates all rationals $r<\beta(\rho^k(p))$. 
	
	For all $i$, we can derive from $M$ a machine $M'_i$ which on input $p\in\Som$ simulates $M$ on $p$ and stops once $M$ outputs a rational $r > t_i$, but otherwise keeps running forever. Then for $p\in\Som$ such that $\rho^k(p) \in \M$, since $M$ enumerates all rationals below $\beta(\rho^k(p))$, $M'$ halts on $p$ if and only if $\beta(\rho^k(p)) > t_i$.
	
	Now construct a third machine $M''$ which on input $p$ runs all $M'_i$ on $p$ in parallel, in a diagonal manner: i.e.\ it first executes a step of $M'_1$, then one of $M'_2$, then one of $M'_1$ again, then $M'_2$, $M'_3$, $M'_1$, $M'_2$, $M'_3$, $M'_4$, $M'_1$, and so on. 
	$M''$ is made to halt whenever one of its $M'_i$ halts, and otherwise keeps running forever.
	
	It follows that for $p\in\Som$ such that $\rho^k(p)\in\M$, $M''$ halts on $p$ if and only if there exists an $i$ such that $M'_i$ halts on $p$, which is the case iff $\beta(\rho^k(p)) > t$, which in turn holds iff $p\in(\rho^k)^{-1}(\D_h^t)$. In other words, we have, for all $p\in\Som$,
	\[ p\in\dom M'' \cap (\rho^k)^{-1}(\M) \iff p\in (\rho^k)^{-1}(\D_h^t) \cap (\rho^k)^{-1}(\M) = (\rho^k)^{-1}(\D_h^t) , \]
	and hence $(\rho^k)^{-1}(\D_h^t)$ is r.e.\ open in $(\rho^k)^{-1}(\M)$. Because by assumption $(\rho^k)^{-1}(\M)$ is r.e.\ open, $(\rho^k)^{-1}(\D_h^t)$ is r.e.\ open by the final comment in Proposition~\ref{prop:intersection-r-e-open}, and therefore $\D_h^t$ itself is.
	The proof for $t<0$ is similar.
\end{proof}

\begin{theorem}\label{thm:flow-preserves-r-e-open}
    Let $E\subseteq\M$ be r.e.\ open and $t\in\R$ computable with $\a(E) < t < \b(E)$. Then $\Phi_t(E)$ is again r.e.\ open.
\end{theorem}
\begin{proof}
    Note that $\Phi_t(E) = (\Phi_{-t})^{-1}(E)$. Let $F_{-t}: \Som\to\Som$ be a $(\rho^k,\rho^k)$-realiser of $\Phi_{-t}$. Because $\D_h^{-t} \subseteq \dom \Phi_{-t}$, we have $(\rho^k)^{-1}(\D_h^{-t}) \subseteq \dom F_{-t}$.
    
    It follows from Theorem~\ref{prop:preimage-r-e-open} that $(\rho^k)^{-1}(\Phi_t(E))$ is r.e.\ open in $\dom F_{-t}$, to wit 
    \[ (\rho^k)^{-1}(\Phi_t(E)) = \dom h \cap \dom F_{-t} \]
    for some computable function $h:\subseteq\Som\to\Sst$.
    We also have $\Phi_t(E) \subseteq \D_h^{-t}$, which implies $(\rho^k)^{-1}(\Phi_t(E)) \subseteq (\rho^k)^{-1}(\D_h^{-t}) \subseteq \dom F_{-t}$. This means
    \[ (\rho^k)^{-1}(\Phi_t(E)) = \dom h \cap \dom F_{-t} \cap (\rho^k)^{-1}(\D_h^{-t}) = \dom h \cap (\rho^k)^{-1}(\D_h^{-t}), \]
    and hence $(\rho^k)^{-1}(\Phi_t(E))$ is also r.e.\ open in $(\rho^k)^{-1}(\D_h^{-t})$. Because $(\rho^k)^{-1}(\D_h^{-t})$ is r.e.\ open by Proposition~\ref{prop:dht-r-e-open}, $(\rho^k)^{-1}(\Phi_t(E))$ and hence $\Phi_t(E)$ itself are r.e.\ open.
\end{proof}

\begin{corollary}
    Let $E\subseteq\M$ be a state, and let $t\in\R$ be computable with $\a(E)<t<\b(E)$. Then $\Phi_t(E)$ is again a state.
\end{corollary}
\begin{proof}
    $\Phi_t:\D_h^{t} \to \D_h^{-t}$ is a diffeomorphism. Because $\forall x\in E: \b(x) \geq \b(E) > t$ and $x\mapsto \b(x)$ is continuous and $E\subseteq\D_h^t$, we have $\o E\subseteq\D_h^t$. We have $\o{\Phi_t(E)} = \Phi_t(\o E)$ because $\Phi_t$ is a diffeomorphism and $\Phi_t(\o E) \subseteq\D_h^{-t}$. Since $\o E$ is compact, $\Phi_t(\o E)$ is compact, so that $\Phi_t(E)$ has compact closure. By Theorem~\ref{thm:flow-preserves-r-e-open}, $\Phi_t(E)$ is also r.e.\ open.
\end{proof}

We have seen in Proposition~\ref{prop:beta-E} that a state $E$ can evolve through time for at least a small time interval. However, if $\b(E) < \infty$ then $E$ contains initial conditions which lead to singularities or escape to infinity in finite time. 
If we assume as before that (neighbourhoods of) singularities (i.e.~points $x\in\R^k\setminus\M$) correspond to extreme physical situations, then it is plausible that the domain of indeterminacy $E$ collapses before the timepoint $\b(E)$ is reached.
For example, when $E$ contains initial conditions which escape to infinity at time $\b(E)$, or in any other way shows chaotic behaviour, then the diameter of the set $\Phi_t(E)$ blows up as $t\to\b(E)$ from below. This means that the domains of indeterminacy of certain physical quantities become very large, which is a possible mechanism that triggers a collapse of the domain of indeterminacy.

Therefore, we assume that a state $E$ always collapses at a timepoint before $\b(E)$, that is, at some computable $0<t<\b(E)$, the state $\Phi_t(E)$ collapses to a smaller $E'\subseteq\Phi_t(E)$.\footnote{One could also remove the requirement that $t$ be computable; then $\Phi_t(E)$ would in general not be r.e.\ open, but the function $t\mapsto \Phi_t(E)$ would be (i.e.\ there is a Turing machine \emph{supplemented with an oracle that gives a $\rho$-name for $t$}, which halts on $p\in\Som$ iff $\rho^k(p)\in \Phi_t(E)$). We focus on the case of computable $t$ here.} This manifests a step in Heraclitus time. More precisely, we define:
\begin{definition}\label{def:blurred-world}
    A \defn{(blurred) world} is a sequence $\{(E_n,t_n)\}_{n\in\N}\subseteq \mathcal E\times\R_c$ such that for all $n\in\N$, we have
    \[ 0 < t_{n+1} - t_{n} < \b(E_n) \quad\text{ and }\quad     E_{n+1} \subseteq \Phi_{t_{n+1} - t_n} (E_n). \]
    Here $\R_c$ is the set of $\rho$-computable reals and $\mathcal E$ is the set of blurred states on $\M$.
    
    For $t\in (\inf_n t_n, \sup_n t_n)$, the \defn{(blurred) state at time $t$} is given by $\Phi_{t-t_m}(E_m)$, where $m$ is the greatest integer $k$ such that $t_k \leq t$.
    
    If $E$ is a blurred state, then a world \defn{through $E$} is any world $W = \{(E_n,t_n)\}_{n\in\N}$ such that $E_n = E$ for some $n$. If this $n=0$, $E$ is an \defn{initial condition} for the world $W$.
\end{definition}

It follows directly from Proposition~\ref{prop:beta-E} that for any blurred state $E\in\mathcal E$ and computable $t\in\R_c$, a blurred world $\{(E_n,t_n)\}_{n\in\N}$ with $E_0 = E$ and $t_0 = t$ exists.

\nieuw{Note that by Liouville's theorem, $\Phi_t$ preserves the volume of phase space regions; hence, the volumes of the states are non-increasing, i.e.\ $\vol(E_{n+1}) \leq \vol(E_n)$ for all $n$. In this sense, the indeterminacy in the state of the system cannot increase with time; however, the diameter of a phase space region can in general increase with Hamiltonian evolution, in particular when the system is chaotic.}

\subsection{Completeness}\label{sec:completeness}
An apparent shortcoming of Definition~\ref{def:blurred-world} is that the state of the system might not be defined for all time values $t\in\R$, namely if $\sup_n t_n < \infty$. Let us call a world $\{(E_n,t_n)\}_n$ that satisfies $\sup_n t_n = \infty$ \defn{complete}. We cannot simply add the condition to Definition~\ref{def:blurred-world} that all worlds are complete, because although every initial condition $(E_0,t_0)\in\E\times\R_c$ appears in a blurred world as defined above, not every initial condition might appear in a complete blurred world. (For example, $\b(E_n)$ might shrink fast enough so that no sequence $t_n$ with $t_n\to\infty$ is possible.)

As a result, we have to accept that not all worlds might be complete (the definite answer to this question depends on $\M$ and the Hamiltonian $h$). This is exactly analogous to the orthodox interpretation of classical mechanics, where we accept that some initial conditions might lead to a singularity in finite time and are consequently not defined beyond a certain timepoint.

Let us try to identify the cases in which a complete blurred world \emph{is} possible.

\begin{definition}
    For any $t>0$, the \defn{pre-singular set} $\Sing_t$ is defined as
    \[ \Sing_t := \M\setminus\D_h^t, \]
    i.e., it is the set of point-like initial conditions (i.e.\ in the orthodox sense) that lead to a singularity or escape to infinity before or at time $t$.
\end{definition}

\begin{theorem}\label{thm:complete-worlds-nowhere-dense}
    \begin{enumerate}[(i)]
        \item Let $E\subseteq\M$ be a state. A complete world through $E$ exists if and only if for all $t>0$, $\Sing_t$ is not dense in $E$.
        \item A complete world exists through any state $E$ if and only if for all $t>0$, $\Sing_t$ is nowhere dense in $\M$ (i.e.\ the closure of $\Sing_t$ has empty interior).
    \end{enumerate}
\end{theorem}
\begin{proof}
    (i) Suppose that $\Sing_t$ is dense in $E$ and that $\{(E_n,t_n)\}$ is a world through $(E,t_0=0)$. Then for all $n\geq 0$ with $t_n < t$, $\Phi_{-t_n}(E_n) \subseteq E$ so $\Sing_t$ is dense in $\Phi_{-t_n}(E_n)$, so that $\b(E_n) < t-t_n$. Hence, a world through $E$ cannot pass beyond time $t$.
    
    Conversely, suppose that for all $t>0$, $\Sing_t$ is not dense in $E$. Then for all $t>0$, the set $E\cap\D_h^t$ is non-empty. Because $E$ is r.e.\ open by Definition~\ref{def:blurred-state} and $\D_h^t$ is r.e.\ open by Proposition~\ref{prop:dht-r-e-open}, $E\cap\D_h^t$ is also r.e.\ open (see Proposition~\ref{prop:intersection-r-e-open}). Take any increasing sequence $(t_n)_{n\in\N}$ such that $t_n\to\infty$ as $n\to\infty$. Then $\{(\Phi_{t_n}(E\cap\D_h^{t_n+1}), t_n)\}_{n\in\N}$ defines a complete world through $E$.
    
    (ii) follows from (i) by noting that a set is nowhere dense iff it is not dense in any basic open set, and that $\E$ forms a basis for the topology on $\M$.
\end{proof}

Unfortunately, whether the conditions in Theorem~\ref{thm:complete-worlds-nowhere-dense} are satisfied is an open problem even in, for example, Newtonian gravitational systems. In these systems, for $N\geq 2$, there are initial conditions (in the orthodox sense) which lead to singularities in the form of particle collisions in finite time. For $N\geq 4$, however, there are also initial conditions which do not lead to a collision, but in which particles escape to infinity in finite time (see \textcite{earman1986} for more discussion). In the case of 4 particles, it is known that the set of initial conditions leading to (collision or noncollision) singularities in finite time, i.e.\ the set $\bigcup_{t>0} \Sing_t$, has Lebesgue measure zero. But this does not imply that the sets $\Sing_t$ are nowhere dense, and cases for $N>4$ remain to be settled.\footnote{Even in classical mathematics, that is; recall that this entire section uses classical mathematics!}

\

However, if the conditions of Theorem~\ref{thm:complete-worlds-nowhere-dense}(ii) \emph{are} satisfied, then the situation is slightly more satisfying than in the orthodox case: while in our theory a world through a given initial condition $E$ extending beyond each timepoint always exists, in orthodox classical physics there are still initial conditions which are guaranteed to lead to a singularity in finite time.
Additionally, if the conditions of Theorem~\ref{thm:complete-worlds-nowhere-dense} are satisfied, and if we assume that on approaching the singularities $\Sing_t$ the state keeps collapsing to ever smaller diameters, then we can even (informally) say that the probability that the system `avoids' the singularities is one.\footnote{Here I mean only a mathematical interpretation of the word `probability': the set of worlds which can pass beyond the singularity has full measure in the set of all possible worlds through a given initial condition.}

\subsection{Time and indeterminism}
Intuitively, we speak of $n\in\N$ as parameterising Heraclitus time; the fact that it is a discrete parameter corresponds to the idea that only a finite amount of new information is generated in a finite amount of time. (Once again, this parallels a concept from intuitionism, namely the fact that the time experienced by the creating subject is modelled by a discrete parameter (see section~\ref{sec:intuitionism-time}).) The parameter $t$ and evolution $E\mapsto \Phi_t(E)$ for $\a(E) < t < \b(E)$ correspond to Parmenides time evolution. The sequence $(t_n)_{n\in\N}$ couples Heraclitus time to Parmenides time.

In the current interpretation, the set of possible worlds $\mathcal W$ is precisely given by the set of blurred worlds on $\M$ as defined above.
It follows directly from Proposition~\ref{prop:beta-E} that for any blurred state $E\in\mathcal E$ and computable $t\in\R_c$, a blurred world $W = \{(E_n,t_n)\}_{n\in\N}$ with $E_0 = E$ and $t_0 = t$ exists; it is also evident that multiple such blurred worlds exist. Therefore, we see that our theory of classical physics is indeterministic in Earman's sense (see section~\ref{sec:determinism}): given any initial condition $(E_0,t_0)$, we may provide two differing worlds $W,W'\in\mathcal W$, both through $(E_0,t_0)$.

\subsection{Relation to intuitionistic reals}
As discussed in the the introduction to this chapter and in section~\ref{sec:DOIs}, the present approach to finite-precision physics is entirely classical and therefore differs conceptually from the approach in section~\ref{sec:intuitionistic-reals-physics}. However, there are still various similarities between the two theories on the technical side. For example, in the present theory, any one-dimensional physical quantity that can be expressed as a continuous function $f:\M\to\R$ of phase space, such as the position or momentum of a particle or the energy of the configuration, inherits a domain of indeterminacy from the blurred state $E$. When $E$ is connected (which we consider another physically justifiable assumption on blurred states), this domain of indeterminacy is an interval. However, it is not a rational interval, as the most literal application of the intuitionistic philosophy of real numbers to physics would suggest.\footnote{I would intuitively expect that if $f$ were $(\rho^k,\rho)$-computable, then $f(E)$ would be r.e.\ open; however, it appears that only taking pre-images preserves r.e.\ openness (see Proposition~\ref{prop:preimage-r-e-open}).}

A difference between the theory of this chapter and the sequences of rational intervals in the definition of the intuitionistic reals is the property that the latter \emph{dwindle}, i.e.\ shrink to arbitrarily small lengths (see Definition~\ref{def:intuitionistic-reals}). In our theory, this would translate to adding the condition to Definition~\ref{def:blurred-world} that
\begin{equation}\label{eq:diam-dwindles}
    \forall m \exists n: \diam(\Phi_{t_0-t_n}(E_n)) < 2^{-m},
\end{equation}
such that when waiting long enough, the state at time $t_0$ becomes determined with arbitrary precision.\footnote{As well as the state at any other time $t$ between $\inf_n t_n$ and $\sup_n t_n$, since $\Phi_{t-t_0}$ is continuous.} I myself do not see direct physical motivation for adding this condition (see, however, the next section).

Finally, note that we can see the two conditions in Definition~\ref{def:blurred-world} as a spread law on the sequences $\{(E_n,t_n)\}\subseteq\mathcal E\times\R_c$, defining the spread of possible worlds (see section~\ref{sec:spreads} for the definition of spreads). This spread law is, however, not decidable and it does therefore not make sense to speak of intuitionistic spread laws in this context.

\subsection{The orthodox theory as limit case}
If we do make the requirement~\eqref{eq:diam-dwindles}, then we recover the orthodox theory of classical physics in the limit $n\to\infty$, because the sequence $\{\Phi_{t_0-t_n}(E_n)\}_n$ encloses a unique real number,\footnote{And it defines a unique world line through $t\mapsto \bigcap_n \{\Phi_{t-t_n}(E_n)\}_n$.} and conversely, every real number can be defined by an infinite shrinking and dwindling sequence of r.e.\ open sets. We see that the orthodox view corresponds to a view back on the initial conditions `from the end of time', as was already suggested in section~\ref{sec:how-have-we-come-to-orthodox}. Moreover, because the infinite sequence $\{\Phi_{t_0-t_n}(E_n)\}_n$ is not necessarily computable, we again see that it in general takes infinitely much information to specify the initial condition of the orthodox interpretation (while every finite initial segment of the sequence requires only finitely much information, by virtue of the sets $E_n$ being r.e.\ open).


    \chapter{Conclusion and prospects}\label{chap:open-questions}

\lettrine[lines=3]{W}{e have discussed} three attempts (each of which either more or less successful) at formulating an indeterministic classical physics based on the assumption that physical quantities do not have infinitely precise values: a theory proposed by \textcite{dSG19} and based on the binary expansion of reals (Chapter~\ref{chap:gisin}); an attempt to use the intuitionistic philosophy of the continuum to describe finitely precise quantities (Chapter~\ref{chap:intuitionistic-physics}); and a theory based on the Hamiltonian formalism, replacing points by regions and using notions from computable analysis (Chapter~\ref{chap:further-development}). Putting the specific technical and conceptual differences between these theories aside, they show us, amongst other things, that contrary to the orthodox interpretation, one can also think of physical quantities as being only finitely precise at each point in time; that the physical continuum can also been as revolving around \emph{regions}, instead of points; that information about physical quantities might only come into existence when it becomes physically relevant, or at least when it is measured; that this naturally introduces indeterminism, in particular in chaotic systems and even in theories such as classical mechanics; and that it is possible that the state of a physical system is describable using only finitely much information. However, we have also seen that the questions of whether our real world is deterministic or indeterministic, whether its quantities are finitely or infinitely precise and whether its state is describable by finite or infinite information are empirically underdetermined and cannot be settled by physical argument only.

Moreover, in Chapter~\ref{chap:intuitionistic-physics} we discussed the debate of using constructive and intuitionistic mathematics to describe physics, and concluded that every attempt to do so would have to deal with the apparent collision between the objectivity of physics and the anthropocentrism of constructive and intuitionistic mathematics, and with the fact that the motivations for indeterminacy in physics and mathematics are completely different. In addition, the rational numbers involved in the intuitionistic approach to real numbers have some shortcomings from the technical side, suggesting that they, while being essential to the human construction of the intuitive continuum, play no significant role in the physical continuum.

However, the intuitionistic way of approaching the continuum through ever shrinking regions instead of points has inspired us in our final formulation of finite-precision physics in Chapter~\ref{chap:further-development}. The approach presented there is only of numerous possible ways to formulate finite-precision classical physics, and it still suffers from some problems. One of these, which is related to the invariance of r.e.\ open sets under coordinate transformations, was discussed in the last paragraph of section~\ref{sec:DOIs}. In this chapter, we discuss some more problems and questions arising in the context of the formalism of Chapter~\ref{chap:further-development} as well as finite-precision theories in general. We also discuss the possibility of finite-precision interpretations of other physical theories than classical mechanics.

\section{Universal constants}
Until now, we have implicitly assumed that the universal physical constants have infinitely precise values. The arguments from Chapter~\ref{chap:whats-the-problem}, however, might suggest that also these physical constants should only be determined up to finite precision at each point in time. This would add another level of indeterminacy to the theory, related to the differential equations themselves instead of the initial conditions. Indeed, the assumption that the Hamiltonian flow $\Phi:\D_h\to\M$ is a well-defined, single-valued operator hinges on the assumption that all constants appearing in Hamilton’s equations are infinitely precise. In the presence of indeterminacy in these constants, even time evolution of the DOI with respect to the differential equations (i.e.\ Parmenides time evolution) would become time-irreversible, since $\Phi_t$ and $\Phi_{-t}$, when applied to subsets of phase space, would not be each other’s inverse.

Furthermore, while quantities like positions and momenta of particles can be said to be localised in space, the value of universal constants is relevant to all regions of space. This might form a problem if new information about the universal constants is generated over time, because this new information would have to be communicated faster than light across the Universe (if this were not the case, it would be interesting to think what would happen when two isolated systems with differing physical constants came into contact). This is similar to the EPR paradox and nonlocality in quantum mechanics.\footnote{Of course, this problem also applies to ‘localised’ quantities such as positions and momenta, in the presence of forces that act at a distance.}

Of course, many universal constants can be set to 1 by choosing a suitable unit system, but in general, not all constants can be set to 1 simultaneously. The arguments from Chapter~\ref{chap:whats-the-problem} might, however, not necessarily apply to physical constants, as physical constants form part of the physical \emph{laws}, whereas Chapter~\ref{chap:whats-the-problem} was mainly aimed at quantities which describe the \emph{configuration} of a physical system.\footnote{Another (but ad hoc) way out of this problem is to postulate that ratios between universal constants of the same dimension are given by rational or computable numbers and are therefore not subject to the arguments against infinite information.} In this case, the formalism of Chapter~\ref{chap:further-development} could still be applied.


\section{The past; the thermodynamic arrow of time}\label{sec:past-thermo}
A world is defined in Definition~\ref{def:blurred-world} as a one-way infinite sequence; hence, Heraclitus time extends in one direction only. One could also try to let the index $n$ range over $\Z$ instead of $\N$. However, if $\a(E_0) > -\infty$, i.e.\ a singularity lies in the past of $E_0$, then it would be impossible for the past to extend beyond this point in time (contrary to the future, for singularities that lie in the future of $E_0$ might be `avoided’ by collapse to a smaller DOI, as discussed in section~\ref{sec:completeness}).

Let me briefly note that there are two ways to view the past in a finite-precision physics; either the past stays fixed, i.e.\ the initial condition is always given by $E_0$, or the past becomes more determined as time progresses, i.e.\ the initial condition is given by $\Phi_{t_0-t_n}(E_n) \subseteq E_0$ at time $t_n$. These views might be compared to the difference between the \emph{growing block universe} and \emph{presentism} in the philosophy of time \cite{sep-time}, respectively.

Moreover, let me note that the assumption that DOIs always collapse to subsets (i.e.\ $\Phi_{t_{n+1}-t_n}(E_n) \subseteq E_{n+1}$) might be unwarranted. It is also possible to imagine a theory which sets a certain minimal resolution to phase space (which perhaps varies over time). As Hamiltonian evolution of a region of phase space is often ergodic (so that the diameter of the region increases), but preserves the volume of the region by Liouville’s theorem, the shape of the region generally becomes increasingly intricate over time; if the level of detail in this shape requires a higher resolution than the minimum allowed resolution, then collapse of the region to a new DOI might require the new DOI to tread outside of the previous one. Such a theory would, however, allow classically impossible worldlines and it is questionable whether it would be empirically equivalent to the orthodox interpretation.

\

It would also be interesting to work out the implications of finite-precision physics for statistical mechanics. These theories are closely related, but at the same time radically different: statistical mechanics deals with uncertainties over the set of microstates, while finite-precision physics assumes that microstates themselves are inherently indeterminate. Indeed, the probability distributions in statistical mechanics are often interpreted as epistemic \cite{FrW11}, while the indeterminacies in finite-precision physics should be thought of as ontological.

In particular, it would be insightful to explore the relation with the second law of thermodynamics in more detail. Boltzmann’s proof of the $H$-theorem, which derives this law from the classical kinetic theory of gases, has historically led to much debate because of its seeming incompatibility with (the orthodox interpretation of) classical mechanics \cite{sep-time-thermo}. Two notable objections to the $H$-theorem were brought forward by Johann Loschmidt in 1877 and Ernst Zermelo in 1896. Loschmidt remarked that since classical mechanics is time-reversal invariant, the $H$-theorem would imply that entropy increases towards the past. This paradox might be resolved if we accept the finite-precision interpretation, which is not time-reversal invariant. Zermelo, on the other hand, based his argument on Poincaré’s recurrence theorem, which states that a bounded isolated classical system always returns arbitrarily closely to its initial state. However, as pointed out by \textcite{dSG19}, the recurrence theorem depends on the necessary condition that this state is defined with infinite precision.\footnote{An interesting related remark is that the (usual) proof of the recurrence theorem is not constructive; it guarantees recurrence, but does not tell us when it will happen \cite{zhang2017witnessing}.}

Finally, \textcite{drossel2015} argues that deriving macroscopic time-irreversible laws such as the second law of thermodynamics is only possible when the underlying microscopic theory is time-reversal invariant itself, and that dropping the assumption of infinite precision would make such derivations possible.

\section{Beyond classical mechanics}
Another prospect for future research is how to formally apply the idea of finite precision to theories beyond Newtonian classical physics. Most modern physical theories are field theories in which time evolution is given by partial differential equations instead of ordinary differential equations. It is therefore important to consider how to define a finite-precision field (i.e.\ a function of space or spacetime $\M$) of which the precision increases over time, and whether it makes sense to impose the requirement that these fields should be describable by only finitely much information. There are multiple possibilities to do this: one could think, for example, of two computable functions $\M\to\R$ defining upper and lower bounds for the value of a one-dimensional field on $\M$ at each point in $\M$; alternatively, one could use the method of approximation of measurable functions by simple functions $\sum_{i=1}^N a_{i}\chi_{A_i}$, where $a_i\in\R$ is e.g.\ computable and $A_i\subseteq\M$ is r.e.\ open (or satisfies another computability property) for each $i$. In intuitionistic mathematics, there is also a way to code continuous functions $\phi:\NN\to\NN$, where $\NN = \N^{\N}$, and hence continuous functions $\R^n\to\R^m$, by sequences of natural numbers, so that it could be possible to see fields as special cases of choice sequences of natural numbers (see e.g.\ \textcite[§1.3.4]{veldman2008borel}), just as we can identify real numbers with choice sequences.

In quantum mechanics, a ‘blurred state’ could alternatively be modelled by a subset of Hilbert space with some computability notion; however, the fact that this Hilbert space is in general not separable could pose problems. Using finite-precision wave functions would add a level of indeterminism to the one already present in quantum mechanics, and it would be interesting to find out what would be the relation between these two types of indeterminism. Moreover, Stone's theorem prohibits the existence of singularities in quantum mechanics such as those that exist in classical mechanics \cite{landsman2017}, which means that a finite-precision quantum mechanics could be drastically different from the one presented in Chapter~\ref{chap:further-development}. (The absence of singularities in quantum mechanics is the reason that quantum mechanics is sometimes called `more deterministic' than classical mechanics, apart from wave function collapse.)

Finally, in the special and general theories of relativity, the relativity of time complicates talking about physical quantities becoming more precisely determined over time. Perhaps we could just see the universe as a four-dimensional `blurred block' of spacetime (although singularities might complicate this), the contents of which become more precise over Heraclitus time, which could therefore be seen as a fifth dimension. In this case the connection between Heraclitus and Parmenides time would be less clear; however, researching this could perhaps even shed light on the problem of time in GR. Alternatively, in GR one could fix a spacelike foliation (i.e.\ set of time slices) and model the solution of the Einstein equations as evolving along that foliation, possibly using the ADM formalism (based on a Hamiltonian) \cite{arnowitt2008republication} to arrive at a similar formulation as in Chapter~\ref{chap:further-development}. Of course, the question remains what foliation (or what Cauchy surface) should be chosen to accurately model the connection with Heraclitus time, and whether Heraclitus time should perhaps also be a relative notion.

    \appendix
    
    \chapter{Classical mechanics and determinism}\label{app:preliminaries-physics}

\section{Hamiltonian mechanics: the orthodox interpretation}\label{sec:hamiltonian-formalism-orthodox}
Although the discussions in this thesis apply to a wide range of physical theories (indeed, all theories that make use of the real numbers), a simple and useful example which we focus on throughout is Newtonian classical physics. In this appendix we will briefly walk through the mathematical formalism of this theory (or rather, family of theories) and discuss its usual interpretation. Some of the notation introduced here is used in section~\ref{sec:mathematics-development}.
We focus on the formalism of Hamiltonian point mechanics, which describes classical systems consisting of $N$ point particles. We will first sketch the purely mathematical context and then discuss the physical interpretation. (The mathematical details are not that important and can be skipped on a first reading. A more detailed exposition can be found in~\textcite[Ch.~3]{landsman2017}.)

We define \defn{phase space} to be an arbitrary Poisson manifold $\M$, i.e.\ a manifold with Poisson bracket $\{\cdot,\cdot\}:C^\infty(\M)\to C^\infty(\M)$.\footnote{A Poisson bracket is a mapping $\{\cdot,\cdot\}:C^\infty(\M)\to C^\infty(\M)$ which is anti-symmetric and satisfies the Jacobi identity, and has the property that the mapping $\d_h: f\mapsto \{h,f\}$ is a derivation of $C^\infty(\M)$ for each $h\in C^\infty(\M)$.}
A smooth function $h\in C^\infty(\M)$ is singled out and called the \emph{Hamiltonian}. Because $\M$ is a Poisson manifold, $h$ corresponds to a derivation $\d_h : f\mapsto \{h,f\}$ of $C^\infty(\M)$, which corresponds to a vector field $X_h$ on $\M$, called the \defn{Hamiltonian vector field}. A \defn{curve} through $x_0\in\M$ is a smooth function $\gamma: J\to \M$, where $J\subseteq\R$ is some open real interval containing $0$ such that $\gamma(0) = x_0$. A curve $\gamma:J\to\M$ \defn{integrates} $\d_h$ (or $X_h$) if for all $f\in C^\infty(\M)$ and $t\in J$, $\d_h f(\gamma(t)) = \frac{\dif}{\dif t} f(\gamma(t))$.
It follows from standard results of the theory of ordinary differential equations that $X_h$ can be integrated around every $x_0\in\M$ by a unique maximal curve $\gamma_{x_0}:J_{x_0} \to \M$, where $J_{x_0}$ is the \defn{maximal interval of existence and uniqueness}. Taken together, these intervals define the Hamiltonian \defn{flow domain}
\[ \D_h = \{ (t,x) \in \R\times\M \mid t\in J_x \}. \]
In section~\ref{sec:mathematics-development} we also use the notation
\begin{equation}\label{eq:dht}
    \D_h^t = \{ x\in\M \mid (t,x)\in D_h \}, \text{\quad for } t\in\R.
\end{equation}
The \defn{Hamiltonian flow} $\Phi:\D_h\to\M$ is a smooth function such that for all $x\in\M$, $\Phi(\cdot,x)$ is a curve through $x$ which integrates $X_h$. Finally, we denote
\[ \Phi_t: \D_h^t \to \M,\ \Phi_t(x) := \Phi(t,x) \text{\quad for all } t\in\R, x\in\D_h^t .\]

\

Physically, the above formalism can be used to describe classical systems consisting of $N$ (generalised) point particles, in which case phase space is parameterised by the (generalised) positions and (generalised) momenta of all particles, denoted by $\bm{q_i}\in\R^3$ and $\bm{p_i}\in\R^3$ respectively for $i\in \{1,2,\ldots,N\}$. (Hence, in the case of a three-dimensional system, $\M\subseteq\R^{6N}$.) The essence of the `orthodox' (i.e.\ usual) interpretation of classical physics, which will be a subject of debate in this thesis, is that a \defn{state} of the system is represented by an (`infinitely precise') point $x$ in phase space $\M$.

The Hamiltonian $h\in C^\infty(M)$ physically represents the energy of the system; in practise, it consists of kinetic energy terms,
\[ \sum_{i=1}^N \frac{|\bm{p_i}|^2}{2m_i}, \]
where $m_i$ is the mass of particle $i$, and potential energy terms like Newton’s gravitational potential
\begin{equation}\label{eq:newton-gravity}
- \sum_{i} \sum_{j\neq i} \frac{G m_i m_j}{|\bm q_i - \bm q_j|},
\end{equation}
where $G\in\R_{>0}$ is Newton’s gravitational constant. A term like the gravitational potential excludes the \defn{collision set}
\[ \mathcal C = \{(\bm{q_1},\bm{q_2},\ldots,\bm{q_N},\bm{p_1},\bm{p_2},\ldots,\bm{p_N}) \mid \bm{q_i} = \bm{q_j} \text{ for some } i\neq j\} \]
from the set of possible states, so that the phase space is in such cases often taken to be $\M = \R^{6N}\setminus\mathcal C$. (It should be respected, however, that the mathematical formalism sketched above captures a very broad set of classical systems. The precise forms of the phase space and the Hamiltonian depend on the particular classical theory under consideration.)

In the orthodox interpretation, a state $x_0\in\M$ uniquely determines a \defn{wordline}, i.e.\ a curve $\gamma_{x_0}: J_{x_0} \to \M$ through $x_0$ integrating the Hamiltonian vector field $X_h$. This represents the evolution of the system through time. With an appropriate choice of Poisson bracket on $\M$ and writing this curve as $\gamma_{x_0}(t) = (\bm{q_1}(t),\ldots,\bm{q_N}(t),\bm{p_1}(t),\ldots,\bm{p_N}(t))$, $\gamma_{x_0}$ is the unique maximal curve through $x_0$ that solves \defn{Hamilton's equations} (or the \emph{Hamiltonian equations}) of motion\footnote{More generally, it follows from a theorem by Darboux that every $2n$-dimensional Poisson manifold with invertible Poisson tensor admits charts around any point $x\in\M$ in which the Poisson bracket is of the desired form, leading to the form of Hamilton's equations shown. See~\textcite[Theorem~22.13]{lee2018manifolds}.}
\begin{align*}
	\ddt{\bm{p_i}}(t) &= -\nabla_{\bm{q_i}} h(\bm{q_1}(t),\ldots,\bm{q_N}(t),\bm{p_1}(t),\ldots,\bm{p_N}(t)); \\
	\ddt{\bm{q_i}}(t) &= \nabla_{\bm{p_i}} h(\bm{q_1}(t),\ldots,\bm{q_N}(t),\bm{p_1}(t),\ldots,\bm{p_N}(t)),
\end{align*}
for all $i\in\{1,2,\ldots,N\}$.

\section{Determinism in physics}\label{sec:determinism}
\emph{Determinism} is roughly the idea that given the configuration of a physical system (or ‘world’) at a time $t$, the future of the system is completely fixed by the natural laws. Although the idea of determinism is ancient, giving a precise and unambiguous definition of determinism has proven a very difficult task; moreover, the requirements for a definition of determinism to be satisfactory also partly depend on the theory under consideration (e.g.\ relativity of time in relativity theory impacts the meaning of determinism). A definition that we will consider in this thesis is formulated in \textcite{earman1986} (similar to the one given by \textcite{montague1974}) and is based on the common philosophical parlance of `possible worlds’. Here, a `world’ means a four-dimensional space-time world, i.e.\ a collection of all events that have happened, are happening or will happen in that world, and a world is a `possible world’ if it is possible according to the laws of nature. The collection of all possible worlds is denoted by $\mathcal W$. A world $W\in\mathcal W$ is called (Laplacian) \defn{deterministic} if for any $W'\in\mathcal W$, if $W$ and $W'$ agree at some time $t$, then $W$ and $W'$ agree at all times. We could mathematically express this as:
\begin{equation}\label{eq:earman-determinism}
    \forall W'\in\mathcal W\,\forall t : W_t = W'_t \implies W = W',
\end{equation}
where $t$ ranges over the timepoints defined by the theory, e.g.\ $t\in\R$ (or $t\in\N$ in section~\ref{sec:lawless-indeterminism}).

A particular difficulty in defining determinism is to separate it from a closely related notion called \defn{predictability}, which is the idea that one can \emph{in principle} predict the future evolution of the world (with arbitrary small error) on the basis of sufficient knowledge about the present configuration. Pierre-Simon de Laplace, one of the first modern philosophers on determinism, seemed to identify determinism with predictability (`Laplace’s demon’ is named after him)\footnote{``An intelligence knowing all the forces acting in nature at a given instant, as well as the momentary positions of all things in the Universe, would be able to comprehend in one single formula the motions of the largest bodies as well as the lightest atoms in the world, provided that its intellect were sufficiently powerful to subject all data to analysis; to it nothing would be uncertain, the future as well as the past would be present to its eyes.’’ \cite{laplace1820theorie}}, as well as Karl Popper \cite{bishop2003predictability}. Some maintain that predictability implies determinism or the other way around, but most agree that such an implication does not exist. For example, reasons to believe that determinism does not imply predictability can be based on the \emph{in principle} limits on measurement (e.g.\ a measurement has finite precision) and information storage capacity (e.g.\ one can store only finitely many measurement results). Also chaotic systems play a role in this distinction: one might be able to approximate the state at some given future timepoint arbitrarily closely by using sufficiently precise measurements of the present state, but might not be able to approximate the state at \emph{arbitrary} future timepoints to the required precision without requiring infinitely precise measurement results (for small changes in the initial condition might blow up over time) \cite{bishop2003predictability}.

However, we will chiefly be concerned with determinism, which is an ontological notion and has little to do with limits on measurement and knowledge, even \emph{in principle} ones. An example of a theory which is usually considered deterministic (i.e.\ a theory in which all possible worlds are deterministic) is classical mechanics: as discussed in section~\ref{sec:hamiltonian-formalism-orthodox}, the solution through a given initial condition $x_0\in\M$ is unique (however, the fact that the maximal interval of existence of a solution might only extend up to a finite time raises questions about the determinism of classical mechanics; see e.g.\ \textcite[Ch.~III]{earman1986}). An example of a theory that is widely regarded as \emph{indeterministic}, on the other hand, is quantum mechanics, in which the outcome of a collapse of the wave function is not predetermined---but does satisfy probabilistic laws (propagation of the wave function itself via the Schrödinger equation is, however, completely deterministic. The existence of at least one event that is not predetermined is sufficient for a world to be indeterministic \cite{popper1950indeterminism}).

It can be shown, however, that for most theories, determinism cannot be verified nor falsified by experiment: for each deterministic theory, there exists an \emph{empirically} (or \emph{observationally}) \emph{equivalent} indeterministic theory which has the same empirical predictions, and vice versa \cite{werndl2009deterministic}. (Quoting \textcite{born1926quantenmechanik}: ``I myself tend to relinquish
determinism in the atomic world. But this is a philosophical question, for which physical arguments alone are not decisive.''\footnote{Translation by \textcite{landsman2020randomness}.})
Multiple deterministic theories of quantum mechanics have been developed, such as Bohmian mechanics, which supplements quantum theory with `hidden variables’ that determine the outcome of wave function collapse, but which cannot be directly measured \cite{sep-bohmian-mechanics}. In this thesis we show, conversely, that indeterministic theories of classical physics are also possible.

    \chapter{Intuitionistic mathematics}\label{app:intuitionism}

\lettrine[lines=3]{I}{ntuitionistic mathematics} was founded by the Dutch mathematician Luitzen Egber\-tus Jan (“Bertus”) Brouwer (1881--1966) in the early twentieth century in response to the rapid formalisation of mathematics at that time, and had as its goal to approach mathematics in a more intuitive or psychological way than the usual approach to mathematics (called ‘classical mathematics’ by intuitionists). Instead of reducing all of mathematics to a formal logical system of symbols and manipulation rules, intuitionistic mathematics is formulated from the point of view of the mathematician. Intuitionists à la Brouwer see mathematics as a `languageless activity of the mind’; only when writing down or communicating mathematical statements and proofs to other people are mathematicians in need of a language to express their ideas in. The patterns that arise in this symbolic language can then be seen as the rules of logic.\footnote{In Brouwer’s view, logic is not a foundational science, but an `observational' science \cite{brouwer1907}.} The logic that arises when studying intuitionism is (unfortunately) called \emph{intuitionistic logic}; however, formal logic is not needed to build up intuitionistic mathematics, and intuitionism cannot be reconstructed from logic only, since also arguments about the nature of the (idealised) mathematician are used. We will see examples of this later on this is appendix.

\section{Constructive mathematics}\label{sec:constructive-mathematics}
Intuitionistic mathematics is a form of constructive mathematics, which primarily differs from classical mathematics in that it only accepts the existence of an object when it can be constructed explicitly, as opposed to accepting it on the basis of e.g.\ a proof by contradiction. In particular, most variants of constructive mathematics adopt the \defn{BHK interpretation} (named after Brouwer, Arend Heyting and Andrei Kolmogorov) of the logical connectives and quantifiers (whenever it is convenient to use logical symbols at all). In this interpretation, the logical constants are not defined through truth tables, but through their use in proofs. A few examples are:
\begin{itemize}
	\item To prove $A\lor B$ we must either have a proof of $A$ or a proof of $B$.
	\item To prove $A\to B$ we must provide a method that converts any proof of $A$ into a proof of $B$.
	\item To prove $\neg A$, we must prove $A\to 0=1$.
\item To prove $\exists x A(x)$ we must construct an object $x$, together with a proof of $A(x)$.
\end{itemize}
See e.g.\ \textcite{dummett2000elements} or \textcite{TrvD} for more on the constructive interpretations of logical constants, and for more comprehensive introductions to intuitionistic mathematics.

From the above examples we see that, for instance, $\neg\forall n\,\neg P(n)$ in general does not imply $\exists n\,P(n)$ in constructive mathematics, whereas this implication does hold in classical mathematics. Indeed, constructive mathematics can be said to require \emph{positive evidence}, and that impossibility of negative evidence is not enough.

From this, as well as from the interpretation of $\lor$, we see that the \emph{principle of the excluded middle} (PEM), which is the assertion that $A\lor\neg A$ holds for all statements $A$, is constructively invalid. We call statements $A$ for which $A\lor\neg A$ does hold \defn{decidable}. PEM, which is equivalent to the principle of double negation elimination ($\neg\neg A\to A$), is one of the basic axioms of classical mathematics. In fact, intuitionistic logic, which arises from the BHK interpretation, is precisely classical logic without double negation elimination.

\

There are many variants of constructive mathematics, all of which follow the rules of intuitionistic logic \cite{sep-mathematics-constructive}. Errett Bishop’s constructive mathematics, for example, is a formalistic approach to mathematics which uses intuitionistic logic instead of classical logic \cite{bishop1967}. As noted before, however, intuitionistic mathematics à la Brouwer is not formalistic and is therefore more than just intuitionistic logic. As a result, intuitionism is only a form of constructive mathematics in the wider sense, since some intuitionistic results are not accepted by all constructivists. The intuitionistic philosophy of time (as experienced by the mathematician) and infinities, in particular, forms a distinguishing feature of intuitionism, and also suggests that intuitionism might be of use in physics. To fully appreciate the difference between intuitionistic and classical mathematics, therefore, let us consider infinite sequences.

\section{Natural numbers, infinite sequences and LPO}\label{sec:NN}
    


While set-theorists like Frege, Dedekind, Russell and Von Neumann defined the set $\N$ using set-theoretical considerations, Brouwer deduced the existence of $\N$ from iteration of the `basic intuition of two-ity' (the `falling apart of a life moment into two distinct things, one of which gives way to the other, but is retained by memory' \cite{brouwer1981} -- see also section~\ref{sec:intuitionism-time}). That is to say, given any finite number of distinct natural numbers, the passage of time allows one to think of another one. Whatever the precise construction of $\N$, the important difference between classical and intuitionistic mathematics is that in classical mathematics $\N$ is seen as a finished, fully-determined, actually infinite set, whereas in intuitionism we see $\N$ as only \emph{potentially infinite}: at every moment of time the mathematician can have only finitely many elements in mind, but realises that more elements can be appended to them. The intuitionist views the natural numbers as a forever unfinished \emph{project}, rather than a finished \emph{object}.

The same holds for any infinite sequence $\a = (\a(0),\a(1),\a(2),\dots)$ of natural numbers. The set $\N^\N$ of infinite sequences of natural numbers (Baire space) is denoted by $\NN$. We use Greek letters $\a,\b,\g$ to denote elements of $\NN$, and Latin letters $m,n$ to denote natural numbers.

Let us study one particular example of such an infinite sequence. Define $\g\in\NN$ by
\begin{equation}\label{eq:goldbach}
	\g(n) = \begin{cases} 
        0 & \text{ if } n=0; \\
        0 & \text{ if $2n = p+q$ for some pair of prime numbers $p,q$}; \\
        1 & \text{ if such $p,q$ do not exist.}
\end{cases}
\end{equation}

Note that $\g:\N\to\N$ defined in this way is a total function: for all $n$, either $\g(n) = 0$ or $\g(n) = 1$. There is, after all, a known finite procedure to determine whether $2n$ can be expressed as the sum of two primes: simply iterate over all pairs of natural numbers $p,q \leq n$, check if they are prime and if they add up to $2n$.

However, the statement $G := \forall n(\g(n) = 0)$ is equivalent to the Goldbach conjecture, an age-old problem in number theory which has neither been proven nor disproven to this day. Since an intuitionistic proof of $G\lor\neg G$ requires a proof of either $G$ or $\neg G$, the statement $G\lor\neg G$ cannot be said to be true: in other words, $G$ is not decidable. The sequence $\g$ therefore forms a so-called \defn{weak} or \defn{Brouwerian counterexample} to the following special case of the principle of the excluded middle, which \textcite{bishop1967} dubbed the \defn{limited principle of omniscience} (LPO):\footnote{A perhaps more general formulation of LPO would be: $\forall n (P(n)\lor \neg P(n)) \to \exists n (P(n)) \lor \forall n(\neg P(x))$, seen as a schema over all 1-ary propositions P with argument in $\N$. In our formulation of LPO, the premise that $\a(n)=0$ be decidable is always true: equality on the natural numbers is decidable, and $\a(n)$ is understood to be defined in a constructive way, i.e.~there is a finite procedure for finding the natural number which it represents.}
\begin{equation}\label{eq:lpo}
	\forall\a(\forall n(\a(n)=0) \lor \exists n(\a(n)\neq 0)).
\end{equation}

Unproved statements like $G\lor\neg G$ are sometimes called  \emph{vermetel}, which, following \textcite{veldman2020treading}, we translate as \defn{reckless}. These weak counterexamples and reckless statements are often used to show that certain classically acceptable statements, such as LPO, are unacceptable in constructive mathematics. For if LPO were true, we would be able to either prove or disprove the Goldbach conjecture: LPO has \emph{reckless consequences}.

An important question to address is how classical mathematicians have come to accept PEM as a generally valid principle. Brouwer believes that this is caused by the fact that we tacitly generalise our experience of handling finite sets to infinite sets \cite{veldmannotes}.
Indeed, if $s = (s(0), \dots, s(k))$ is a finite sequence of natural numbers then we can prove
\[ \forall n\leq k (s(n)=0) \lor \exists n\leq k (s(n)\neq 0) \]
by simply checking whether $s(n)=0$ for $n\leq k$, since equality of natural numbers is decidable. This does not tell us, however, that we can conclude the same about infinite sequences.\footnote{A similar argument applies to infinite precision of physical quantities: see section~\ref{sec:how-have-we-come-to-orthodox}.}

\

LPO has many reckless consequences: for example, it implies decidability of many well-known unsolved mathematical problems, such as the Goldbach conjecture, the odd perfect number problem, the twin prime problem, the Riemann hypothesis and the $abc$ conjecture \cite{mandelkern1989}. Since it is possible that these will be proven at one day (decidability cannot be disproven), these individual weak counterexamples do not suffice to \emph{prove} $\neg$LPO;\footnote{Some historically significant works on constructivism base their discussion on Fermat's last principle, which was proven by Andrew Wiles in 1995.} however, the quantifier $\forall\a$ leads us to think that LPO is not intuitionistically acceptable. Indeed, there is no proof (read: finite procedure) that can be used to decide whether $\forall n(\a(n)=0)$ or $\exists n(\a(n) \neq 0)$ \emph{for arbitrary} $\a$ (by the insolubility of the \emph{Entscheidungsproblem}). Therefore, Brouwer set out to make explicit that LPO is inconsistent.

\section{Choice sequences and the continuity principle}\label{sec:choice-sequences-continuity}
While Brouwer's thesis from 1907 laid the foundations for intuitionistic mathematics, it was only in 1917 that he introduced to intuitionism one of its most defining features, the notion of a \emph{choice sequence}. The continuity principle, which follows naturally from this notion, can be used to prove $\neg$LPO.\footnote{In fact, the main motivation for Brouwer to introduce choice sequences originated in considerations about the real number continuum. We will focus on LPO first because it is more general. We will discuss the real numbers in section~\ref{sec:int-reals}.}

As stressed before, an infinite sequence of natural numbers is seen as a developing project: at each point in time, only finitely many digits are known to the mathematician. We have seen examples of sequences in which the elements follow a rule or law that is specified beforehand. For example, the $n$-th entry of the sequence $\N = \{0,1,2,\dots\}$ is obtained by adding one to zero $n-1$ times, and the entries of $\g = (\g(0),\g(1),\g(2),\dots)$ defined as in (\ref{eq:goldbach}) are found by determining whether $2n$ can be written as a sum of two primes. These sequences are called \defn{lawlike} sequences and are in some sense predetermined, even though, being infinite, they are projects that are never finished.

To prove $\neg$LPO, Brouwer realised that it is necessary to also consider \defn{choice sequences}, which can be seen as projects that are \emph{not necessarily predetermined} by a law or algorithm, but produce elements of the sequence in a step-by-step manner. There can be several ways in which the successive elements are chosen, but the mathematician may not know what process this is; at any given moment, all that needs to be known is a finite initial segment of the sequence under consideration. In section~\ref{sec:lawless-sequences-mathematics} We will discuss one particular type of choice sequence, the lawless sequence.

Now, suppose we are given a relation $R \subseteq\NN\times\N$ such that $\forall\a\exists n(\a R n)$ holds, i.e.\ given any $\a$ we can find an $n$ and prove that $\a R n$. Here $\a$ can also be a choice sequence, which is revealed to us in a step-by-step manner; hence, whenever we decide that $\a R n$, we must be able to do so on the basis of only the finite initial segment of $\a$ that is known to us at that time. This reasoning leads to \defn{Brouwer's continuity principle (BCP)}, which is a schema over all relations $R\subset\NN\times\N$:
\begin{equation}\label{eq:bcp}
	\forall\a\exists n (\a R n) \longrightarrow \forall\a\exists n\exists m\forall\b (\o\b m = \o\a m \to \b R n),
\end{equation}
where $\o\a m$ denotes the initial segment of length $m$ of the infinite sequence $\a$.

We already saw that LPO has reckless consequences, and hence cannot be intuitionistically true. But now, as promised, we can use BCP to show that LPO is even \emph{false}:
\begin{theorem}
	BCP $\implies \neg$LPO.
\end{theorem}
\begin{proof}
    \def\zeroes{\underline{0}}
    Assuming LPO~\eqref{eq:lpo}, we must be able to decide for any $\a$ whether $\forall n(\a(n)=0)$ or $\exists n(\a(n)\neq 0)$ (following the constructive interpretation of $\lor$). Intuitively, we see this leads to a contradiction because we should be able to decide between these two options on the basis of only a finite initial segment of $\a$, which is clearly not possible. Formally, note that
    \[ \forall \a\exists i[ (i=0 \land \forall n(\a(n)=0)) \lor (i\neq0 \land \exists n(\a(n)\neq0))] \]
    holds. We define the relation $R\subseteq \NN\times\N$ such that $\a R i$ holds if and only if the expression between square brackets above holds. Taking $\a = \zeroes \equiv (0,0,\ldots)$, an infinite sequence of 0s, and applying BCP~\eqref{eq:bcp}, find $i,m\in\N$ such that $\forall\b (\o\b m = \o{\zeroes} m \to \b R i)$. In particular, we have $\zeroes R i$. We have $i\neq0\lor i= 0$ since equality of natural numbers is decidable. $i\neq 0$ leads to $\exists n(\zeroes(n)\neq 0)$, which is clearly contradictory. If on the other hand $i=0$ then we can conclude that $\forall\b (\o\b m = \o{\zeroes} m \to \forall n(\b(n)=0))$. This is also not possible, as can be seen by e.g.\ taking the sequence $\b$ that is everywhere 0 except in its $m+1$-th element. Hence, LPO leads to a contradiction.
\end{proof}

Note that we have now proven that $\neg(\forall\a(\forall n(\a(n)=0) \lor \exists n(\a(n)\neq 0)))$, but that we can never give a `strong' counterexample: that is, we cannot provide an $\a$ such that $\neg(\forall n(\a(n)=0) \lor \exists n(\a(n)\neq 0))$, because this would lead to a contradiction. We see that intuitionistic mathematics including BCP is truly incompatible with classical mathematics, as LPO is classically valid.

\section{Real numbers}\label{sec:int-reals}
\paragraph{Construction of the reals} A significant part of Brouwer's motivation to develop intuitionism lied in creating a more intuitive notion of the continuum. The continuum is also of great importance to physics and in this thesis it forms the main reason to consider using intuitionism for physics. In this section, we will translate the results about $\NN = \N^\N$ in the previous section to results about the set of intuitionistic real numbers $\R$, which can be seen as a subset of $\NN$.

The construction of the integers and rationals goes in much the same way as in classical mathematics; still, it is important to present the entire procedure because we have to stick to the intuitionistic philosophy along the way. In constructing the intuitionistic reals, we follow \textcite{veldmannotes}.\footnote{This approach is akin to Brouwer's own formulation; see \textcite[p69]{brouwer1992intuitionismus}. See also e.g.\ \textcite{driessenintuitionistic} for a publicly available reference exhibiting the approach in \textcite{veldmannotes}.} We will later briefly discuss some other approaches to constructing the reals.

\

Given $\N$, we can construct $\N^2$, and from there any power $\N^d$, by using a bijection $K:\N\xrightarrow{\sim}\N\times\N$ (which, of course, should be defined in a constructive way, like its inverse). In this way, every natural number $n$ encodes a pair of naturals $K(n)$, which we denote by $(n',n'')$. We often identify $n$ with $K(n)$, i.e. $n = (n',n'')$.

Next, the integers and rationals are defined in the usual way. The integers $\Z$ are seen as pairs of natural numbers $n = (n',n'')$ subject to the equivalence relation $=_\Z$ defined by $(n',n'') =_\Z (m',m'')$ iff $m' + n'' = m'' + n'$. Arithmetic operations such as addition $+_\Z$ and multiplication $\cdot_\Z$ are defined as usual, as are relations such as $<_\Z$ and $\leq_\Z$. Next, the rationals are pairs of integers $p = (p',p'')$, with $p'' \neq 0$, subject to the equivalence relation $(p',p'') =_\Q (m',m'')$ iff $p' \cdot_\Z q'' = p'' \cdot_\Z q'$. We again define $+_\Q$, $\cdot_\Q$, $<_\Q$ and $\leq_\Q$ in the expected way.

Now, we define the set of \defn{rational segments} $\S$ as consisting of all pairs $s = (s',s'')$ of rational numbers such that $s' \leq_\Q s''$. We define the following relations for $s,t\in\S$:
\begin{itemize}
    \item $s \leq_\S t$ iff $s' \leq_\Q t''$;
    \item $s <_\S t$ iff $s'' <_\Q t'$;
    \item $s \sqsubseteq_\S t$ iff $t'\leq_\Q s' \leq_\Q s'' \leq_\Q t''$ ($t$ \defn{covers} $s$);
    \item $s \approx_\S t$ iff $s\leq_\S t \land t\leq_\S s$ ($s$ and $t$ \defn{overlap}).
\end{itemize}

Note that $\S\cong\N$, so that every rational segment $s$ can be encoded with a natural number; this means that any given rational segment (as well as any rational number and integer) can be seen as a \emph{finished object}. Next, we define the real numbers as infinite sequences of rational segments:

\begin{definition}\label{def:intuitionistic-reals}
	A sequence $\a\in\S^\N$ is called a \defn{real number} if
    \begin{enumerate}
    	\item[(i)] $\a$ \emph{shrinks}: $\forall n (\a(n+1) \sqsubseteq_\S \a(n))$;
    	\item[(ii)] $\a$ \emph{dwindles}: $\forall m\exists n(\a''(n)-\a'(n)) \leq_\Q \inv{m}$,
    \end{enumerate}
    where $\a'(n) := (\a(n))'$ and $\a''(n) := (\a(n))''$, such that $\a(n) = (\a'(n),\a''(n))$. We define the following relations on the real numbers:
    \begin{itemize}
    	\item $\a =_\R \b$ iff $\forall n(\a(n)\approx_\S\b(n))$ ($\a$ and $\b$ \defn{coincide});
    	\item $\a \leq_\R \b$ iff $\forall n(\a(n) \leq_\S \b(n))$;
    	\item $\a <_\R \b$ iff $\exists n(\a(n) <_\S \b(n))$;
    	\item $\a \neq_\R \b$ iff $\neg(\a =_\R \b)$
    	\item $\a \apart_\R \b$ iff $\exists n \neg(\a\approx_\S\b)$ ($\a$ and $\b$ are \defn{apart}).
    \end{itemize}
    The set of real numbers $\R$ is the set of equivalence classes of real numbers relative to $=_\R$.
\end{definition}

Some remarks are in order.
\begin{itemize}
    \item All relations defined above respect the relation of coincidence: that is, if $\a=_\R\a'$ and $\b=_\R\b'$, then $\a\leq_\R\b \leftrightarrow \a'\leq_\R\b'$.
    \item The rationals can be embedded into the reals by identifying $q\in\Q$ with the sequence $\a: n\mapsto (q,q)$, and in this embedding all defined relations on the rationals translate to the corresponding relations on the reals.
    \item Recall that in the intuitionistic philosophy, infinite sequences, including ones defining real numbers, are seen as only forever unfinished projects rather than finished objects. This means that at any point in time, the mathematician knows the real number only up to the precision of a rational segment.
    \item However, because intuitionistic mathematics is constructive, it is assumed that $n$ in (ii) can be constructed explicitly from $m$; hence, we can pinpoint the location of a real number on the number line with \emph{arbitrary} precision. To emphasise this, we could also define a real number to be a pair of (`potentially') infinite sequences $\a\in\S^\N$, $\mu\in\N^\N$ such that (i) is satisfied (ii) is replaced by
    \begin{itemize}
        \item[(ii$'$)] $\forall m(\a''(\mu(m))-\a'(\mu(m))) \leq_\Q \inv{m}$.
    \end{itemize}
    \item Classically, $\neq_\R$ and $\apart_\R$ are equivalent. This is not true constructively: the fact that we can prove that $\forall n(\a(n)\approx_\S\b(n))$ is contradictory does not mean that we can constructively find an $n$ such that $\neg(\a\approx_\S\b)$. $\apart_\R$ is a \emph{positive} notion, while $\neq_\R$ is a weaker \emph{negative} notion.
    \item Equality $=_\R$ of two reals is in general undecidable (as is any other of the relations defined above), and hence the equivalence classes of $\R$ are undecidable. We will discuss an example below.
    \item The set of real numbers is not countable (contrary to what one would perhaps expect from constructive mathematics). It is even \emph{positively uncountable}, as the following theorem shows.
\end{itemize}

\begin{theorem}
	$\R$ is positively uncountable, that is: for any infinite sequence $\a_0,\a_1,\dots$ or real numbers, we can construct a real number $\a$ such that for all $i\geq 0$, $\a\apart_\R\a_i$.
\end{theorem}
\begin{proof}
	We proceed along a Cantor-like argument. First of all, for a rational segment $s\in\S$, divide $s$ into three even parts and define $L_3(s)$ and $R_3(s)$ to be the left- and rightmost parts, respectively. (That is, $L_3(s) = (s', \frac23 s' + \frac13 s'')$ and $L_3(s) = (\frac13 s' + \frac23 s'', s'')$.) We define the elements of $\a$ one by one. First define $\a(0) = (0,1)$; next we recursively define the remaining $\a(n)$ such that $\a''(n)-\a'(n) = \inv n$ for $n\geq 1$. Suppose that $\a(n)$ has been defined; find $m$ such that $\a_n''(m) - \a_n'(m) < 3^{-(n+1)}$. Such an $m$ exists and can be found because of the axiom that $\a_n$ \emph{dwindles}. Note that $\a_n(m)$ overlaps at most one of $L_3(\a(n))$ and $R_3(\a(n))$. We want $\a$ to be apart from $\a_n$; therefore, define $\a(n+1) = L_3(\a(n))$ if $\a_n(m)$ does not overlap $L_3(\a(n))$, and $\a(n+1) = R_3(\a(n))$ otherwise. $\a$ constructed in this way satisfies $\a\apart_\R\a_i$ for every $i\geq 0$.
\end{proof}

Note how the above proof works: we construct $\a$ step by step, and at each step we need to know only finite initial segments of finitely many $\a_i$.

\paragraph{Other constructions of the reals}\label{sec:int-reals-other-constructions}
The above is only one possible way to construct the reals from the rationals. Another example is identical to the one based on Cauchy sequences usually used in classical mathematics. We say that a rational sequence $\a\in\Q^\N$ has the \emph{Cauchy property} if
\[ \exists\mu\in\NN \forall m,n > \mu(k)\,(|\a(n)-\a(m)| < 1/k). \]
The sequence $\mu$ is called the \emph{Cauchy modulus} of $\a$. Similarly to the fourth bullet above, to prove that a sequence $\a$ is Cauchy, also its Cauchy modulus should be constructed, and conversely, when proving properties of Cauchy sequences, also their Cauchy modulus can be explicitly used in proofs. This approach based on Cauchy sequences is almost identical to the one using rational segments.

In addition, one could try to define the reals on the basis of their decimal (or binary) expansions; that is, we define a real number as an infinite sequence of natural numbers below 10, together with a natural number indicating the position of the decimal point. However, such a definition is fruitless in constructive mathematics, as the obtained set of real numbers is, for example, not closed under addition. Suppose, namely, that we define a real $x = 0.999\ldots$ which has as its $n$-th digits a $9$ if $\g(n) = 0$ and a 0 if $\g(n) = 1$, where $\g\in\NN$ is defined as in~\eqref{eq:goldbach}. Moreover, define $y = 0.000\ldots$, which has as its $n$-th digit the value $\g(n+1)$. Then $x+y$ cannot be constructively defined through its decimal expansion, for its unit digit depends on the validity of the Goldbach conjecture, for which there exists no known proof. In other words, the addition operation is not a total real function $\R\times\R\to\R$. (Cf.\ the discussion below Definition~\ref{def:some-representations}.)

\paragraph{Undecidability of relations on the reals}
In section~\ref{sec:NN} we defined a sequence $\g\in\NN$ such that the statement $\forall n(\g(n) = 0) \lor \neg\forall n(\g(n) = 0)$ is reckless. We can turn this into a statement about real numbers by defining the number $\rho^\g$ as follows:
\begin{equation}\label{eq:rho-gamma}
	\rho^\g(n) = \begin{cases}
		(-\inv n, \inv n) & \text{ if } \forall k\leq n (\g(k) = 0); \\
		(\inv k_0,\inv k_0) & \text{ if } \exists k\leq n (\g(k) \neq 0), \text{ where $k_0$ is the least $k$ such that $\g(k) = 0$}.
	\end{cases}
\end{equation}
We see that $\forall n(\g(n)=0) \leftrightarrow \rho^\g(n) \leq 0$, and $\exists n(\g(n) \neq 0) \leftrightarrow \rho^\g > 0$.\footnote{From now on, we will omit the $\R$-subscripts on relations and operations. $0 \equiv 0_\R$ is the sequence formed by repeating the rational segment $(0_\Q, 0_\Q)$.} Therefore, $\rho^\gamma$ is a well-defined real number for which the statement
\[\rho^\g \leq 0 \lor \rho^\g > 0\]
does not hold (i.e.\ it is a reckless statement). We say that $\rho^\g$ \defn{floats around} 0.

Defining a real number $\rho^\a$ for every sequence $\a\in\NN$ in the way of Equation~\eqref{eq:rho-gamma}, we see that the statement
\begin{equation}\label{eq:lpo-reals}
    \forall \rho\in\R\,(\rho \leq 0 \lor \rho > 0).
\end{equation}
implies LPO, and is therefore strictly false under BCP. The statement can in fact be shown to be equivalent to LPO \cite{mandelkern1989}. In particular, we see that \emph{equality on the reals is undecidable}: for $\rho,\s\in\R$, we can in general not decide whether $\rho=_\R\s$ or $\rho\neq_\R\s$.

The \defn{lesser limited principle of omniscience} (LLPO) is a weaker variant of LPO which is also reckless in intuitionistic mathematics but a tautology in classical mathematics. Similarly to LPO, LLPO is false under BCP, and it is equivalent to the statement\footnote{In the study of constructive mathematics it is often insightful to characterise reckless statements by their equivalence to a small set of reckless principles, such as PEM, LPO, LLPO and \emph{Markov’s principle} MP. The programme of exploring these equivalences is called \emph{constructive reverse mathematics}.}
\begin{equation}\label{eq:llpo-reals}
    \forall\rho\in\R\,(\rho\leq 0\lor\rho\geq 0).
\end{equation}

\paragraph{Spreads}\label{sec:spreads}
Let us briefly introduce the general intuitionistic notion of \emph{spread}, of which the real numbers can be seen as a special case. We first need some notation: denote by $\N^*$ the set of all finite sequences $s$ of natural numbers, by $\langle\,\rangle$ the sequence of length 0 and by $\langle m_1, m_2,\ldots,m_\ell\rangle$ the sequence of length $\ell$ containing the elements $m_1,\ldots,m_\ell$. The concatenation operation is denoted by $*$: e.g.\ $\langle a,b\rangle * \langle c\rangle = \langle a,b,c\rangle$. Finally, remember that for an infinite sequence $\a \in\NN$, $\o\a n = \langle \a(0),\ldots,\a(n-1)\rangle$ is the initial segment of $\a$ of length $n$.

A \defn{spread law} is an element $\s\in\NN$ for which
\begin{enumerate}[(i)]
    \item $\s(\langle\,\rangle) = 0$ and
    \item for all $s\in\N^*$, $\s(s) = 0$ if and only if $\exists n(\s(s * \langle n\rangle) = 0)$.
\end{enumerate}
If $\s(s) = 0$ for $s\in\N^*$, we say that the finite sequence $s$ is \defn{admitted by} the spread law $\s$. We call an infinite sequence $\a\in\NN$ \defn{admitted by} $\s$ if all its initial segments are admitted by $\s$, i.e.\ $\forall n(\s(\o\a n) = 0)$. The set $\mathcal F_\s\subseteq\NN$ of all infinite sequences admitted by $\s$ is called the \defn{spread} determined by $\s$. Hence, a spread is essentially a tree which has at least one branch (by (i)) and no leaves (by (ii)).

The real number field $\R$ can be seen as a spread if we replace (ii) in Definition~\ref{def:intuitionistic-reals} by the assumption that a real $\a$ dwindles by at least some predefined rate, e.g.
\begin{itemize}[(ii$''$)]
    \item $\forall n (\a''(n)-\a'(n)) \leq \inv n$.
\end{itemize}
(Namely, define the spread law $\s$ to admit a sequence $s\in\N^*$ iff its elements encode a finite sequence of rational segments that shrink and dwindle by the predefined rate.) Modulo the equivalence relation $=_\R$, this yields the same set $\R$ as Definition~\ref{def:intuitionistic-reals}.

\section{Real functions}\label{sec:int-real-functions}
In intuitionism, a \emph{real function} $f:\R\to\R$ is a constructive method to define a real number $f(\a)$ given a real number $\a$, and for which
\[ \forall \rho,\sigma\in\R\,(\rho =_\R \s \to f(\rho) =_\R f(\s)) \]
holds.

\paragraph{Analytic theorems and approximate variants}
As you might imagine, the invalidity of \eqref{eq:lpo-reals} has some profound impacts on the general validity of theorems in classical analysis. A notorious example of a theorem that fails to hold intuitionistically is the \emph{intermediate value theorem}, which states:
\begin{quote}
	For any continuous function $f$ from $[0,1]$ to $\R$ with $f(0) = 0$ and $f(1) = 1$, we have $\forall \rho\in[0,1]\exists \s\in [0,1] (f(\s) = \rho).$
\end{quote}
\begin{proposition}
    The intermediate value theorem is reckless; in particular, it implies $\forall \rho\in\R(\rho \leq 0 \lor \rho \geq 0)$, which is equivalent to LLPO and therefore even false under BCP.
\end{proposition}

\begin{wrapfigure}{o}[1cm]{0pt}
    \centering
    \fbox{\begin{tikzpicture}[scale=0.7,x=1cm, y=1cm]%
        \draw[->] (0,0) node[anchor=east] {\footnotesize 0} -- (0,3) node[anchor=east] {\footnotesize 1};%
        \draw[->] (0,0) node[anchor=north] {\footnotesize 0} -- (3,0) node[anchor=north] {\footnotesize 1};%
        \draw[dotted] (0,1.5) node[anchor=east] {\footnotesize $\dfrac12$} -- (3,1.5);%
        \draw (0,0) -- (1,1.6) -- (2,1.6) -- (3,3);%
        \node[anchor=center] at (2.3,2.4) {\footnotesize $f$};%
        \node[anchor=center, align=left] at (1.5,-1.2) {\footnotesize The function $f$\\\footnotesize for small positive $\rho$.};
    \end{tikzpicture}%
    }
\end{wrapfigure}

\noindent\textit{\proofname}. Assume the intermediate value theorem and let $\rho\in\R$ be given. Define a function $f:[0,1]\to\R$ such that $f(0)=0$, $f(\frac{1}{3}) = f(\frac{2}{3}) = \inv2 + \rho$, $f(1) = 1$ and $f$ is linear between these points. This is a well-defined intuitionistic function since the function value $f(\s)$ can be approximated arbitrarily closely by approximating $\s$ and $\rho$ arbitrarily closely, for $\s\in\R$. Applying the intermediate value theorem, find $\s$ such that $f(\s) = \inv2$. By finding $m$ such that $0 \leq \s''(m) - \s'(m) < \inv3$, we can decide that either $\s < \frac23$ or $\s < \frac13$. If $\s < \frac23$, then $\rho \geq 0$ because $f(\frac23) = \inv2 + \rho$ and $f$ is non-decreasing (see the figure). Similarly, if $\s < \frac13$ then $\rho \leq 0$. We conclude that $\rho \leq 0 \lor \rho\geq 0$ for any $\rho\in\R$. For the proof that $\forall \rho\in\R\,(\rho \leq 0 \lor \rho \geq 0)$ is equivalent to LLPO, see \textcite[Theorem~5]{mandelkern1989}.\hfill\qed

\

Many other classically valid analytical theorems are intuitionistically invalid, such as the Bolzano-Weierstrass theorem (equivalent to LPO) and, ironically, Brouwer’s fixed point theorem (equivalent to LLPO) \cite{mandelkern1989}.\footnote{As Brouwer himself well realised. Brouwer has made some significant contributions to classical topology; he understood that to be taken seriously in his unusual approach to mathematics, he first had to prove himself an able mathematician.} This does not mean all of analysis is lost: several `approximate’ variants of these principles hold constructively. The \emph{approximate intermediate value theorem}:
\begin{quote}
	For any continuous function $f$ from $[0,1]$ to $\R$ with $f(0) = 0$ and $f(1) = 1$, we have $\forall\rho\in[0,1] \forall n\exists\s\in [0,1] (|f(\s) - f(\rho)| < \inv n)$,
\end{quote}
for example, is classically equivalent to the intermediate value theorem but is valid intuitionistically.

\paragraph{Continuity of real functions}
From the way a real function is defined in intuitionism, it already appears that every real function is continuous. We can make this explicit by using a version of Brouwer's continuity principle applied to real numbers, which we denote by BCP$_\R$. It states that for every real relation $R\subseteq \R\times\N$, we have\footnote{BCP$_\R$ can be deduced from intuitive argument or from a generalised version of BCP which holds for spreads.}
\begin{equation}\label{eq:bcp-r}
    \forall\rho\in\R\exists n (\rho R n) \longrightarrow  \forall\rho\in\R\exists n\exists m \forall\s\in\R \Big(\left|\s-\rho\right| < \inv{m} \to \s R n\Big).
\end{equation}
In words: if we are able to find and $n$ such that $\rho R n$ for \emph{any} $n$, then we must be able to do so on the basis of some finite approximation of $\rho$.

BCP$_\R$ can be used to prove:
\begin{theorem}
    Every real function $f:\R\to\R$ is continuous.
\end{theorem}

\section{The role of time in intuitionism}\label{sec:intuitionism-time}
We have already seen multiple ways in which the notion of time enters into intuitionism. In this section we discuss these in more detail. The importance of time clearly resounds in Brouwer’s two ‘acts of intuitionism’, formulated in his Cambridge lectures in the 1940s. The first act is:
\begin{quote}
	\textsc{First act of intuitionism}\ \ \textit{Completely separating mathematics from mathematical language and hence from the phenomena of language described by theoretical logic, recognizing that intuitionistic mathematics is an essentially languageless activity of the mind having its origin in the perception of a move of time. This perception of a move of time may be described as the falling apart of a life moment into two distinct things, one of which gives way to the other, but is retained by memory. If the twoity thus born is divested of all quality, it passes into the empty form of the common substratum of all twoities. And it is this common substratum, this empty form, which is the basic intuition of mathematics.} \cite{brouwer1981}
\end{quote}
Hence, the notion of ‘time’ central to intuitionism refers to the time experienced by the mathematician, and can also be characterised as `intuitive’ or `subjective’ time. The first act gives rise to the natural numbers (section~\ref{sec:NN});in addition, the realisation that the mathematician needs time to prove theorems and do calculations already implies that the principle of the excluded middle does not hold, even with respect to lawlike sequences such as the one given in \eqref{eq:goldbach} \cite{sep-intuitionism}.

The second act of intuitionism expresses the acceptance of (free) choice sequences, and gives rise to the existence of the continuum:
\begin{quote}
	\textsc{Second act of intuitionism}\ \ \textit{Admitting two ways of creating new mathematical entities: firstly in the shape of more or less freely proceeding infinite sequences of mathematical entities previously acquired} (so that, for example, infinite decimal fractions having neither exact values, nor any guarantee of ever getting exact values are admitted); \textit{secondly in the shape of mathematical species, i.e.\ properties supposable for mathematical entities previously acquired, satisfying the condition that if they hold for a certain mathematical entity, they also hold for all mathematical entities which have been defined to be 'equal' to it, definitions of equality having to satisfy the conditions of symmetry, reflexivity.} \cite{brouwer1981}
\end{quote}
The ‘step-by-step’ generation of choice sequences, and the fact that they introduce a form of `indeterminism’ to intuitionism in the sense that their elements are not necessarily fixed in advance by a law, makes the importance of time in intuitionism even more explicit. However, choice sequences are not necessary to realise that PEM does not hold intuitionistically (as can also be seen from the fact that no other form of constructive mathematics accepts the existence of choice sequences).

\

Brouwer, at least in his early writings, indicates that the `intuitive time’ central to intuitionism should be clearly distinguished from `scientific time’. He considers the first to be the only \emph{aprioristic} element of science (the only ``necessary prerequisite for the possibility of science’’), while the second only appears \emph{a posteriori} from experience and can be reduced to a one-dimensional parameter serving to `catalogue observed phenomena' \cite{brouwer1907}. On the other hand, it is tempting to compare the distinction between lawlike and lawless (i.e.\ free choice sequences) sequences to deterministic and indeterministic processes in physics, which are solely related to Brouwer’s `scientific time’.\footnote{Unless one associates indeterminism to free will and the psychological arrow of time, in which `intuitive time’ might come into play.} However, in Brouwer’s intuitionism, both lawlike sequences and choice sequences are subject to the `intuitive time’,\footnote{Many intuitionists maintain, moreover, that the collection of lawlike sequences is not clearly separable from the collection of (free) choice sequences, meaning that these are no `well-circumscribed’ sets and that it is only meaningful to regard the continuum (or Baire space $\N$) as a whole. We will return to this in section~\ref{sec:lawless-sequences-mathematics}} which suggests that this comparison might not be justified.\footnote{Moreover, Oscar Becker, a contemporary of Brouwer, has compared the distinction between lawlike sequences and choice sequences in intuitionism to Heidegger’s distinction between `natural time’ and `historical time’, where the former is `time as measured’ and the latter can be seen as the cultural phenomenon of passage from one generation to another \cite{roubach2005being}. It is unclear whether physical indeterminism plays a role in this comparison, but because `natural time’ presumably coincides with Brouwer’s `scientific time’, also this comparison can be questioned from an intuitionistic perspective.}

\

In the previous sections, we have seen that some of Brouwer’s arguments for intuitionistic principles, such as his continuity principle~\eqref{eq:bcp}, are based on reasoning \emph{about} the mind of the mathematician; here an \emph{idealised} mind is meant, for Brouwer considered intuitionistic mathematics to be independent of psychology \cite{sep-intuitionism}. Brouwer called this idealised mind the \emph{creating subject}. A formal theory of the creating subject was designed by Kreisel \cite{kreisel1967informal} (and later Troelstra) in order to formalise and analyse Brouwer’s creating subject arguments. It concretises the temporal aspect of intuitionism by introducing the notation $\square_n P$, which should be interpreted as saying that the creating subject has a proof of or `experiences the truth of’ statement $P$ at time $n$, together with some axioms. $n$ ranges over a discrete (countable) set, representing the `stages’ of the creating subject’s arguments \cite{vanAtten2018bks}.

The creating subject also enters into intuitionism via arguments based on sequences given by the \emph{Brouwer-Kripke schema} (BKS), proposed by Brouwer and formalised by Saul Kripke \cite{vanAtten2018bks}. The schema asserts that for every proposition $P$, one can construct a sequence $\a\in\NN$ such that $\a(n) = 0$ if at stage $n$ (or `day $n$’), the creating subject has not proven $P$, and $\a(n) = 1$ if it has.\footnote{It is an axiom of the theory of the creating subject that $\square_n(P)$ is decidable. Hence, BKS can be expressed as $\exists \a(\a(n) \neq 0 \leftrightarrow \square_n P)$, or, not using the creating subject formalism, $\exists\a(\exists n(\a(n)\neq 0) \leftrightarrow P)$, for all $P$ \cite{sep-intuitionism,vanAtten2018bks}.} Such a sequence is perhaps the clearest example showing that an infinite sequence develops in time and is never a static, finished object. Furthermore, BKS sequences are sometimes seen as lawlike, which casts further doubts on the comparison between lawlike sequences and determinism in physics alluded to before.

\section{Lawless sequences}\label{sec:lawless-sequences-mathematics}
In section~\ref{sec:choice-sequences-continuity} we already discussed the notion of \emph{lawlike} sequences and \emph{choice sequences}, the latter being sequences of which the elements are revealed one by one, not necessarily according to a law. What is often presented as perhaps the simplest example of a choice sequence is a \emph{lawless sequence} or \emph{(absolutely) free choice sequence} of natural numbers. This kind of sequences was first considered by Georg Kreisel.\footnote{The term \emph{lawless} was coined by Gödel~\cite{kreisel1968}.} In his words, lawless sequences
\begin{quote}
	are those where the simplest kind of restriction on restrictions is made, namely some finite initial segment of values is prescribed, and, beyond this, no restriction is to be made. \cite{kreisel1968}
\end{quote}
However, trying to more precisely characterise the set of lawless sequences has turned out far from simple, and attempts to do this are surrounded by controversy. Needless to say, there is a strong connection between lawless sequences and mathematical notions of randomness as well as indeterminism in physics. The latter connection is be discussed in section~\ref{sec:lawless-indeterminism}. Lawless sequences in intuitionism are generally associated with Kreisel and Troelstra’s formalisation; therefore, we briefly discuss their approach, even though it is considered controversial in some respects. Before doing this, however, let’s take a closer look at lawlike ones.

\subsection{Lawlike sequences}\label{sec:lawlike-sequences}
Lawlike sequences are roughly those sequences that are ‘completely determined [and] fully described’ \cite{TrvD} or ‘fixed by a recipe’. What exactly a ‘law’ or ‘recipe’ means is a matter of interpretation.

The definition of ‘law’ that appeals to most mathematicians is a Turing-computable (recursive) law. In this interpretation, ‘lawlike’ is seen as synonymous to ‘recursive’, i.e.~the lawlike sequences are precisely the recursive sequences. This identification can be expressed by an axiom called \emph{Church’s thesis}.\footnote{Not to be confused with the \emph{Church-Turing thesis}, nor with the \emph{physical Church-Turing thesis}, which are discussed in Appendix~\ref{sec:church-turing-computability}. For a formal expression of Church’s thesis, see \textcite[§4.3.1]{TrvD}.}
Lawlikeness can also be taken to refer to a broader concept, however. An example of a class of sequences that are not recursive, but are still sometimes said to be lawlike, are the sequences given by the Brouwer-Kripke schema (BKS), introduced in the previous section. While the elements of a BKS sequence are indeed ‘fixed by a recipe’, they are less ‘predetermined’ than recursive sequences.

Finally, it is interesting to note that while under most notions of lawlikeness, the collection of lawlike sequences is countable from the classical point of view, no lawlike enumeration exists.\footnote{This can be proven classically if `lawlike' is taken to mean `recursive'; \textcite{moschovakis1993} argues that this holds for any notion of lawlikeness.} Indeed, to Brouwer, the collection of lawlike sequences was `\emph{countably unfinished}' (`\emph{aftelbaar onaf}')
; however, he argued that the collection was still too small to build up the continuum; in particular, to fully justify his continuity principle \eqref{eq:bcp}, he needed to introduce the notion of \emph{choice sequence}. Georg Kreisel~\cite{kreisel1968} later derived from this the notion of \emph{lawless sequence}.

\subsection{Intensional lawlessness; Kreisel and Troelstra’s formalisation}\label{sec:lawless-kreisel-troelstra}
Kreisel \cite{kreisel1968}, and subsequently Troelstra \cite{troelstra77}, developed an intuitionistic theory of lawless sequences which is based on an \defn{intensional} concept of lawlessness: this means that not the \emph{elements} of the sequences themselves are what characterise lawlessness, but the \emph{process} by which the elements are determined. Properties of sets that only depend on their elements, on the other hand, are called \defn{extensional} properties.\footnote{Cf. the \emph{extensionality axiom} of ZF, which states that two sets are equal iff they have the same elements.} \textcite[§4.6.2]{TrvD} explain the concept of lawless sequence (or more precisely, proto-lawless sequence, see below) perhaps somewhat more clearly than Kreisel:
\begin{quote}
	A lawless sequence is to be viewed as a process (say $\a$) of choosing values $\a0,\a1,\ldots\in\N$ such that at any stage in the generation of $\a$ we know only finitely many values of $\a$. If we generate a lawless sequence, then we have \emph{a~priori} decided not to make general restrictions on future values at any stage.
\end{quote}
Thus \textcite{TrvD} distinguish lawless sequences from what they call \emph{hesitant sequences}, which are processes of generating values ``such that at any stage we either decide that henceforth we are going to conform to a law in determining future values, or, if we have not already decided to conform to a law at an earlier stage, we freely choose a new value [of the sequence]’’.

Kreisel and Troelstra proposed four axioms LS1--LS4 for lawless sequences~\cite{troelstra77}. I will only discuss the first three axioms; especially the second is instructive because it illustrates the intensional nature of Kreisel and Troelstra’s notion.

The first axiom, the \emph{density axiom} LS1, states that
\begin{equation}\label{eq:ls1}
	\forall s\in\N^*\exists\a(s \sqsubset \a),
\end{equation}
where $\a$ ranges over the lawless sequences, and $s \sqsubset \a$ means that $s$ is an initial segment of $\a$. In words, the axiom says that for each possible finite sequence of natural numbers, we can construct a lawless sequence which has that finite sequence as an initial segment. (This makes the sequence slightly less ‘lawless’; for an extended discussion on this axiom, see~\cite[§12.2.2]{TrvD}. Troelstra calls lawless sequences for which no initial segment is specified beforehand \emph{proto-lawless}.)

To justify the second axiom, Troelstra considers the relation $\equiv$, which one could say represents \emph{intensional equality} or \emph{identity} of sequences: if $\a,\b\in\NN$, then $\a \equiv \b$ means that $\a$ and $\b$ refer to the \emph{same} process of generating values. Obviously, we have
\[ \a\equiv\b \lor \a\not\equiv\b, \text{\quad as well as \quad} \a\equiv\b\to\a=\b; \]
For lawless $\a$ and $\b$, Troelstra argues that we also have
\[ \a\not\equiv\b \to \a\neq\b, \]
since if $\a\not\equiv\b$, it would be impossible to have a proof of $\a=\b$, as at each point in time only finitely many values of both $\a$ and $\b$ can be known.\footnote{This argument is surrounded by controversy, also within intuitionistic mathematics.\footnote{Wim Veldman (personal communication).} Having a proof of $\a=\b$ is contradictory, but does this mean that the statement $\a=\b$ is itself contradictory?} This suggests the second axiom:
\begin{equation}\label{eq:ls2}
	\a=\b \lor \a\neq\b,
\end{equation}
where $\a$ and $\b$ are lawless sequences (obviously this does not hold for all sequences; this would imply LPO).

The third axiom, the \emph{principle of open data}, is like Brouwer’s continuity principle but with the $\forall \a$-quantifier removed, since Kreisel and Troelstra, unlike Brouwer, speak of individual lawless sequences. Let $A$ be a unary predicate and $\a$ a lawless sequence, then the simplest version of LS3 states that
\[ A(\a) \to \exists n\forall \b(\o\a n = \o\b n \to A(\b)). \]

From LS1 and LS3, one can in particular conclude that for any lawless $\a$ and lawlike $a$, we have
\begin{equation}\label{eq:lawless-neq-lawlike}
    \neg(\a = a).
\end{equation}
(Namely, assuming $\a = a$, use LS3 to find $n$ such that $\forall \b(\o\a n = \o\b n \to \b = a))$, and using LS1, find $\gamma$ such that $\o\gamma n = \o\a n$ but $\gamma(n) \neq \a(n)$. Then $\gamma\neq a$ but this contradicts LS3.)

Being able to prove this contradiction seems to conflict with the fact that at each moment in time only a finite initial segment of $\a$ is known, but it follows from the intensional nature of Kreisel and Troelstra’s definition of lawless sequences.

\

A disconcerting consequence of the intensional nature of Kreisel and Troelstra’s lawlessness is that a sequence $\g$ defined by pointwise addition of two sequences $\a$ and $\b$ cannot be lawless, even if $\a$ and $\b$ are both lawless; for $\g$ is \emph{completely described} by the law $\forall n(\g(n)=\a(n)+\b(n))$. On the other hand, one can take $\g$ and $\a$ to be lawless, but then $\b$ cannot be lawless. In fact, to make this rigorous, one can prove from the axioms LS1 and LS3 that the identity operation is the only lawlike operation under which the universe of lawless sequences is closed~\cite[§12.2.9]{TrvD}.

Kreisel and Troelstra use their four axioms LS1–LS4 to prove an \emph{elimination theorem}, which states that all formulas with quantifiers ranging over lawless sequences are equivalent to formulas not involving lawless sequences.\footnote{Granted that one accepts the Kreisel and Troelstra's formulation of lawless sequences; others argue that one cannot quantify over the lawless sequences in the first place because that set is not well-circumscribed \cite{swart1992spreads}.} In this way, they conclude, lawless sequences can be seen as nothing more than a ‘figure of speech’ but are not necessary to prove results.\footnote{However, there are several reasons to believe that the elimination theorem does \emph{not} imply that lawless sequences are `unnecessary’ in intuitionistic mathematics. See e.g.~\textcite[p41]{vanAtten2006brouwer}.}

It is useful to mention that one can also consider `lawless’ elements of a spread $\mathcal F_\s$; in this case, the choice is restricted to elements that are admitted by the spread law $\s$, but is otherwise free. In this way we can speak of, e.g., lawless binary sequences and lawless reals.

\subsection{Separating lawlike from lawless}\label{sec:separating-lawlike-lawless}
Although Kreisel and Troelstra's approach is probably the most well-known formalisation of lawless sequences, it is surrounded by controversy, also among intuitionists. Its dependence on intensional considerations causes it to be `language-dependent and [have] no absolute mathematical meaning' \cite{swart1992spreads}. This is particularly clear from the example involving the sequences $\g,\a$ and $\b$ just discussed. In the words of \textcite{gielen1981continuum}:
\begin{quote}
    We admit that an individual sequence can be (more or less) lawlike, or determinate, as we will put it later on, but we do not think it meaningful to speak of `the set of all lawlike sequences', or, for that matter, of `the set of all lawless sequences'. (The `denumerably unfinished totality of all lawlike sequences', to which Brouwer sometimes referred, is \textelp{} unable to do justice to the fullness of the geometric intuition of the continuum. In Brouwer's own view the continuum as a whole is a far clearer concept.)
\end{quote}
The sets of lawlike and lawless sequences are also sometimes called `not \emph{well-circumscribed}', and it is argued that one can only quantify over well-circumscribed sets \cite{swart1992spreads}, rendering much of Kreisel and Troelstra's formalisation invalid. Their philosophy of lawless sequences remains, however, an interesting phenomenological point of discussion.

\subsection{Extensional lawless sequences}\label{sec:lawless-extensional}
Also extensional definitions of lawless sequences have been proposed, viz.~sequences that are lawless by virtue of their elements only and not the process by which the elements are generated. One of these definitions has been studied by Joan Moschovakis, who describes lawless sequences as those that ``evade description by any fixed law’’, aiming to solve the problem that Kreisel and Troelstra's lawless sequences are not well-circumscribed \cite{moschovakis1993, moschovakis2016}. It has proven difficult to analyse these sequences intuitionistically. This stems from the fact that extensional lawlessness of a sequence cannot be decided on the basis of an initial segment only. Lawlessness can also be expressed through classical notions such as \emph{1-randomness} (also referred to as \emph{Martin-Löf randomness}), although this yields different results \cite{moschovakis2016}, and notions related to \emph{generic sets} \cite{moschovakis1993}. In this thesis, we mainly consider Kreisel and Troelstra’s notion, and compare it to indeterminism in physics (section~\ref{sec:lawless-indeterminism}); however, analysing the connection between physical indeterminism and these extensional lawlessness definitions would also be interesting.
    \chapter{Computability theory}\label{app:computability-theory}

\lettrine[lines=3]{T}{he field of} computability theory, also called recursion theory, emerged in the twentieth century to try to answer the question of what exactly an \emph{algorithm} is. It was developed in response to the foundational mathematical debates of the early twentieth century, and is therefore closely connected to the \textit{Entscheidungsproblem}, Gödel's incompleteness theorems and the debate between constructive and classical mathematics. For a more complete introduction to computability theory, see e.g.\ \textcite{terwijn2004syllabus}. Some notations in this chapter are borrowed from \textcite{weihrauch2000}, and others from \textcite{brattka2008tutorial}.

Many different definitions of the notion of algorithm were developed in the 1930's and onwards by notable figures including but not limited to Kurt Gödel, Emile Post, Alonzo Church and Stephen Cole Kleene. Perhaps the best-known formalisation was given by Alan Turing. Remarkably, all these notions of algorithm are formally equivalent, and give rise to the same formal notion of computable function.

Turing's approach is based on Turing machines, which represent idealised computing devices. A Turing machine can be described to consist of a \emph{tape} of unbounded length, divided into infinitely many \emph{cells}, each being either `blank' or containing a 0 or 1; together with a \emph{head} which can move along the tape and manipulate the cells' values. The action of the head depends only on a \emph{program}, i.e.\ a finite set of rules or instructions, given in advance, together with the \emph{state} of the program, of which there are finitely many, and the value of the cell at the current location of the head. Given an initial specification of the values of the cells and an initial location of the pointer, the Turing machine may either continue running forever, or terminate after a finite number of steps. In the latter case, the Turing machine is said to \emph{halt}.

In anticipation of section~\ref{sec:computable-analysis-preliminaries}, we generalise Turing machines to working with arbitrary alphabets, instead of only the alphabet $\{0,1\}$ as in the above.
\begin{convention}
    An \defn{alphabet} is a finite set containing at least two elements. In the following, $\Sigma$ denotes some fixed alphabet. $\Sst$ is the set of all finite sequences of elements in $\Sigma$.
\end{convention}

\begin{notation}
    For any two sets $X,Y$, by $\varphi:\subseteq X\to Y$ we mean a function whose domain $\dom\varphi$ is contained in $X$ and whose range $\ran\varphi$ is contained in $Y$. $\phi$ is called (strictly) \defn{partial} if $\dom\varphi\subsetneq X$ and \defn{total} if $\dom\varphi = X$.
\end{notation}

\begin{definition}\label{def:sst-sst-computable}
    A function $\varphi:\subseteq\Sst\to\Sst$ is called \defn{(Turing) computable} if there exists a Turing machine $M$ such that for any finite sequence $s\in\Sst$, when the initial configuration of the tape contains $s$ and is otherwise blank and the initial location of the head is at the start of $s$, then the Turing machine halts if and only if $s\in\dom\varphi$, and if it does, it finishes in a final configuration in which the tape contains the finite sequence $\varphi(s)$ and is otherwise blank.
\end{definition}

The original objects of study for Turing were not functions $\varphi:\subseteq\Sst\to\Sst$ but functions $\varphi:\subseteq\N\to\N$. These can, however be easily encoded by the former functions. In section~\ref{sec:computable-analysis-preliminaries} we will discuss in more detail how to define computability on arbitrary sets with cardinality $<2^\omega$.

\paragraph{Church-Turing thesis}\label{sec:church-turing-computability}
As we noted before, in addition to Turing's formulation, many different formalisations of the informal notion of algorithm have been developed, and (almost) all have turned out to be equivalent. This suggests that Turing's formal notion is indeed the `right' notion of algorithmic computability. This is captured in the heuristic statement referred to as the \defn{Church-Turing thesis}:
\begin{quote}
    A function $\varphi:\N\to\N$ can be computed by an algorithm (in the informal sense) if and only if it is Turing computable.
\end{quote}
An algorithm in the informal sense is also sometimes referred to as an \defn{effective} or \emph{mechanical}, i.e.\ `pen-and-paper' method. A variant of the Church-Turing thesis relevant to this thesis is the \defn{physical Church-Turing thesis}, which states that
\begin{quote}
    A function $\varphi$ can be physically computed if and only if it is Turing computable.
\end{quote}
That is, not only pen-and-paper methods can be used but any process allowed according to the laws of physics. We discuss the status of the physical Church-Turing thesis in somewhat more detail in section~\ref{sec:constructivising-physics}.


\section{Computable analysis} \label{sec:computable-analysis-preliminaries}
While the origins of computability theory lie in the study of algorithms that compute functions $\varphi:\N\to\N$ of natural numbers, the second half of the twentieth century showed an increased interest in computability notions for functions on non-discrete spaces, in particular the real numbers. Many different approaches have been developed to generalise computability theory to the reals. These form the field of computable analysis, which is still very active today. We will focus on Klaus Weihrauch's \emph{Type-2 theory of effectivity} and in particular on the language of \emph{representations} of sets.\footnote{Another notable approach was given by \textcite{PER89computability}, who also discuss computability notions on Banach spaces and operators on Banach spaces, which are therefore relevant to physics. Computability notions on manifolds, however, which could be useful to classical and relativistic mechanics, have only recently been introduced \cite{aguilar2017computable}.} A comprehensive introduction is presented in \textcite{weihrauch2000}, and a brief introduction is given in \textcite{brattka2008tutorial}. An even briefer and less complete introduction follows below.

\

We adopt the same convention as in the above that $\Sigma$ is a finite alphabet containing at least two elements. $\Som = \{ p \mid p : \N\to\Sigma \}$ is now the set of (one-way) infinite sequences over $\Sigma$.

Whereas Turing machines usually compute functions $\varphi:\Sst\to\Sst$, a \defn{Type-2 machine} computes functions $\Som\to\Sst$, $\Sst\to\Som$ or $\Som\to\Som$. A Type-2 machine of the latter type consists of a one-way input tape, a one-way output tape and a working tape. Each tape has a separate head which can move along the cells of its tape. That the input and output tapes are \emph{one-way} means that the heads can only move to the right (one cell at a time). The input head only reads values of cells on the input tape, while the output head only writes values (from $\Sigma$) on its cells, and the head at the working tape can do both (and can move to the left and right). This restriction on the head of the output tape allows us to at any time during the execution of the machine retrieve an initial segment of the (potentially infinite) sequence $p\in\Som$ that the machine outputs (without allowing the machine to change the values it had already output at an earlier stage).
Again, the behaviour of the three heads is completely determined by the \emph{program} of the machine, specified in advance, its \emph{state}, of which there are finitely many, and the value of the cells at each of the heads' positions.

\begin{definition}\label{def:som-som-computable}
    A function $F:\subseteq\Som\to\Som$ is \defn{computable} if there exists a Type-2 machine which, when the initial values of the cells on the input tape are given by the infinite sequence $p\in\Som$, keeps writing on the output tape forever and outputs the sequence $F(p)$ if $p\in\dom F$, and terminates after a finite number of steps if $p\notin\dom F$.\footnote{Remark that in the case for functions $\Som\to\Som$, the computation is `successful' whenever the corresponding machine keeps running forever, while for functions $\Sst\to\Sst$ (see Definition~\ref{def:sst-sst-computable}), the computation is `successful' whenever the Turing machine terminates after a finite number of steps.}
\end{definition}

Please note that this is only one of the many possible ways of formulating machines that compute $\Som\to\Som$ functions, but that (to us) only the resulting formal computability notion is of importance.

\subsection{Representations and computable functions on the reals}
To define computability on more general sets, we use the framework of representations.
\begin{definition}
    A \defn{representation} of a set $X$ is a surjective function $\delta:\subseteq\Som\to X$. If $x\in X$, then any $p\in\Som$ with $\d(p) = x$ is called a \defn{$\d$-name} of $x$. We call the pair $(X,\d)$ a \defn{represented space}.
\end{definition}

\begin{definition}\label{def:cauchy-representation}
    The \defn{Cauchy representation} $\rho:\subseteq\Som\to\R$ of the real numbers is defined by $\rho(p) = x\in\R$ for $p\in\Som$ if and only if $p$ encodes a sequence of rationals $(q_i)_{i\in\N}\subseteq\Q$ that converges rapidly to $x$, i.e.\ $\forall i\in\N: |x-q_i| < 2^{-i}$. (Here we assume that there is a canonical way to encode a sequence of rationals in a sequence $p\in\Som$.)
\end{definition}

Naturally, by construction of $\R$, the function $\rho$ is surjective. We are, however, interested in the subset of $\R$ of \emph{computable} numbers.

\begin{definition}\label{def:computable-sequence}
    A sequence $p\in\Som$ is called \defn{computable} if there is a Type-2 machine that on any (or no) input produces the sequence $p$. (This is equivalent to there being a total computable function $\varphi:\N\to\Sigma$ that produces the symbol $p(n)$ on input $n\in\N$.)
    
    Let $(X,\d)$ be a represented space. An element $x\in X$ is called \defn{computable} if it has a computable $\d$-name (i.e.\ there exists a computable $p\in\Som$ such that $\d(p) = x$).
\end{definition}

\begin{example}\label{ex:computable-real}
    It follows from Definitions~\ref{def:cauchy-representation} and~\ref{def:computable-sequence} that a real $x\in\R$ is $\rho$-computable if and only if there exists a computable sequence of rationals (i.e.\ a computable total function $\N\to\Q$, in the sense of Definition~\ref{def:sst-sst-computable}) that converges rapidly to $x$. This can be shown to be equivalent to, for instance (see \textcite[Theorem~3.2]{brattka2008tutorial} for more examples):
    \begin{itemize}
        \item There exists a Turing machine that outputs the infinite binary or decimal expansion of $x$;
        \item There exists a computable sequence of shrinking and dwindling (in the sense of Definition~\ref{def:intuitionistic-reals}) rational intervals, of which the intersection contains only $x$.
    \end{itemize}
    
    Because there are countably many Type-2 machines, there are only countably many computable reals;\footnote{Note that this chapter is written from a classical viewpoint. When one restricts oneself to using computable mathematical objects only, as is done in e.g.\ Markov's recursive constructive mathematics, then one cannot say that the set of computable sequences is countable, as there exists no computable enumeration of all computable numbers. Cf.\ section~\ref{sec:lawlike-sequences}.} hence, almost every real (w.r.t.~the Lebesgue measure) is uncomputable.
\end{example}

We now define computability of functions on represented sets.
\begin{definition}
    Let $(X,\d_X)$ and $(Y,\d_Y)$ be represented spaces and let $f:\subseteq X\to Y$ be a function. A function $F:\subseteq\Som\to\Som$ is called a \defn{$(\d_X,\d_Y)$-realiser} of $f$ if and only if $\dom(f\circ\d_X) \subseteq \dom(\d_Y\circ F)$ and
    \[ (\d_Y\circ F)(p) = (f\circ \d_X)(p) \text{\quad for all } p\in\dom(f\circ\d_X). \]
    (I.e.\ for any $\d_X$-name $p$ for $x\in X$, the value $F(p)$ is a $\d_Y$-name for $f(x)$.) The function $f$ is called \defn{computable} if it admits a computable $(\d_X,\d_Y)$-realiser (in the sense of Definition~\ref{def:som-som-computable}).
\end{definition}

We first define some more useful representations of the real numbers, and then discuss what computability of functions with respect to these representations intuitively mean, without going into the details.
\begin{definition}\phantomsection\label{def:some-representations}
    \begin{itemize}
        \item For any representation $\d$ of $X$, we can define a representation $\d^2$ of $X\times X$ by
        \[ \d^2(\langle p_1,p_2\rangle) = (x_1,x_2)\ \ :\iff\ \d(p_i) = x_i \text{ for } i = 1,2, \]
        where for $p_1,p_2\in\Som$, we define $\langle p_1,q_1\rangle = (p_1(0),p_2(0),p_1(1),p_2(1),\ldots) \in \Som.$
        
        Iterating this, we can define representations $\d^k$ of $X^k$ for any power $k\in\N$. In particular, the \defn{Cauchy representation} $\rho^k$ of $\R^k$ is defined in this way.
        
        \item The representation $\rho_\text{int}:\subseteq\Som\to\R$ is defined as
        \begin{align*}
            \rho_\text{int}(p) = x\ \ :\iff\ & p \text{ encodes a sequence } ([a_i, b_i])_{i\in\N}\subseteq\Q^2 \text{ of rational intervals}\\
            & \text{such that for all } i, a_i \leq a_{i+1} \leq b_{i+1} \leq b_i \text{ and } |b_i-a_i| < 2^{-i}, \\
            & \text{and } a_i \to x, b_i \to x.
        \end{align*}
        
        \item The \defn{binary representation} $\rho_2:\subseteq\Som\to [0,1]$ is defined in the usual sense by 
        \[ \rho_2(p) = x\in\R\ \ :\iff\ p_i\in\{0,1\} \text{ for all } i \text{ and } x = \sum_{i=1}^\infty 2^{-i\, p(i)} \]
        for $p\in\Som$ (where $\Sigma$ is assumed to contain at least the symbols $0$ and $1$.)
        
        \item The representation $\rho_< :\subseteq \Som\to\R$ is defined by
        \[ \rho_<(p) = x\in\R\ \ :\iff\ p \text{ enumerates all } q\in\Q \text{ with } q<x. \]
        A $\rho_<$-computable real number is also called \defn{lower-semicomputable}, as is a function $f:(X,\d)\to\R$ which is $(\d,\rho_<)$-computable. $\rho_<$-computability is weaker than $\rho$-computability. \defn{Upper-semicomputability} is similarly defined, by replacing $<$ with $>$.
    \end{itemize}
\end{definition}

Intuitively, a function $f:\R^k\to\R^k$ is $(\rho^k,\rho^k)$-computable iff when given a sequence of rational vectors in $\Q^k$ rapidly converging to $x\in\R^k$, one can compute from this another sequence of rational vectors rapidly converging to $f(x)\in\R^k$.

Similarly, a function $f:\R\to\R$ is $(\rho_\text{int},\rho_\text{int})$-computable iff there is an effective method to determine, for each $i$, the value $f(x)$ up to a precision of $2^{-i}$, given any required level of precision on the argument $x$. This is very much like the way real functions are defined in intuitionism (see section~\ref{sec:int-real-functions}). Note, however, that a \emph{constructive} method need not be \emph{effective} (i.e.\ computable); indeed, there are intuitionistic functions that are not lawlike.

\label{computability-decimal-expansion-problem}It follows from the discussion in Example~\ref{ex:computable-real} that a real number $x$ is $\rho_2$-computable if and only if it is $\rho$- or $\rho_\text{int}$-computable. However, when it comes to computability of \emph{functions}, computability with respect to $\rho_2$ is a stronger notion than computability with respect to the other two representations. In particular, the addition function $f:\R^2\to\R, (x,y)\mapsto x+y$ is \emph{not} $(\rho_2^2, \rho_2)$-computable. The argument for this is completely analogous to the one given in section~\ref{sec:int-reals-other-constructions} on page~\pageref{sec:int-reals-other-constructions}, where it is shown by a Brouwerian argument that when the reals are defined by through binary expansion, the set of real numbers is constructively not closed under addition.

\

We mention the following result, which naturally follows from the discussion above, because it stresses the connection between intuitionistic mathematics and computability theory once again:
\begin{theorem}
    Every $(\rho^k,\rho^k)$-computable function is continuous.
\end{theorem}


\subsection{Recursively enumerable open subsets}
In this thesis, we are not only interested in computability notions for points in $\R^k$, but also for subsets of $\R^k$. A first attempt, inspired by the notion of decidability on subsets of $\N$, would be to call a subset $A\subseteq\R^k$ \emph{decidable} if there exists a Turing machine that on input $x\in\R$ decides whether $x\in A$ or $x\notin A$---that is, the characteristic function $\chi_A:\R^k\to\R$ which maps $A$ to $\{1\}$ and $\R^k\setminus A$ to $\{0\}$ is $(\rho^k,\rho)$-computable. We see, however, that such a definition would not make sense:
\begin{proposition}
    The only decidable subsets of $\R^k$ are the trivial subsets $\emptyset$ and $\R^k$.
\end{proposition}
\begin{proof}
    Every $(\rho^k,\rho)$-computable function is continuous, and the only subsets of $\R^k$ with continuous characteristic function are $\emptyset$ and $\R^k$. The characteristic functions $\chi_\emptyset$ and $\chi_{\R^k}$ are obviously computable.
\end{proof}

A more useful definition is the following:
\begin{definition}
    A subset $V\subseteq\Som$ is called \defn{recursively enumerable open (r.e.\ open)} if $V = \dom h$ for some computable function $h:\subseteq\Som\to\Sst$ (i.e.\ if there is a Turing machine which halts on input $p\in\Som$ if and only if $p\in V$). $V$ is called \defn{r.e.\ open in $W$} for $W\subseteq\Som$ if $V = \dom h \cap W$ for such $h$.
    
    If $(X,\d)$ is a represented space, a subset $U\subseteq X$ is called \defn{$\d$-r.e.\ open} if $\d^{-1}(U)$ is r.e.\ open as a subset of $\Som$. (That is, there is a Turing machine which, given a $\d$-name of any element $x\in X$, halts if and only if $x\in U$.)\footnote{\label{fn:r-e-open-in-w}It would be natural to define, in addition, that $U$ is \emph{$\d$-r.e.\ open in $W$} for some $W\subseteq X$ if $U = V\cap W$ for some $V\subseteq X$ which is $\d$-r.e.\ open. However, this definition seems not useful, at least not to us (see footnote~\ref{fn:r-e-open-in-dom-f}); this might be the reason that it is not given in e.g.\ \textcite{weihrauch2000}. Also note that this is not (at least not trivially) equivalent to $\d^{-1}(U)$ being r.e.\ open in $\d^{-1}(W)$ (namely, if $\d^{-1}(U) = \dom h \cap \d^{-1}(W)$ for some computable $h:\subseteq\Som\to\Sst$ then this implies that $U = \d(\dom h) \cap W$, but $\d^{-1}(\d(\dom h))$ might be strictly larger than $\dom h$ and hence not necessarily r.e.\ open.)}
\end{definition}

In the case of real numbers, we can characterise $\rho^k$-r.e.\ openness as follows. Here ${B_{\e}(q)\subseteq\R^k}$ is the Euclidean open ball with radius $\e$ around $q$.
\begin{proposition}
    A set $U\subseteq\R^k$ is $\rho^k$-r.e.\ open if and only if $U = \bigcup_{i\in\N} B_{\e_i}(q_i)$ for some computable sequence of pairs $(q_i,\e_i) \in \Q^k\times\Q_{>0}$. In particular, any r.e.\ open subset $U\subseteq\R^k$ is open in the Euclidean topology.
\end{proposition}

For a proof see e.g.\ \textcite[§5.1.16]{weihrauch2000}. Now follow some propositions which will prove useful in section~\ref{sec:mathematics-development} when applying them to the Hamiltonian flow.

\begin{proposition}\label{prop:intersection-r-e-open}
    The intersection $V\cap W$ of any two r.e.\ open sets $V,W\subseteq\Som$ is again r.e.\ open. In particular, if $V$ is r.e.\ open in $W$ and $W$ is r.e.\ open, then $V$ is r.e.\ open.
\end{proposition}
\begin{proof}
    (From \textcite[Theorem~2.4.5]{weihrauch2000}.) Let $M_V$ and $M_W$ be Turing machines that halt on input $p\in\Som$ iff $p\in V$, resp.\ $p\in W$. Let $M$ be a machine which on input $p\in\Som$ simulates both $M_V$ and $M_W$ on input $p$ in parallel (by alternately executing an instruction in the program of $M_V$ and of $M_W$, as long as none of the machines has halted), and which halts whenever both $M_V$ and $M_W$ have halted. Then $M$ halts on $p$ if and only if both $M_V$ and $M_W$ halt on $p$, which is the case if and only if $p\in V\cap W$.
\end{proof}

\begin{proposition}\label{prop:preimage-r-e-open}
    Let $X,Y$ be sets and $\d_X : \Som\to\Sst$, $\d_Y : \Som\to\Sst$ be representations. Let $f:\subseteq X\to Y$ be $(\d_X, \d_Y)$-computable with $(\d_X, \d_Y)$-realiser $F:\Som\to\Som$. If $U\subseteq Y$ is r.e.\ open, then $\d_X^{-1}(f^{-1}(U))$ is r.e.\ open in $\dom F$.\footnote{\label{fn:r-e-open-in-dom-f}It is suggested in \textcite{ziegler2006effectively} that, conversely, any function for which taking pre-images preserves r.e.\ openness is also computable; but it is unclear whether or not this holds only for total functions, and we do not need this result here. Furthermore, an even more elegant result than the present Proposition would be that $f^{-1}(U)$ is r.e.\ open in $\dom f$, in the sense defined in footnote~\ref{fn:r-e-open-in-w}; however, this might not necessarily follow from the present result, for the reason explained in that footnote. Luckily, we need no more than the present Proposition for our results in section~\ref{sec:mathematics-development}.}
\end{proposition}
\begin{proof}
    Because $U$ is r.e.\ open, there is a computable function $g:\subseteq\Som\to\Sst$ such that $\dom g = \d_Y^{-1}(U)$.
    Because $F$ is a $(\d_X, \d_Y)$-realiser of $f$, we have $\d_Y F (p) = f \d_X (p)$ for all sequences $p\in\dom(f\d_X)$.

    We use the fact that the function $g\circ F:\subseteq\Som\to\Sst$ (which has domain $F^{-1}(\dom g)$) has a computable extension $h:\subseteq\Som\to\Sst$ such that $\dom h \cap \dom F = \dom(g\circ F)$.
    For a proof of this fact, see \textcite[Theorem~2.1.12]{weihrauch2000}. Note that $\dom (f\circ\d_X) \subseteq \dom F$. It follows that for $p\in\dom(f\circ\d_X)$, we have $p\in \d_X^{-1}(f^{-1}(U)) \iff f(\d_X(p)) \in U \iff \d_Y(F(p)) \in U \iff F(p) \in \d_Y^{-1}(U) = \dom g \iff {p\in \dom(g\circ F)} \iff p\in \dom h \cap \dom F$. Since $\dom h\subseteq\Som$ is r.e.\ open by definition, we conclude that $\d_X^{-1}(f^{-1}(U))$ is r.e.\ open in $\dom F$.
\end{proof}
    
    \backmatter
    \phantomsection
    \addcontentsline{toc}{chapter}{\bibname}
    \printbibliography

@article{gisin2018,
  title={Indeterminism in Physics, Classical Chaos and Bohmian Mechanics: Are Real Numbers Really Real?},
  author={Gisin, Nicolas},
  journal={Erkenntnis},
  year={2019},
  publisher={Springer}
}

@article{dSG19,
  title={Physics without determinism: Alternative interpretations of classical physics},
  author={Del Santo, Flavio and Gisin, Nicolas},
  journal={Physical Review A},
  volume={100},
  number={6},
  year={2019},
  publisher={APS}
}

@book{hooft2016cellular,
  title={The cellular automaton interpretation of quantum mechanics},
  author={'t Hooft, Gerard},
  sortname={Hooft, Gerard},
  volume={185},
  year={2016},
  publisher={Springer}
}

@incollection{drossel2015,
  title={On the relation between the second law of thermodynamics and classical and quantum mechanics},
  author={Drossel, Barbara},
  booktitle={Why more is different},
  pages={41--54},
  year={2015},
  publisher={Springer}
}

@inproceedings{dowek2013real,
  title={Real numbers, chaos, and the principle of a bounded density of information},
  author={Dowek, Gilles},
  booktitle={International Computer Science Symposium in Russia},
  pages={347--353},
  year={2013},
  organization={Springer}
}

@article{lev2017,
  title={Why finite mathematics is the most fundamental and ultimate quantum theory will be based on finite mathematics},
  author={Lev, Felix M},
  journal={Physics of Particles and Nuclei Letters},
  volume={14},
  number={1},
  pages={77--82},
  year={2017},
  publisher={Springer}
}

@article{visser2012,
  title={Which number system is ``best'' for describing empirical reality?},
  author={Visser, Matt},
  journal={arXiv preprint arXiv:1212.6274},
  year={2012}
}

@inproceedings{caticha2019information,
  title={The information geometry of space-time},
  author={Caticha, Ariel},
  booktitle={Multidisciplinary Digital Publishing Institute Proceedings},
  volume={33},
  number={1},
  year={2019}
}

@InCollection{sep-time,
	author       =	{Markosian, Ned},
	title        =	{Time},
	booktitle    =	{The Stanford Encyclopedia of Philosophy},
	editor       =	{Edward N. Zalta},
	howpublished =	{\url{https://plato.stanford.edu/archives/fall2016/entries/time/}},
	year         =	{2016},
	edition      =	{Fall 2016},
	publisher    =	{Metaphysics Research Lab, Stanford University}
}

@article{friedman1999fuzzy,
  title={Fuzzy dynamics as an alternative to statistical mechanics},
  author={Friedman, Y and Sandler, U},
  journal={Fuzzy Sets and Systems},
  volume={106},
  number={1},
  pages={61--74},
  year={1999},
  publisher={Elsevier}
}

@incollection{dasgupta2011,
  title={Mathematical foundations of randomness},
  author={Dasgupta, Abhijit},
  booktitle={Philosophy of Statistics},
  pages={641--710},
  year={2011},
  publisher={Elsevier}
}

@article{dS20,
  title={Indeterminism, causality and information: Has physics ever been deterministic?},
  author={Del Santo, Flavio},
  journal={arXiv preprint arXiv:2003.07411},
  year={2020}
}

@article{chaitin2004,
  title={How real are real numbers?},
  author={Chaitin, Gregory},
  year={2004},
  journal={arXiv:math/0411418}
}

@article{caticha2008,
  title={Lectures on probability, entropy, and statistical physics},
  author={Caticha, Ariel},
  journal={arXiv preprint arXiv:0808.0012},
  year={2008}
}

@article{ellis2018,
  title={The physics of infinity},
  author={Ellis, George FR and Meissner, Krzysztof A and Nicolai, Hermann},
  journal={Nature Physics},
  volume={14},
  number={8},
  pages={770--772},
  year={2018},
  publisher={Nature Publishing Group}
}

@article{gisin2019real,
  title={Real numbers are the hidden variables of classical mechanics},
  author={Gisin, Nicolas},
  journal={Quantum Studies: Mathematics and Foundations},
  pages={1--5},
  year={2019},
  publisher={Springer}
}

@article{lloyd2002,
  title={Computational capacity of the universe},
  author={Lloyd, Seth},
  journal={Physical Review Letters},
  volume={88},
  number={23},
  pages={237901},
  year={2002},
  publisher={APS}
}

@article{bishop2003predictability,
  title={On separating predictability and determinism},
  author={Bishop, Robert C},
  journal={Erkenntnis},
  volume={58},
  number={2},
  pages={169--188},
  year={2003},
  publisher={Springer}
}

@incollection{gisin2017time,
  title={Time really passes, science can’t deny that},
  author={Gisin, Nicolas},
  booktitle={Time in Physics},
  pages={1--15},
  year={2017},
  publisher={Springer}
}

@article{bekenstein1981,
  title={Universal upper bound on the entropy-to-energy ratio for bounded systems},
  author={Bekenstein, Jacob D},
  journal={Physical Review D},
  volume={23},
  number={2},
  pages={287},
  year={1981},
  publisher={APS}
}

@article{page2018,
  title={The Bekenstein Bound},
  author={Page, Don N},
  journal={arXiv preprint arXiv:1804.10623},
  year={2018},
  publisher={World Scientific}
}

@inbook{FrW11,
  title={Entropy -- A Guide for the Perplexed},
  author={Werndl, C and Frigg, R},
  booktitle={Probabilities in Physics},
  editor={Beisbart, C and Hartmann, S},
  pages={115--142},
  publisher={Oxford University Press},
  year={2011},
  note={Equation (21)}
}

@article{Cat08,
  title={Lectures on probability, entropy, and statistical physics},
  author={Caticha, Ariel},
  journal={arXiv preprint arXiv:0808.0012},
  year={2008}
}

@InCollection{sep-information-entropy,
	author       =	{Maroney, Owen},
	title        =	{Information Processing and Thermodynamic Entropy},
	booktitle    =	{The Stanford Encyclopedia of Philosophy},
	editor       =	{Edward N. Zalta},
	howpublished =	{\url{https://plato.stanford.edu/archives/fall2009/entries/information-entropy/}},
	year         =	{2009},
	edition      =	{Fall 2009},
	publisher    =	{Metaphysics Research Lab, Stanford University}
}

@article{GrW08,
  title={Shannon information and Kolmogorov complexity},
  author={Grunwald, Peter and Vit{\'a}nyi, Paul},
  journal={arXiv preprint cs/0410002},
  year={2008}
}

@incollection{trzesicki1994,
  title={Intuitionism and Indeterminism (Tense-Logical Considerations)},
  author={Trzesicki, Kazimierz},
  booktitle={Philosophical Logic in Poland},
  pages={271--296},
  year={1994},
  publisher={Springer}
}

@incollection{bauer2013intuitionistic,
  title={Intuitionistic mathematics and realizability in the physical world},
  author={Bauer, Andrej},
  booktitle={A Computable Universe: Understanding and Exploring Nature as Computation},
  bookauthor={Zenil, H.},
  pages={143--157},
  year={2013},
  publisher={World Scientific}
}

@incollection{cattaneo1995,
  title={Constructivism and Operationalism in the Foundations of Quantum Mechanics},
  author={Cattaneo, G and Dalla Chiara, ML and Giuntini, R},
  booktitle={The Foundational Debate},
  pages={21--31},
  year={1995},
  publisher={Springer}
}

@incollection{svozil1995,
  title={A constructivist manifesto for the physical sciences — Constructive re-interpretation of physical undecidability},
  author={Svozil, Karl},
  booktitle={The Foundational Debate},
  pages={65--88},
  year={1995},
  publisher={Springer}
}

@book{weyl1949,
  title={Philosophy of mathematics and natural science},
  author={Weyl, Hermann},
  translator={Helmer, Olaf},
  year={1949},
  publisher={Princeton University Press},
  location={Princeton}
}

@article{gisin2020comment,
  title={Mathematical languages shape our understanding of time in physics},
  author={Gisin, Nicolas},
  journal={Nature Physics},
  pages={1--3},
  year={2020},
  publisher={Nature Publishing Group}
}

@online{bauer2008blog,
  title={Intuitionistic mathematics for physics},
  author={Bauer, Andrej},
  year={2008},
  url={math.andrej.com/2008/08/13/intuitionistic-mathematics-for-physics/},
  urldate={2020-06-02}
}

@book{tadaki2019,
  title={A statistical mechanical interpretation of algorithmic information theory},
  author={Tadaki, Kohtaro},
  year={2019},
  publisher={Springer Nature}
}

@article{billinge1997,
  title={A constructive formulation of Gleason's theorem},
  author={Billinge, Helen},
  journal={Journal of philosophical logic},
  volume={26},
  number={6},
  pages={661--670},
  year={1997},
  publisher={Springer}
}

@article{bridges1999,
  title={Can constructive mathematics be applied in physics?},
  author={Bridges, Douglas S},
  journal={Journal of Philosophical Logic},
  pages={439--453},
  year={1999},
  publisher={JSTOR}
}

@article{hellman1993,
  title={Gleason's theorem is not constructively provable},
  author={Hellman, Geoffrey},
  journal={Journal of Philosophical Logic},
  pages={193--203},
  year={1993},
  publisher={JSTOR}
}

@article{hellman1998,
  title={Mathematical constructivism in spacetime},
  author={Hellman, Geoffrey},
  journal={The British journal for the philosophy of science},
  volume={49},
  number={3},
  pages={425--450},
  year={1998},
  publisher={Oxford University Press}
}

@article{richman1999,
  title={A constructive proof of Gleason's theorem},
  author={Richman, Fred and Bridges, Douglas},
  journal={Journal of Functional Analysis},
  volume={162},
  number={2},
  pages={287--312},
  year={1999},
  publisher={Elsevier}
}

@article{gisin1991,
  title={Propensities in a non-deterministic physics},
  author={Gisin, Nicolas},
  journal={Synthese},
  volume={89},
  number={2},
  pages={287--297},
  year={1991},
  publisher={Springer}
}

@book{bishop1967,
  title={Foundations of constructive analysis},
  author={Bishop, Errett},
  volume={60},
  year={1967},
  publisher={McGraw-Hill New York}
}

@article{andreka2009hypercomputing,
  title={General relativistic hypercomputing and foundation of mathematics},
  author={Andr{\'e}ka, Hajnal and N{\'e}meti, Istv{\'a}n and N{\'e}meti, P{\'e}ter},
  journal={Natural Computing},
  volume={8},
  number={3},
  pages={499--516},
  year={2009},
  publisher={Springer}
}

@article{arrighi2012physical,
  title={The physical Church-Turing thesis and the principles of quantum theory},
  author={Arrighi, Pablo and Dowek, Gilles},
  journal={International Journal of Foundations of Computer Science},
  volume={23},
  number={05},
  pages={1131--1145},
  year={2012},
  publisher={World Scientific}
}

@book{veldmannotes,
  title={Notes on Intuitionistic Mathematics, 1--8: Course at Radboud University, Nijmegen},
  author={Veldman, W.H.M.},
  year={2019}
}

@article{vanAtten2018bks,
  title={The Creating Subject, the Brouwer--Kripke Schema, and infinite proofs},
  author={van Atten, Mark},
  journal={Indagationes Mathematicae},
  volume={29},
  number={6},
  pages={1565--1636},
  year={2018},
  publisher={Elsevier}
}

@book{troelstra77,
  title={Choice sequences: A chapter of intuitionistic mathematics},
  author={Troelstra, Anne Sjerp},
  year={1977},
  publisher={Oxford University Press},
  location={Oxford}
}

@book{vanAtten2006brouwer,
  title={Brouwer meets Husserl: on the phenomenology of choice sequences},
  author={van Atten, Mark},
  volume={335},
  year={2006},
  publisher={Springer Science \& Business Media}
}

@article{popper1950indeterminism,
  title={Indeterminism in quantum physics and in classical physics. Part I},
  author={Popper, Karl R},
  journal={The British Journal for the Philosophy of Science},
  volume={1},
  number={2},
  pages={117--133},
  year={1950},
  publisher={Oxford University Press}
}

@book{earman1986,
  title={A primer on determinism},
  author={Earman, John},
  volume={37},
  year={1986},
  publisher={Springer Science \& Business Media}
}

@inproceedings{moschovakis1993,
  title={An intuitionistic theory of lawlike, choice and lawless sequences},
  author={Moschovakis, Joan R},
  booktitle={Logic colloquium},
  volume={90},
  pages={191--209},
  year={1993}
}

@article{moschovakis2016,
  title={Iterated definability, lawless sequences and Brouwer’s continuum},
  author={Moschovakis, Joan Rand},
  year={2016},
  publisher={Oxford University Press, Oxford}
}

@article{kreisel1968,
  title={Lawless sequences of natural numbers},
  author={Kreisel, Georg},
  journal={Compositio Mathematica},
  volume={20},
  pages={222--248},
  year={1968}
}

@mvbook{TrvD,
  title={Constructivism in Mathematics},
  author={Troelstra, Anne Sjerp and van Dalen, Dirk},
  year={1988},
  volumes={2},
  publisher={Elsevier Science}
}

@article{mandelkern1989,
  title={Brouwerian counterexamples},
  author={Mandelkern, Mark},
  journal={Mathematics Magazine},
  volume={62},
  number={1},
  pages={3--27},
  year={1989},
  publisher={Taylor \& Francis}
}

@book{weihrauch2000,
  title={Computable analysis: An introduction},
  author={Weihrauch, Klaus},
  year={2000},
  publisher={Springer-Verlag Berlin Heidelberg}
}

@inproceedings{graca2018,
  title={Computability of ordinary differential equations},
  author={Gra{\c{c}}a, Daniel S and Zhong, Ning},
  booktitle={Conference on Computability in Europe},
  pages={204--213},
  year={2018},
  organization={Springer}
}

@article{aguilar2017computable,
  title={Computable structures on topological manifolds},
  author={Aguilar, Marcelo A and Conde, Rodolfo},
  journal={arXiv preprint arXiv:1703.04075},
  year={2017}
}

@book{terwijn2004syllabus,
  title={Syllabus Computability Theory, Technical University of Vienna},
  author={Terwijn, S A},
  year={2004},
  note={Available at web pages author}
}

@incollection{brattka2008tutorial,
  title={A tutorial on computable analysis},
  author={Brattka, Vasco and Hertling, Peter and Weihrauch, Klaus},
  booktitle={New computational paradigms},
  pages={425--491},
  year={2008},
  publisher={Springer}
}

@book{PER89computability,
  title={Computability in Analysis and Physics},
  author={Pour-El, M B and Richards, J I},
  series={Perspectives in Mathematical Logic},
  publisher={Springer},
  location={Berlin},
  year={1989}
}

@book{landsman2017,
  title={Foundations of quantum theory: from classical concepts to operator algebras},
  author={Landsman, N P},
  year={2017},
  publisher={Springer Nature}
}

@book{lee2018manifolds,
  title={Introduction to Smooth Manifolds},
  author={Lee, John M},
  year={2018},
  publisher={Springer}
}

@incollection{terwijn2016randomness,
  title={The mathematical foundations of randomness},
  author={Terwijn, S A},
  booktitle={The Challenge of Chance},
  pages={49--66},
  year={2016},
  publisher={Springer}
}

@article{quanta2020,
  title={Does Time Really Flow? New Clues Come From a Century-Old Approach to Math},
  author={Wolchover, Natalie},
  date={2020-04-07},
  url={www.quantamagazine.org/does-time-really-flow-new-clues-come-from-a-century-old-approach-to-math-20200407/},
  journal={Quanta Magazine}
}

@article{van2002brouwer,
  title={Brouwer and Weyl: The phenomenology and mathematics of the intuitive continuum},
  author={van Atten, Mark and van Dalen, Dirk and Tieszen, Richard},
  journal={Philosophia Mathematica},
  volume={10},
  number={2},
  pages={203--226},
  year={2002},
  publisher={OUP}
}

@book{weyl1918kontinuum,
  title={Das Kontinuum},
  author={Weyl, Hermann},
  year={1918},
  publisher={Von Veit},
  location={Leipzig}
}

@book{hellman2018varieties,
  title={Varieties of continua: From regions to points and back},
  author={Hellman, Geoffrey and Shapiro, Stewart},
  year={2018},
  publisher={Oxford University Press}
}

@article{johnstone1983point,
  title={The point of pointless topology},
  author={Johnstone, Peter T and others},
  journal={Bulletin (New Series) of the American Mathematical Society},
  volume={8},
  number={1},
  pages={41--53},
  year={1983},
  publisher={American Mathematical Society}
}

@article{roubach2005being,
  title={Being and Time and Brouwer's Intuitionism},
  author={Roubach, Michael},
  journal={Angelaki Journal of the Theoretical Humanities},
  volume={10},
  number={1},
  pages={181--186},
  year={2005},
  publisher={Taylor \& Francis}
}

@article{chaitin1978note,
  title={A note on Monte Carlo primality tests and algorithmic information theory},
  author={Chaitin, Gregory J and Schwartz, Jacob T},
  journal={Communications on Pure and Applied Mathematics},
  volume={31},
  number={4},
  pages={521--527},
  year={1978},
  publisher={Citeseer}
}

@article{ydri2001fuzzy,
  title={Fuzzy physics},
  author={Ydri, Badis},
  journal={arXiv preprint hep-th/0110006},
  year={2001}
}

@article{hamkins2000infinite,
  title={Infinite time Turing machines},
  author={Hamkins, Joel David and Lewis, Andy},
  journal={Journal of Symbolic Logic},
  pages={567--604},
  year={2000},
  publisher={JSTOR}
}

@article{vandalen1990war,
  title={The war of the frogs and the mice, or the crisis of the Mathematische Annalen},
  author={van Dalen, Dirk},
  journal={The Mathematical Intelligencer},
  volume={12},
  number={4},
  pages={17--31},
  year={1990},
  publisher={Springer}
}

@InCollection{sep-qm-collapse,
	author       =	{Ghirardi, Giancarlo and Bassi, Angelo},
	title        =	{Collapse Theories},
	booktitle    =	{The Stanford Encyclopedia of Philosophy},
	editor       =	{Edward N. Zalta},
	howpublished =	{\url{https://plato.stanford.edu/archives/sum2020/entries/qm-collapse/}},
	year         =	{2020},
	edition      =	{Summer 2020},
	publisher    =	{Metaphysics Research Lab, Stanford University}
}

@InCollection{sep-operationalism,
	author       =	{Chang, Hasok},
	title        =	{Operationalism},
	booktitle    =	{The Stanford Encyclopedia of Philosophy},
	editor       =	{Edward N. Zalta},
	howpublished =	{\url{https://plato.stanford.edu/archives/win2019/entries/operationalism/}},
	year         =	{2019},
	edition      =	{Winter 2019},
	publisher    =	{Metaphysics Research Lab, Stanford University}
}

@book{brouwer1981,
  title={Brouwer's Cambridge lectures on intuitionism},
  author={Brouwer, L E J},
  editor={van Dalen, Dirk},
  year={1981},
  publisher={Cambridge University Press}
}

@article{zhang2017witnessing,
  title={Witnessing a Poincar{\'e} recurrence with Mathematica},
  author={Zhang, JM and Liu, Y},
  journal={Results in physics},
  volume={7},
  pages={3373--3379},
  year={2017},
  publisher={Elsevier}
}

@InCollection{sep-bohmian-mechanics,
	author       =	{Goldstein, Sheldon},
	title        =	{Bohmian Mechanics},
	booktitle    =	{The Stanford Encyclopedia of Philosophy},
	editor       =	{Edward N. Zalta},
	howpublished =	{\url{https://plato.stanford.edu/archives/sum2017/entries/qm-bohm/}},
	year         =	{2017},
	edition      =	{Summer 2017},
	publisher    =	{Metaphysics Research Lab, Stanford University}
}

@book{brouwer1907,
  title={Over de grondslagen der wiskunde},
  author={Brouwer, L E J},
  year={1907},
  publisher={Maas \& van Suchtelen}
}

@article{swart1992spreads,
  title={Spreads or choice sequences?},
  author={de Swart, H C M},
  journal={History and philosophy of logic},
  volume={13},
  number={2},
  pages={203--213},
  year={1991},
  publisher={Taylor \& Francis}
}

@article{gielen1981continuum,
  title={The continuum hypothesis in intuitionism},
  author={Gielen, W and de Swart, Harrie and Veldman, Wim},
  journal={The Journal of Symbolic Logic},
  volume={46},
  number={1},
  pages={121--136},
  year={1981},
  publisher={JSTOR}
}

@incollection{montague1974,
  title={Deterministic Theories},
  author={Montague, Richard},
  editor={Thomason, R H},
  booktitle={Formal Philosophy},
  publisher={Yale University Press},
  location={New Haven},
  year={1974}
}

@article{lorenz1963,
  title={Deterministic nonperiodic flow},
  author={Lorenz, Edward N},
  journal={Journal of the atmospheric sciences},
  volume={20},
  number={2},
  pages={130--141},
  year={1963}
}

@InCollection{sep-chaos,
	author       =	{Bishop, Robert},
	title        =	{Chaos},
	booktitle    =	{The Stanford Encyclopedia of Philosophy},
	editor       =	{Edward N. Zalta},
	howpublished =	{\url{https://plato.stanford.edu/archives/spr2017/entries/chaos/}},
	year         =	{2017},
	edition      =	{Spring 2017},
	publisher    =	{Metaphysics Research Lab, Stanford University}
}

@article{hilbert1926,
  title={{\"U}ber das Unendliche},
  author={Hilbert, David},
  journal={Mathematische Annalen},
  volume={95},
  number={1},
  pages={161--190},
  year={1926},
  publisher={Springer}
}

@book{goldblatt2012lectures,
  title={Lectures on the hyperreals: An introduction to nonstandard analysis},
  author={Goldblatt, Robert},
  volume={188},
  year={2012},
  publisher={Springer Science \& Business Media}
}

@article{arnowitt2008republication,
  title={Republication of: The dynamics of general relativity},
  author={Arnowitt, Richard and Deser, Stanley and Misner, Charles W},
  journal={General Relativity and Gravitation},
  volume={40},
  number={9},
  pages={1997--2027},
  year={2008},
  publisher={Springer}
}

@article{ziegler2006effectively,
  title={Effectively open real functions},
  author={Ziegler, Martin},
  journal={Journal of Complexity},
  volume={22},
  number={6},
  pages={827--849},
  year={2006},
  publisher={Elsevier}
}

@book{laplace1820theorie,
  title={Th{\'e}orie analytique des probabilit{\'e}s},
  author={Laplace, Pierre Simon},
  translator={Nagel, E},
  year={1820},
  publisher={Courcier}
}

@article{werndl2009deterministic,
  title={Are deterministic descriptions and indeterministic descriptions observationally equivalent?},
  author={Werndl, Charlotte},
  journal={Studies in history and philosophy of science part B: studies in history and philosophy of modern physics},
  volume={40},
  number={3},
  pages={232--242},
  year={2009},
  publisher={Elsevier}
}

@InCollection{sep-qt-issues,
	author       =	{Myrvold, Wayne},
	title        =	{Philosophical Issues in Quantum Theory},
	booktitle    =	{The Stanford Encyclopedia of Philosophy},
	editor       =	{Edward N. Zalta},
	howpublished =	{\url{https://plato.stanford.edu/archives/fall2018/entries/qt-issues/}},
	year         =	{2018},
	edition      =	{Fall 2018},
	publisher    =	{Metaphysics Research Lab, Stanford University}
}

@incollection{brukner2017quantum,
  title={On the quantum measurement problem},
  author={Brukner, {\v{C}}aslav},
  booktitle={Quantum [Un] Speakables II},
  pages={95--117},
  year={2017},
  publisher={Springer}
}

@InCollection{sep-intuitionism,
	author       =	{Iemhoff, Rosalie},
	title        =	{Intuitionism in the Philosophy of Mathematics},
	booktitle    =	{The Stanford Encyclopedia of Philosophy},
	editor       =	{Edward N. Zalta},
	howpublished =	{\url{https://plato.stanford.edu/archives/fall2020/entries/intuitionism/}},
	year         =	{2020},
	edition      =	{Fall 2020},
	publisher    =	{Metaphysics Research Lab, Stanford University}
}

@incollection{kreisel1967informal,
  title={Informal rigour and completeness proofs},
  author={Kreisel, Georg},
  booktitle={Studies in Logic and the Foundations of Mathematics},
  volume={47},
  pages={138--186},
  year={1967},
  publisher={Elsevier}
}

@article{born1926quantenmechanik,
  title={Zur Quantenmechanik der Sto{\ss}vorg{\"a}nge},
  author={Born, Max},
  journal={Zeitschrift f{\"u}r Physik},
  volume={37},
  pages={863--867},
  year={1926},
  publisher={Springer}
}

@article{landsman2020randomness,
  title={Randomness? What randomness?},
  author={Landsman, Klaas},
  journal={Foundations of Physics},
  volume={50},
  number={2},
  pages={61--104},
  year={2020},
  publisher={Springer}
}

@InCollection{sep-mathematics-constructive,
	author       =	{Bridges, Douglas and Palmgren, Erik},
	title        =	{Constructive Mathematics},
	booktitle    =	{The Stanford Encyclopedia of Philosophy},
	editor       =	{Edward N. Zalta},
	howpublished =	{\url{https://plato.stanford.edu/archives/sum2018/entries/mathematics-constructive/}},
	year         =	{2018},
	edition      =	{Summer 2018},
	publisher    =	{Metaphysics Research Lab, Stanford University}
}

@InCollection{sep-time-thermo,
	author       =	{Callender, Craig},
	title        =	{Thermodynamic Asymmetry in Time},
	booktitle    =	{The Stanford Encyclopedia of Philosophy},
	editor       =	{Edward N. Zalta},
	howpublished =	{\url{https://plato.stanford.edu/archives/win2016/entries/time-thermo/}},
	year         =	{2016},
	edition      =	{Winter 2016},
	publisher    =	{Metaphysics Research Lab, Stanford University}
}

@article{veldman2008borel,
  title={The Borel hierarchy theorem from Brouwer's intuitionistic perspective},
  author={Veldman, Wim},
  journal={The Journal of Symbolic Logic},
  volume={73},
  number={1},
  pages={1--64},
  year={2008},
  publisher={Cambridge University Press}
}

@book{dummett2000elements,
  title={Elements of intuitionism},
  author={Dummett, Michael},
  year={2000},
  publisher={Oxford University Press}
}

@incollection{veldman2020treading,
  title={Treading in Brouwer’s footsteps},
  author={Veldman, Wim},
  journal={Journal of Applied Logics},
  booktitle={Contemporary Logic and Computing},
  editor={Rezus, A},
  publisher={College Publications},
  year={2020},
  location={London},
  pages={355--396}
}

@article{driessenintuitionistic,
  title={Intuitionistic Probability Theory},
  author={Driessen, Bob},
  url={https://www.ru.nl/publish/pages/813276/driessen_bob_2018.pdf}
}

@book{brouwer1992intuitionismus,
  title={Intuitionismus},
  author={Brouwer, L E J},
  editor={van Dalen, D},
  location={Mannheim},
  publisher={Bibliographisches Institut, Wissenschaftsverlag}
}
\end{document}